\documentclass[preprint]{imsart}
\usepackage[english]{babel}
\RequirePackage[OT1]{fontenc}
\RequirePackage{amsthm,amsmath,amssymb,float,color,multirow,booktabs,hhline,xr,natbib}
\RequirePackage[colorlinks,citecolor=blue,urlcolor=blue]{hyperref}

\usepackage{epsfig}

\usepackage{graphicx}

\usepackage{float}
\usepackage{multicol}
\usepackage[framemethod=tikz]{mdframed}

\newtheorem{thm}{Theorem}
\newtheorem{lmma}{\large Lemma}
\newtheorem{propn}{Proposition}
\newtheorem{conjet}{\large Conjecture}
\newtheorem{problm}{Problem}
\newtheorem{define}{Definition}
\newtheorem{estimate}{\large Estimation Procedure}
\newtheorem{estR}{\large Estimation Procedure and Result}
\newtheorem{resl}{\large Result}
\newtheorem{ques}{\large Question}
\newtheorem{remrk}{Remark}
\newtheorem{ass}{Assumption}
\newtheorem{coroll}{Corollary}

\newenvironment{theorem}{\begin{thm}\hskip-6pt{\bf}\enspace
    \sl}{\end{thm}}

\newenvironment{lemma}{\begin{lmma}\hskip-6pt{\bf}\enspace \sl}{\end{lmma}}
\newenvironment{definition}{\begin{define}\hskip-6pt{\bf}\enspace
    \sl}{\end{define}}

\newenvironment{remark}{\begin{remrk}\hskip-6pt{\bf}\enspace
    \rm}{\end{remrk}}

\newtheorem{example}{Example}

\newcommand{\be}{\begin{eqnarray*}}
\newcommand{\ee}{\end{eqnarray*}}
\newcommand{\ben}{\begin{eqnarray}}
\newcommand{\een}{\end{eqnarray}}

\newcommand{\var}{{\rm Var}}
\newcommand{\cov}{{\rm Cov}}
\newcommand{\tr}{{\rm tr}}
\newcommand{\argmin}{\operatornamewithlimits{argmin}}

\newcommand{\grad}{\nabla}

\newcommand{\E}{\mathbb{E}}

\newcommand{\twonorm}[1]{\left\lVert#1\right\rVert_2}

\newcommand{\shtwonorm}[1]{\lVert#1\rVert_2}

\newcommand{\fnorm}[1]{\left\lVert#1\right\rVert_F}
\newcommand{\norm}[1]{\left\lVert#1\right\rVert}
\newcommand{\abs}[1]{\left\lvert#1\right\rvert}

\newcommand{\minus}{\ensuremath{-}}

\newcommand{\M}{{\mathcal{M}}}

\newcommand{\V}{\mathcal{V}}

\newcommand{\N}{{\mathcal N}}

\newcommand{\ora}{{\rm oracle}}

\newcommand{\half}{\ensuremath{\frac{1}{2}}}
\newcommand{\inv}[1]{\frac{1}{#1}}
\def\conv{\mathop{\text{\rm conv}\kern.2ex}}

\newcommand{\RE}{\textnormal{\textsf{RE}}}
\newcommand{\ip}[1]{\;\langle{\,#1\,}\rangle\;}

\newcommand{\maxnorm}[1]{\ensuremath{\left|#1\right|_{\max}}}
\newcommand{\expct}[1]{\ensuremath{\mathbb E}\left[#1\right]}
\newcommand{\silent}[1]{}
\newcommand{\mvec}[1]{\rm{vec}\left\{\,#1\,\right\}}

\newcommand{\ve}{\varepsilon}

\def\qed{\hskip1pt $\;\;\scriptstyle\Box$}
\def\Ber{\mathop{\text{Bernoulli}\kern.2ex}}
\def\supp{\mathop{\text{supp}\kern.2ex}}

\def\corr{\mathop{\text{corr}\kern.2ex}}
\def\prec{\mathop{\text{precision}\kern.2ex}}
\def\recall{\mathop{\text{recall}\kern.2ex}}
\def\mnorm{\mathcal{N}_{f,m}\kern.2ex}
\def\var{\mathop{\text{Var}\kern.2ex}}
\def\ess{\mathop{\text{ess}\kern.2ex}}
\def\dom{\mathop{\text{dom}\kern.2ex}}
\def\lin{\mathop{\text{lin}\kern.2ex}}

\newcommand{\func}[1]{\ensuremath{\mathrm{#1}}}
\newcommand{\diag}{\func{diag}}

\newcommand{\probto}{\mbox{$\stackrel{{\rm P}}{\longrightarrow}\,$}}

{\end{eqnarray}\end{subequations}\hskip-4.0pt}

\let\hat\widehat
\let\tilde\widetilde

\def\supp{\mathop{\text{\rm supp}\kern.2ex}}
\def\argmin{\mathop{\text{arg\,min}\kern.2ex}}

\newcommand{\prob}[1]{\ensuremath{\mathbb P}\left(#1\right)}

\newcommand{\beq}{\begin{equation}}
\newcommand{\eeq}{\end{equation}}
\newcommand{\bnum}{\begin{enumerate}}
\newcommand{\enum}{\end{enumerate}}
\newcommand{\bit}{\begin{itemize}}
\newcommand{\eit}{\end{itemize}}
\newcommand{\bens}{\begin{eqnarray*}}
\newcommand{\eens}{\end{eqnarray*}}

\newcommand{\R}{{\mathbb R}}

\newcommand{\B}{\ensuremath{\mathcal B}}

\newcommand{\SM}{\ensuremath{{\mathcal{S}^m}}}

\newcommand{\e}{\epsilon}

\newcommand{\vp}{\varphi}

\newcommand{\onem}{\textstyle \frac{1}{m}}
\newcommand{\onen}{\textstyle \frac{1}{n}}

\newcommand{\ignore}[1]{}{}

\begin{document}

\begin{frontmatter}
\title{Non-separable covariance models for spatio-temporal data, with applications to neural encoding
  analysis\thanksref{T1}}
\runtitle{Non-separable covariance models}
\thankstext{T1}{The research is supported by DMS-13-16731, and the Elizabeth Caroline Crosby Research Award to SZ from the University of Michigan.}

\begin{aug}
\author{\fnms{Seyoung} \snm{Park}$^*$\ead[label=e1]{seyoung.park@yale.edu}},
\author{\fnms{Kerby} \snm{Shedden}$^{\ddagger}$\ead[label=e2]{kshedden@umich.edu}}
\and
\author{\fnms{Shuheng} \snm{Zhou}$^{\dagger}$
\ead[label=e3]{shuhengz@umich.edu}}

\runauthor{Park, Shedden, Zhou}

\affiliation{Yale University\thanksmark{m1} and University of Michigan\thanksmark{m2}}

\address{ 
School of Public Health, Yale University$^*$\\
\printead{e1}}

\address{
University of Michigan$^{\ddagger \dagger}$\\
\printead{e2}\quad \printead{e3}}
\end{aug}

\begin{abstract}\quad
Neural encoding studies explore the relationships between 
measurements of neural activity and measurements of a 
behavior that is viewed as a response to that activity. 
The coupling between neural and behavioral measurements is typically imperfect and difficult to measure.
To enhance our ability to understand neural encoding relationships, 
we propose that a behavioral measurement may be decomposable 
as a sum of two latent components, such that the direct neural
influence and prediction is primarily localized to the component which
encodes temporal dependence. For this purpose,
we propose to use a non-separable Kronecker sum covariance model to
characterize the behavioral data as the sum of terms with exclusively
trial-wise, and exclusively temporal dependencies.  
We then utilize a corrected form of Lasso regression in combination with the nodewise regression approach for 
estimating the conditional independence relationships between and among
variables for each component of the behavioral data, 
where normality is necessarily assumed.
We provide the rate of convergence for estimating the precision
matrices associated with the temporal as well as spatial components in the Kronecker sum model. 
We illustrate our methods and theory using simulated data, 
and data from a neural encoding study of hawkmoth flight; 
we demonstrate that the neural encoding signal for hawkmoth wing
strokes is primarily  localized to a latent component with temporal
dependence, which is partially obscured by a second component with trial-wise dependencies.
\end{abstract}


\end{frontmatter}

\section{Introduction}

Statistics, machine learning and a broad range of application areas
such as genomics, neuroscience, and spatio-temporal modeling~\citep{Smith03,BCW08,WJS08,Yu09,KLLZ13}
that heavily rely on large-scale automated data analysis are sparking
one another's development in recent years.
In this setting, an important role for statistics is to provide models
that can accommodate the science and design computational efficient 
methods for dealing with large, complex and high dimensional
data arise from these application domains.

Neuroscience experiments often involve a large number of trials over
varying experimental conditions, often on only a modest number of
subjects.  In an experimental setting,~\cite{Sponberg:2015} are able to directly measure
neural activity and motion characteristics of hawkmoths, which are
known to be agile flyers whose flights are controlled by their complex neural systems.
Motor control is a complex process that involves the coordination of
many muscles to produce a complex movement.  
Animal flight is an especially intricate type of motion that is
produced by wing strokes that take place at high temporal frequency.
Studying neurocontrol of flight therefore involves a regression
relationship between measurements of neural activity and measurements of motion which are high-frequency, 
high-dimensional, and affected by substantial measurement error.

We first review the errors-in-variables (EIV) regression model with dependent measurements as studied in~\cite{RZ15}. 
Suppose that we observe $y\in \R^n$ and $X \in \R^{n \times m}$ in the following regression model 
\begin{eqnarray}
\label{eq::oby} 
y & =  & X_0 \beta^* + \ve, \quad \text{ where } \; \quad  X  =  X_0 + W,
\end{eqnarray}
where $\beta^* \in \R^m$ is an unknown vector to be estimated,
$X_0$ is an $n \times m$ design matrix and $W$ is a mean zero $n \times m$ random noise matrix, 
independent of $X_0$  and $\ve$, whose columns are also independent and each consists of dependent elements.
That is, we consider $\E \omega^j \otimes \omega^j = B$ for all $j=1, \ldots, m$, 
where $\omega^j$ denotes the $j^{th}$ column vector of $W$. Here $\otimes$ denotes the Kronecker Product.
Here we assume that the noise vector $\ve \in \R^n$ is independent of $W$ or
$X_0$, with independent entries $\ve_{j}$ satisfying $\E[\ve_{j}] = 0$ and
$\norm{\ve_{j}}_{\psi_2} \leq M_{\ve}$, 
where recall the  $\psi_2$ condition on a scalar random variable $V$
is equivalent to the subgaussian tail decay of $V$, which means
$\prob{|V| >t} \leq 2 \exp(-t^2/c^2), \; \; \text{for all} \; \; t>0$
for some constant $c$.

Using data from a study of \citet{Sponberg:2015},   
we show that the hawkmoth  flight motion data can be decomposed
additively into two components $X_0$ and $W$,
(cf. \eqref{eq::addmodel}),  
where $X_0$ contains most of the temporal correlation that is related to neural firing time difference $\Delta$, 
and $W$ is minimally related to neural activity and therefore can be
viewed as noise in this setting; moreover, our covariance model allows
dependencies among trials to be explicitly modeled, which is
significantly different from those measurement error models analyzed
in the literature. This novel additive model allows the
spatio-temporal features in observation 
data $X$ to be parsimoniously specified, 
in the sense that its covariance on $\mvec{X}$ can be succinctly
written as the  Kronecker sum as in \eqref{eq::addmodel}. 
For an $n \times m$ matrix $X$, $\mvec{X}$ is obtained by stacking the columns of
the matrix $X$ into a vector in $\R^{mn}$.

\subsection{Models and methods}
\label{sec::envregr}
We begin with a model considered in~\cite{RZ15}.  
Denote by $Z$ a subgaussian random ensemble where $Z_{ij}$ are
independent subgaussian random variables such that 
\ben
\label{eq::Zdef}
\expct{Z_{ij}} = 0 \; \text{ and } \; \; \norm{Z_{ij}}_{\psi_2} \leq K, \; \;  \mathrm{for\ all}\  i,j
\een 
for some finite constant $K$.
Let $Z_1, Z_2$ be independent copies of the subgaussian random matrix $Z$ as in~\eqref{eq::Zdef}.
Assume that the random matrix $X$ in~\eqref{eq::oby} satisfies
\ben
\label{eq::addmodel} 
X  & = & Z_1 A^{1/2} + B^{1/2} Z_2 
\; \sim \; \M_{n,m}(0, A \oplus B)\\
\nonumber
&& \text{ where } \; \;   A \oplus B := A \otimes I_n + I_m \otimes B
\een
denotes the Kronecker sum of positive definite $A \in \R^{m \times m}$
and $B \in \R^{n \times n}$. 
We use $X \sim \M_{n,m}(0,  A \oplus B)$ to denote the subgaussian random
matrix $X_{n \times m}$ which is generated using \eqref{eq::addmodel}
with $Z_1, Z_2$ being independent copies of $Z$ as in~\eqref{eq::Zdef}. 
When $Z_{i, j k} \sim \N(0, 1)$ for all $i =1, 2$ and all $j, k$, we
use $X \sim \N_{n,m}(0, A \oplus B)$, which is equivalent to say that
the  $n$ by $m$ random matrix $X$ follows $\mvec{X} \sim \N(0, A
\oplus B)$.
 In this covariance model, the first component, 
$A \otimes I_n$, describes the covariance of the {\it signal} $X_0 =
Z_1 A^{1/2}$, which is an ${n \times m}$ random design matrix with
independent subgaussian row vectors, and the other component, $I_m
\otimes B$, describes the covariance for the {\it noise matrix } $W
=B^{1/2} Z_2$, which contains independent subgaussian column vectors
$w^1, \ldots, w^m$, independent of $X_0$.  
This leads to a non-separable class of models for the observation $X$.

Suppose that $\hat\tr(B)$ is an estimator for $\tr(B)$; for example, 
as constructed in~\eqref{eq::trBest} for a single copy of data $X$ as in \eqref{eq::addmodel}. 
Let 
\ben
\label{eq::hatGamma}
\hat\Gamma & = & 
\inv{n} X^T X - \inv{n} \hat\tr(B)  I_{m}\;\;
\; \text{ and } \;
\hat\gamma \; = \; \inv{n} X^T y.
\een
For chosen parameters $\lambda$ and $b_1$, we exploit the following regularized estimation with the $\ell_1$-norm penalty to estimate $\beta^*$ (\citet{RZ15}):
\begin{eqnarray}
\label{eq::origin} \; \; 
\hat \beta & = & \argmin_{\beta: \norm{\beta}_1 \le b_1} \frac{1}{2} \beta^T \hat\Gamma \beta 
- \ip{\hat\gamma, \beta} + \lambda \|\beta\|_1. 
\end{eqnarray}
where $b_1$ is understood be chosen so that $\norm{\beta^*}_1 = \sum_{j=1}^m \abs{\beta^*_j} \le b_1$.
For replicated data, we can use replicates to obtain an estimator $\hat{B}$ for covariance $B$, 
using methods to be described in Subsection~\ref{sec:replicate}.
Then we can compute the trace of $\hat{B}$ and denote that by $\hat \tr(B) = \tr(\hat{B})$.
The non-convex optimization function in \eqref{eq::origin} can be solved efficiently using the composite gradient descent 
algorithm; see \citet{ANW12}, \citet{LW12}, and \citet{RZ15} for details.
An estimator similar to~\eqref{eq::origin} was considered by~\cite{LW12}, 
which is a variation of the  Lasso \citep{Tib96} or the Basis
Pursuit~\citep{Chen:Dono:Saun:1998} estimator. 
For the related Conic programming estimators (and related Dantzig
selector-type) for estimating $\beta^*$ in~\eqref{eq::oby}, see~\cite{RT10,RT13}, \cite{BRT14}, and \cite{RZ15}.

We are interested in studying the estimating the inverse covariance for $A$ and $B$, which we will
describe in Section~\ref{sec:graphical}.  
In Section \ref{sec:graphical}, the corrected Lasso estimator will
be used in the  nodewise regression procedure to estimate the concentration matrices 
 $\Theta = A^{-1}$ and $\Omega=B^{-1}$. We first focus on estimating the structures
by using the nodewise regression-based approach as in \citet{MB06}.
We then apply the refit procedure as in~\citet{Yuan10} and \citet{LW12} to obtain the final estimates of 
$\Theta$ and $\Omega$. 
In our numerical examples,  the estimated 
concentration matrices are subject to a suitable level of thresholding \citep{ZRXB11}.

\subsection{Data analysis}

The ability of flying animals to turn while in flight is controlled by
the firing of neurons located in the left and right dorsolongitudinal
muscles (DLMs).  Recently, \cite{Sponberg:2012} and
~\cite{Sponberg:2015} conducted experiments and developed an analytic
approach to better understand this neural-motor control mechanism.  
In these studies, torque profiles were obtained for individual wing strokes of
hawkmoths (a large bird-like moth).   
The torque profiles were sampled at high frequency, yielding around 500 measurements per wing stroke.
The spike times of neuronal firing for the left DLM ($t_L$) and right DLM ($t_R$) were also
measured for each wing stroke.

These spike times represent the time at which neuronal firing occurred immediately prior to each wing stroke. 
There is one pair of $(t_L, t_R)$ measurements per wing stroke and  298-928 wing strokes are recorded for each moth.
It is generally accepted that $t_L \approx t_R$ during straight flight
and $t_L < t_R$ or $t_L > t_R$ when the moth makes a left or a right turn respectively.  
The measured difference in DLM firing $\Delta = t_L - t_R$ 
may be taken as a summary measure of the neural signal for
turning.  The observed torque profile reflects
the actual turning behavior during one wing stroke.  While this profile can be summarized through a
simple measure such as the average torque, it is not clear what features of the torque profile are most tightly linked to the neural
signal represented by $\Delta$, and therefore may be presumed to be most directly under neural control.

\citet{Sponberg:2015} aim to capture the variation in movement that relates to the neural signals.
They apply the partial least squares (PLS) to extract the encoded features of movement based
on the cross-covariance of neural signals $(t_L, t_R)$ and torque profiles.
They exploit the extracted motor features to test whether neural signals act as a synergy or independently encode information about movement.

In the present work, we focus on the Kronecker sum model to encode 
the covariance structure for data matrix $X$ collected from each single moth. 
We allow dependencies between and among wing strokes (that is, repeated 
experiments under varying conditions) to be explicitly specified through the covariance 
matrix $B \succ 0$ in \eqref{eq::addmodel} while covariance matrix $A \succ 0$ 
is used to model the temporal dependencies among time points within each wing stroke, 
much like the classical time series analysis. 
The relative contributions of the ``signal" and ``noise" components vary by moth, 
perhaps due to inhomogeneities in the experimental conditions. 
Taking a purely data-driven approach, the neural encoding of motion can therefore be studied 
through the regression relationship between quantitative measures of neural activity in  $\Delta = t_L - t_R$ (as $y$)
and quantitative measures of motion which correspond to wing stroke data in $X_0$ using the model~\eqref{eq::oby}.

\subsection{Contributions}
In the proposed research, we aim to study a general class of
matrix decomposition and regression  problems,
where the design matrix $X_0$ and the random error $W$ may possess
dynamic and complex dependency structures.  
The presence of spatially correlated noise in $W$ across different flights motivates the consideration of 
errors-in-variables regression through the regression function \eqref{eq::oby} in combination with the Kronecker 
sum covariance model for $X$.  
Spatial dependencies in $X$ are understood to encode correlations between and among different trials of wing strokes, 
which are present possibly due to the correlated measurement errors and experimental conditions.
After accounting for this measurement error, the relationship between
neural activity and motion is shown to be stronger. 

The rest of this paper is organized as follows.
In Section~\ref{sec:graphical}, we describe our methods for estimating the concentration matrices.
In Section \ref{sec:theory}, we establish statistical convergence properties for  the inverse covariance estimation problem for the matrix-variate normal distribution with Kronecker sum covariance model.
In Section \ref{sec:simul}, we present simulation study which show
that our proposed methods indeed achieve consistent estimation of $\Theta$ and $\Omega$ in the operator norm. 
In Section \ref{sec:realdata}, 
we apply the Kronecker sum covariance model and its related errors-in-variables regression method to analyze
the hawkmoth neural encoding data. 
We also describe the estimated $\Theta$ and $\Omega$ which encode conditional independence relationships between
time points,  as well as among wing strokes. 
In Section~\ref{sec::conclude}, we conclude. \\

\noindent{\bf Notation.}
For a matrix $A = (a_{ij})_{1\le i,j\le m}$, we use $\twonorm{A}$ to denote its operator norm and
let $\norm{A}_{\max} =\max_{i,j} |a_{ij}|$ denote  the entry-wise max norm. 
Let $\norm{A}_{1} = \max_{j}\sum_{i=1}^m\abs{a_{ij}}$ denote the matrix $\ell_1$ norm.
The Frobenius norm is  $\norm{A}^2_F = \sum_i\sum_j a_{ij}^2$. 
For a square matrix $A$, let ${\rm tr}(A)$ be the trace of $A$, $\diag(A)$ be a diagonal matrix with the same diagonal as
$A$, and $\kappa(\Sigma)$ denote the condition number of $A$.
Let $I_n$ be the $n$ by $n$ identity matrix.  For two numbers $a, b$, $a \wedge b
:= \min(a, b)$ and $a \vee b := \max(a, b)$.
For a function $g: \R^m \to \R$, we write $\grad g$ to denote a
gradient or subgradient, if it exists. 
  Let $(a)_+ := a \vee 0$.
We use $a =O(b)$ or $b=\Omega(a)$ if $a \le Cb$ for some positive absolute
constant $C$ which is independent of $n, m$ or sparsity parameters.
We write $a \asymp b$ if $ca \le b \le Ca$ for some positive absolute
constants $c,C$.
The absolute constants $C, C_1, c, c_1, \ldots$ may change line by
line. 
We list a set of symbols we use throughout the paper in 
Table \ref{table_def} at the end of the paper.

\section{The additive Gaussian graphical models}
\label{sec:graphical}
Consider
\eqref{eq::addmodel} where  we assume $Z_1$ and $Z_2$ are independent  
copies of a Gaussian random ensemble $Z$, where $Z_{ij} \sim \N(0,1)$ for all $i, j$.  
The results in~\cite{RZ15} naturally lead to the following
considerations  for estimating the precision matrix $\Theta := A^{-1}$.
Similarly, we obtain $\hat{\Omega}$, the estimator for precision matrix $\Omega =
B^{-1}$. \\

\noindent {\bf Estimating $\Theta$ via Nodewise Regression}
To construct an estimator for $\Theta = A^{-1}$ with $X = X_0 + W$ as
defined in~\eqref{eq::addmodel}, we obtain $m$ vectors of
$\hat{\beta}^i, i=1, \ldots, m$,  
by solving
\eqref{eq::origin} with $\hat\Gamma$ and $\hat\gamma$ set to be 
\ben
\label{eq::gamma}
\hat\Gamma^{(i)} = \onen X_{\minus i}^T  X_{\minus i} - \hat\tau_B I_{m-1} 
\; \text{ and } \; \hat\gamma^{(i)} \; = \; \onen  X_{\minus i}^T X_i
\een
where $X_{\minus i}$ denotes columns of $X$ without $i$ and
$\hat\tau_B$ is an estimator for $\tau_B = \tr(B)/n$.
For a chosen penalization parameter $\lambda \geq 0$ and a fixed $b_1 > 0$, consider the following  variant of the Lasso estimator
\begin{eqnarray}
\label{eq::hatTheta}
\hat \beta^i & = & 
\argmin_{\beta \in \R^{m-1},  \norm{\beta}_1 \le b_1} \left\{ \frac{1}{2} \beta^T \hat\Gamma^{(i)} \beta 
- \ip{\hat\gamma^{(i)}, \beta} + \lambda \|\beta\|_1\right\}.
\end{eqnarray}
Covariance estimation can be obtained through procedures which involve
calculating variances for the residual errors after obtaining
regression coefficients or through the MLE refit procedure based on the associated edge set 
\citep[cf.][]{Yuan10,ZRXB11}. The choice of $b_1$ in~\eqref{eq::origin} will depend on the model class for $\Theta$,
which will be chosen to provide an upper bound on the matrix $\ell_1$ norm $\norm{\Theta}_1$.
More precisely, we can set $b_1 \asymp \max_{j} (\theta_{jj} \sum_{i\not=j}^m \abs{\theta_{ij}}) \le M
\norm{\Theta}_1$, assuming that $\theta_{jj} \le M$ for some absolute constant $M$.

To solve the non-convex optimization problem~\eqref{eq::hatTheta},   
we use the composite gradient descent algorithm as studied in \citet{ANW12} and \citet{LW12}. Let 
$L(\beta):=    \frac{1}{2} \beta^T \hat\Gamma^{(i)} \beta - \ip{\hat\gamma^{(i)}, \beta}$.
The gradient of the loss function is $\grad L(\beta) = \hat{\Gamma}^{(i)} \beta -\hat\gamma^{(i)}$. The composite gradient
descent algorithm produces a sequence of iterates $\{\beta^{(t)},\ t=0,1,2, \cdots, \}$ by
\begin{equation}
\label{grad_comp_2}
\beta^{(t+1)}= \argmin_{\|\beta\|_1 \le b_1} L(\beta^{(t)})+ <\grad L(\beta^{(t)}), \beta- \beta^{(t)}> + \frac{\eta}{2} \|\beta-\beta^{(t)}\|_2^2 + \lambda \|\beta\|_1 
\end{equation}
with the step size parameter $\eta>0$.

For the Kronecker sum model as in \eqref{eq::oby} and
\eqref{eq::addmodel} to be identifiable,  
we assume the trace of $A$ is known.
\bnum
\item[(A1)]
We assume $\tr(A) = m$ is a known parameter. 
\enum
For example, any $m$-dimensional correlation matrix satisfies (A1).
We defer the discussion on condition (A1) in
Section~\ref{sec::related}. For a model where replicated measurements
of $X_0$ are available, we do not need to assume that (A1) holds.
Assuming $\tr(A)$ or $\tr(B)$ is known is unavoidable as the
covariance model is not identifiable otherwise for a single sample case.
By knowing $\tr(A)$, we can construct an estimator for
$\tr(B)$ following~\cite{RZ15},
\ben
\label{eq::trBest}
\hat\tr(B) &:= &
\onem \big(\fnorm{X}^2 -n \tr(A)\big)_{+} \; \; \;
\text{and } \; \; 
\hat\tau_B  := \onen \hat\tr(B).
\een
Let $Z_1 =(Z_{1, ij})_{n \times m}$ be a Gaussian random ensemble with independent standard normal entries.
To estimate the precision matrix $\Theta$,  first consider the following regressions, where we regress one variable against
all others: for $X_0 = Z_1 A^{1/2}$ defined in \eqref{eq::addmodel},
the $j$th column of $X_0$ satisfies  
\begin{eqnarray}
\label{eq::regr}
&& X_{0,j}   =  
X_{0, -j}\beta^{j} + V_{0, j}, \\
&&\ \text{where}\;\; V_{0,j} \sim {\N}(0_n, \sigma_{V_j}^2 I_n)\
 \mbox{is independent of}\ X_{0, -j},  \label{eqLreggg2}
\end{eqnarray}
where $X_{0,j}$ denotes the $j$th column of $X_0$, $X_{0,-j}$ denotes the matrix of $X_0$ with its $j$th column removed,
$\sigma_{V_j}^2 := A_{jj}-A_{j,-j} A_{-j,-j}^{-1} A_{-j,j}$, 
and $\beta^{j} \in \R^{m-1}$ corresponds to a vector of regression coefficients. 
One can verify that 
\begin{eqnarray}
\label{eq::cond-beta}
 \Theta_{j j}  & = & (A_{jj} -A_{j, \minus j} \beta^{j})^{-1}\; \text{
   and }\;  \Theta_{j, \minus j} \;  = \; -(A_{jj} -A_{j,   \minus j} \beta^{j})^{-1} \beta^{j}.
\end{eqnarray}  
Now the $j^{th}$ column of $X$ in \eqref{eq::oby} can be written as
\bens
\label{eq::additive}
&&X_j  =  X_{0,-j} \beta^{j} + V_{0,j} + W_j  
=: X_{0,-j} \beta^{j}  + \ve_j,  \quad j=1, \ldots, m, \\
&& \text{where we  observe } X_{-j} = X_{0, -j} + W_{-j}, \; \; \text{which resembles \eqref{eq::oby}.}
\eens
The noise $\ve_j =   V_{0,j} + W_j$ is independent of $\{X_{i}; i \neq j\}\ (i=1,\ldots,m)$, but the components of $\ve_j$ are
correlated due to $W_j$.  Thus this fits in the errors-in-variables
framework  despite the complication due to the dependence within
components of $\ve_j$ for all $j$. 
We estimate $\beta^i$ for each $i$ with \eqref{eq::hatTheta} and 
then obtain $\Theta$ by using \eqref{eq::cond-beta}. 
We summarize the algorithm as follows:\\

\noindent
\textbf{Algorithm 1: Input ($\hat\Gamma^{(j)}, \hat\gamma^{(j)}$) as
  in~\eqref{eq::gamma} and  $\hat\Gamma$ as in~\eqref{eq::hatGamma}} \\
\noindent{\textbf{(1)}}
Perform $m$ regressions using \eqref{eq::hatTheta} to obtain vectors of
$\hat{\beta}^{j} \in \R^{m-1}$, $j=1, \ldots, m$, with  
penalization parameters $\lambda^{(j)}, b_1 > 0$ to be specified. \\
\noindent{\textbf{(2)}}
Construct $\tilde{\Theta} \in \R^{m \times m}$  in view of \eqref{eq::cond-beta}, 
such that for $j=1, \ldots, m$,
\bens
\tilde{\Theta}_{j, \minus j} = -(\hat\Gamma_{jj} - \hat\Gamma_{j,
    \minus j} \hat{\beta}^j)^{-1}\hat{\beta}^j, \; \text{and}  \; \tilde{\Theta}_{j j} = (\hat\Gamma_{jj} - \hat\Gamma_{j,
   \minus j} \hat{\beta}^{j})^{-1}.
\eens
\noindent{\textbf{(3)}}
Project $\tilde{\Theta}$ onto the space $\SM$ of $m \times m$ symmetric matrices: 
\[
\hat\Theta = \argmin_{\Theta \in \SM} \norm{\Sigma - \tilde{\Theta}}_1.
\]

Our theoretical results show that this procedure indeed achieves consistent estimation of
$\Theta$ in the operator norm under suitable
assumptions.

\subsection{Related work}
\label{sec::related}
For matrix-variate data with two-way dependencies,   
prior work depended on a large number of replicated data to obtain certain convergence guarantees, 
even when the data is observed in full and free of measurement error;
see for example~\cite{Dut99}, \cite{WJS08}, \cite{LT12}, and \cite{THZ13}.
A recent line of work on matrix variate models \citep{KLLZ13,Zhou14a,RZ15}
have focused on the design of estimators and efficient algorithms
while establishing theoretical properties by using the Kronecker sum and product covariance models when a single or a small
number of replicates are available from such matrix-variate
distributions; See also~\cite{Efr09}, \cite{AT10}, and \cite{HSZ15} for related models and applications.
Among these models, the Kronecker sum provides a
covariance or precision matrix which is sparser than the
Kronecker product (inverse) covariance model.

Variants of the linear errors-in-variables models in the high
dimensional setting has been considered in recent work~\citep{RT10,LW12,RT13,BRT14,CC13,SFT14}, where oblivion in
the covariance structure for row or columns of $W$, and a general dependency condition in the single data matrix $X$ are
not simultaneously allowed.  
The second key difference between our framework and
the existing work is that we assume that only one observation matrix $X$ with the single measurement
error matrix $W$ is available. Assuming (A1) allows us to estimate $\E W^T W$ as required
in the estimation procedure \eqref{eq::hatGamma} directly, given the
knowledge that $W$ is composed of independent column vectors.  
In contrast, existing work needs to assume that the covariance matrix 
$\Sigma_W := \onen \E W^T W$ of the independent row vectors
of $W$ or its functionals are either known a priori, or can be 
estimated from a dataset independent of $X$, or from replicated $X$
measuring the same $X_0$; see for example \cite{carr:rupp:2006}, \cite{RT10},  \cite{LW12}, \cite{RT13}, and \cite{BRT14}.
Although the model we consider is different from those in the
literature, the identifiability issue, which arises from the fact that
we observe the data under an additive error model, is common.
Such repeated measurements are not always available or costly to 
obtain in practice~\citep{carr:rupp:2006}. We will explore such
tradeoffs in future work.

\section{Theoretical properties}
\label{sec:theory}
We require the following assumptions.
\bnum
\item[(A2)]
The minimal eigenvalue $\lambda_{\min}(A)$
of the covariance matrix $A$ is bounded: $1 \ge \lambda_{\min}(A) > 0$.
\item[(A3)]
The condition number $\kappa(A)$ is upper bounded
by $O\left(\sqrt{\frac{n}{\log m}}\right)$ and $\tau_B = O(\lambda_{\max}(A))$.
\item[(A4)]
The covariance matrix $B$ satisfies
\[
\frac{\fnorm{B}^2}{\twonorm{B}^2} =\Omega(\log m),\quad
\frac{\tr(B)}{\twonorm{B}}  = \Omega\left( \frac{n}{\log m}
\log \frac{m \log m }{n}\right).
\]
\enum
Theorem \ref{coro::Theta} shows the statistical consistency of $\hat\Theta$ in the operator norm.
See Table \ref{table_def} for the notations used in the theorem.
Proof of the theorem can be found in Section \ref{Sec:gau}. 

\begin{theorem}
\label{coro::Theta}
Suppose conditions (A1)-(A4).   
Suppose the columns of $\Theta$ is $d$-sparse, i.e., the number of nonzero entries on each column in $\Theta$ is bounded by $d$, which satisfies
$d=O(n/ \log m)$.
Suppose the
condition number $\kappa(\Theta)$ is finite.
Let $\hat{\beta}^i$ be an optimal solution to the nodewise regression with $b_1 := b_0 \sqrt{d}$, where $b_0$ satisfies $\phi b_0^2 \le \twonorm{\beta^{i}}^2 \le b_0^2$ for 
some $0< \phi <1$, and $D'_0 = {\twonorm{B}}^{1/2} + a_{\max}^{1/2}$ and
\begin{equation}
\label{eq:lambda}
\lambda^{(i)}  \ge  \psi_i\sqrt{\frac{\log m}{n}} 
  \; \; \text{ where } \;\;     \psi_i := C_0 D_0' K^2  \left(\tau_B^{+/2} \twonorm{\beta^{i}} + \sigma_{V_i}  \right)
\end{equation}
for some positive absolute constant $C_0$, where $\sigma_{V_i}^2 :=
A_{ii}-A_{i,-i} A_{-i,-i}^{-1} A_{-i,i}$.   
Then, 
\begin{equation}
\label{eq:bound:thm1}
\twonorm{\hat\Theta - \Theta} 
=O_P\left( \frac{d K^2(A)}{\lambda_{\min}^2(A)}  \max_i \lambda^{(i)} \right).
\end{equation}
Moreover, suppose that $\tau_B \le \twonorm{B} = O(\lambda_{\max}(A))$ and
$ \lambda^{(i)}  \asymp \psi_i\sqrt{\frac{\log m}{n}}$ for each $i$. Then
\begin{equation}
\label{eq:rel:bd}
\frac{\twonorm{\hat{\Theta}-\Theta}}{\|\Theta\|_2} = O_P\left(\frac{K^4(A)}{\lambda_{\min}(A)}\sqrt{\frac{d \log
      m}{n}}\right).
\end{equation}
\end{theorem}
Using the similar argument,  we can show consistency results of
$\Omega = B^{-1}$.   See Theorem \ref{coro::Omega} of the Supplementary materials for
details.
\begin{remark}
\label{rem::factor}
In \eqref{eq:lambda},  since $\beta^{i}= -(A_{ii}-A_{i,-i} A_{-i,-i}^{-1} A_{-i,i}) \Theta_{-i,i}$, 
\bens
 \|\beta^{i}\|_2 \le a_{\max}  \|\Theta_{-i,i}\|_2 \le a_{\max}
 \lambda_{\max} (\Theta) =  
\frac{a_{\max}}{\lambda_{\min}(A)},
\eens
where $a_{\max} := \max_{i} A_{ii} \ge 1$ by (A1). 
Together with $\sigma_{V_i}^2 \le a_{\max}$ and  $\tau_B^{+/2} = O(\sqrt{\tau_B})$, it holds that
\bens
\max_i \psi_i & \asymp & 
\left({\twonorm{B}}^{1/2} + a_{\max}^{1/2}\right) \left(\sqrt{\tau_B}
  \frac{a_{\max}}{\lambda_{\min}(A)} + a_{\max}^{1/2}\right) \\
& \le & 
 \kappa(A) \left(2\twonorm{B} + \tau_B + 2 a_{\max}\right) \\
\text{ thus }
\inv{\lambda_{\min}^2(A)}\max_i \psi_i 
& \le & \frac{\kappa(A)}{\lambda_{\min}(A)} 
\left(\frac{3\twonorm{B}}{\lambda_{\min}(A)} + 2\kappa(A) \right)
\asymp \frac{ \kappa^2(A)}{\lambda_{\min}(A)}
\eens
so long as $\twonorm{B} = O(\lambda_{\max}(A))$.
This calculation shows that  so long as 
$\tau_B \le \twonorm{B} = O(\lambda_{\max}(A))$ and
$ \lambda^{(i)}  \asymp \psi_i\sqrt{\frac{\log
    m}{n}}$, \eqref{eq:rel:bd} holds.
When the noise in  $W$ is negligible in the sense that 
$\tau_B$ (and hence $\|B\|_2$ in view of (A4)) is close to zero, 
then \eqref{eq:bound:thm1} reduces to
\bens
\twonorm{\hat\Theta - \Theta} = O_P\left(K^4(A) d \sqrt{\frac{\log m}{n}}\right).
\eens
The bound \eqref{eq:bound:thm1} is analogous to that of Corollary 5 in
\citet{LW12}, where the measurement error $W$ has independent row
vectors with a known $\Sigma_W := \onen \E W^T W$.  
\end{remark}

\begin{remark}
Our theoretical results show that $\hat{\Theta}$ consistently estimates $\Theta$ in the
spectral norm under suitable conditions. 
However, $\hat{\Theta}$  is not necessarily positive-semidefinite.
One can obtain a consistent and positive-semidefinite estimator  
by considering an additional estimation procedure, which is described in Section \ref{sec:add:proc}. 
\end{remark}

\subsection{Using replicates}
\label{sec:replicate}
In the current work, we focus a class of generative models which rely
on the sum of Kronecker product covariance matrices to model complex trial-wise
dependencies as well as to provide a general statistical framework for
dealing with signal and noise decomposition. 
It can be challenging to handle this type of data because 
there are often inhomogeneities and dependencies among the trials, and
thus they can not be treated as independent replicates.
In the present work, we propose a framework
for explicitly modeling the variation in a set of ``replicates'' that
may be neither independent nor identical.
We denote the number of subjects by $n$,  the number of time
points by $m$, and the number of (non-i.i.d.) replicates by $N$.  
Let $A \in \R^{m \times m}$, $B, C  \in \R^{n\times n}$ be $m$ by $m$ and $n$ by $n$ positive definite 
matrices, respectively. We consider the following generative model:
\begin{eqnarray}
\label{main:eq:1}
&& X_i= X_0+W_i, \quad \forall i=1,\cdots, N 
\; \; \text{where }\\
\nonumber
&& \mathrm{vec}(X_0) \sim \mathcal{L}(0, A \otimes C) \quad \text{ and
}  
\quad \mathrm{vec}(W_i) \sim \mathcal{L}(0, I_m \otimes B),
\end{eqnarray}
where $X_0, W_i \in \R^{n \times m}$ are independent subgaussian random
matrices such that $\mvec{X_0}, \mvec{W_i} \in \R^{mn}$ 
have covariances $A \otimes C$ and $I_m \otimes B$, respectively,  
where recall $\mvec{X_0}$ and $\mvec{W_i}$ are obtained by stacking
the columns of $X_0$ and $W_i$ into vectors in $\R^{mn}$.
Here the mean response matrix $X_0$ and the experiment-specific variation matrices $W_i$
can jointly encode the temporal and spatial dependencies. 
More generally, we can use a subgaussian random matrix
to model replicate-to-replicate fluctuations:
\bens
\mvec{W_i} \sim \mathcal{L}(0, I_m \otimes B(t)),\; \; \text{ where } \; \; B(t) \succ 0
\eens
is the covariance matrix describing the spatial dependencies which may
vary cross the $N$ replicated experiments.
In this subsection, we consider the general model \eqref{main:eq:1},
but with 
\begin{equation}
\label{eq:main}
\mathrm{vec}(X_0) \sim \N(0, A \otimes I_n) \; \text{ and } \; 
 \mathrm{vec}(W_i) \sim \N(0, I_m \otimes B),
\end{equation}
where $X_0$ and $W_i$ are independent each other.
Note that in model \eqref{eq:main},  one can avoid the assumption that
the trace of $A$ is known.
To estimate $\Omega =B^{-1}$, we note that
\[
\forall i \neq j,\quad
\mathrm{vec}(X_i-X_j) \sim \N(0, I_m \otimes 2B).
\]
Without loss of generality, assume $N$ is an even number. We have $N/2$ replicates to estimate $\Omega= B^{-1}$: each column of $\widetilde{W}_i=X_{2i-1}-X_{2i}$ for $i=1,\cdots, N/2$ is a random sample of $\N(0_n, 2B)$.    
To estimate $\Theta=A^{-1}$, consider the following observed mean response:
\ben
\label{eq::Beffect}
\overline{X} =\frac{1}{N} \sum_{i=1}^N X_i, \quad \rm{where}\quad
\cov(\mvec{\overline{X}}) =  A \otimes I_n + \frac{1}{N} I_m \otimes B.
\een
We omit theoretical properties of these estimators. We summarize the estimation procedures as follows: \\

\noindent
\textbf{Algorithm 2-1: Obtain $\hat{\Omega}$ with i.i.d. input vectors $\{\widetilde{W}_i\}_{i=1,\cdots, N/2}$}\\
\noindent
\noindent{\textbf{(1)}}
Let $\widetilde{B} := \sum_{i=1}^{N/2} \widetilde{W}_i
\widetilde{W}_i^T /(Nm) \succeq 0$ be an unbiased estimator of $B$.\\
\noindent{\textbf{(2)}}
Apply graphical Lasso \citep{FHT07} with the input $\widetilde{B} \succeq 0$ to obtain $\hat{\Omega}$.\\
\noindent
\textbf{Algorithm 2-2: Obtain $\hat{\Theta}$ with input $\hat{\tau}_B=
  \tr(\widetilde{B})/n$}.\\  
\noindent
Let $\widetilde{A}$  be an unbiased estimator of $A$;
Estimate $\Theta$ using  Algorithm 1 with $$\hat \Gamma = \widetilde{A}
:= \overline{X}^T \overline{X}/n-\frac{\hat{\tau}_B}{N} I_m,\quad
\hat\Gamma^{(j)} =\hat \Gamma_{-j, -j} 
\quad \text{ and } \; \hat\gamma^{(j)} \; = \; \onen  \overline{X}_{\minus j}^T \overline{X}_j.
$$

\section{Simulation results}
\label{sec:simul} 
  
We perform simulations to investigate the performances of the estimators.
Consider the following covariance models for $A$ and $B$ and their inverses as used in
\citet{Zhou14a}.    

\begin{enumerate}
\item AR(1) model :  For $\rho \in (0,1)$, the covariance matrix $A$ is of the form $A=(a_{ij})$ such that $a_{ij}=\rho^{|i-j|}$.

\item Star-Block model:  The covariance matrix $A$ consists of sub-blocks
with size $16$ whose inverses correspond to star graphs, where $A_{ii}=1$:
In each subgraph, five nodes are connected to a central hub node with no other connections.
Covariance matrix for each block $S$ in $A$ is generated as follows: $S_{ij} = 0.5$ if $(i, j) \in E$ and
$S_{ij} = 0.25$, otherwise.  

\item Random graph: 
The graph is generated according to a Erdos-Renyi random graph model. 
Initially, we set $\Omega= I_{n \times n}$. Then we randomly select $n$ edges and update $\Omega$ as follows: for each new edge $(i,j)$, a weight $w>0$ is chosen uniformly at random from $[0.1,0.3]$; 
we subtract $w$ from $\omega_{ij}$ and $\omega_{ji}$, and increase $\omega_{ii}$ and $\omega_{jj}$ by $w$. 
We multiply the constant $c>0$ to the $\Omega$ such that
$\rm{tr}(c^{-1}\Omega^{-1})=n$ holds,  and set $B := c^{-1} \Omega^{-1}$.
This sets the trace of $B$ to be $n$.
\end{enumerate}

Throughout the simulation, we use $\lambda^{(i)} = 2 \psi_0 \sqrt{\log
  m/n}$, $b_1=\sqrt{d} a_{\max}$ and step size parameter $\eta = 1.5
\lambda_{\max}(A)$ for the composite gradient descent algorithm \eqref{grad_comp_2}.

\subsection{Statistical and optimization error}
In this subsection, we study the optimization error $\log
(\|\tilde{\beta}^i_t - \hat{\beta}^i\|_2)$ and the statistical error
$\log (\|\tilde{\beta}^i_t - \beta^i\|_2)$, based on nodewise
regression estimator $\tilde{\beta}^i_t$, which is the $t^{th}$
iteration in the composite gradient descent algorithm for the $i^{th}$ nodewise regression.

Figure \ref{fig_opt} displays the error when $m=512$ and $n=256$. 
The matrix $A$ is generated using the Star-Block model (left), 
and the Erdos-Renyi random graph model with $\|A\|_2=3.0$ and $\kappa(A)=9.0$ (right)
respectively, while setting $B=0.3B^*$, where $B^*$ is generated using
the Erdos-Renyi random graph model, where $\|B\|_2=0.54$ and $\kappa(B)=8.40$.
For the Star-Block model, we randomly select twenty hubs and twenty
leaves linked to hubs are considered in the graph.    
For the random graph model, we randomly select twenty nodes in the graph.

\begin{figure}[t]
\centering
  \begin{tabular}{@{}cc@{}}
        \includegraphics[width=5.6cm, height=5.6cm]{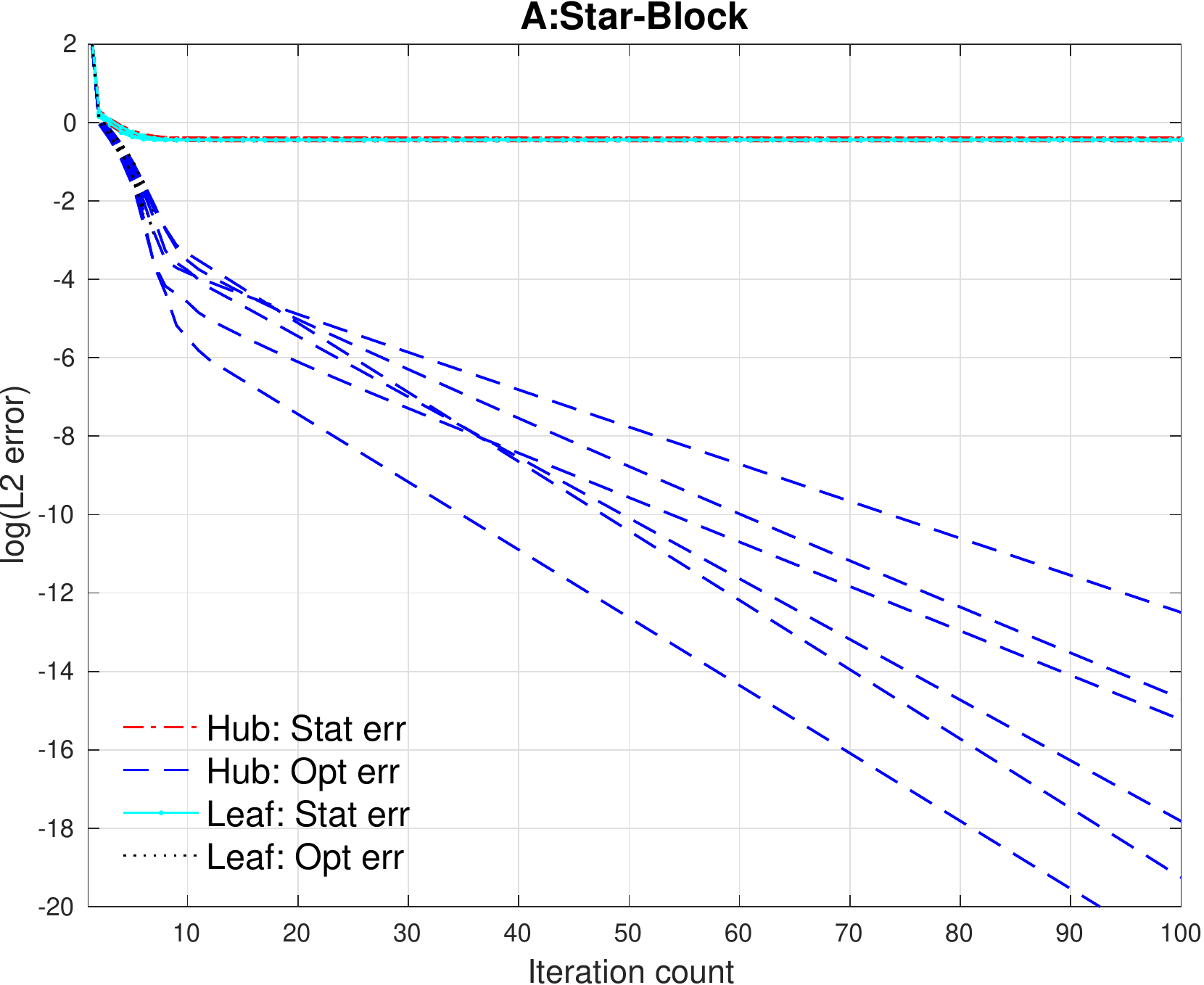}  
        & \includegraphics[width=5.6cm, height=5.6cm]{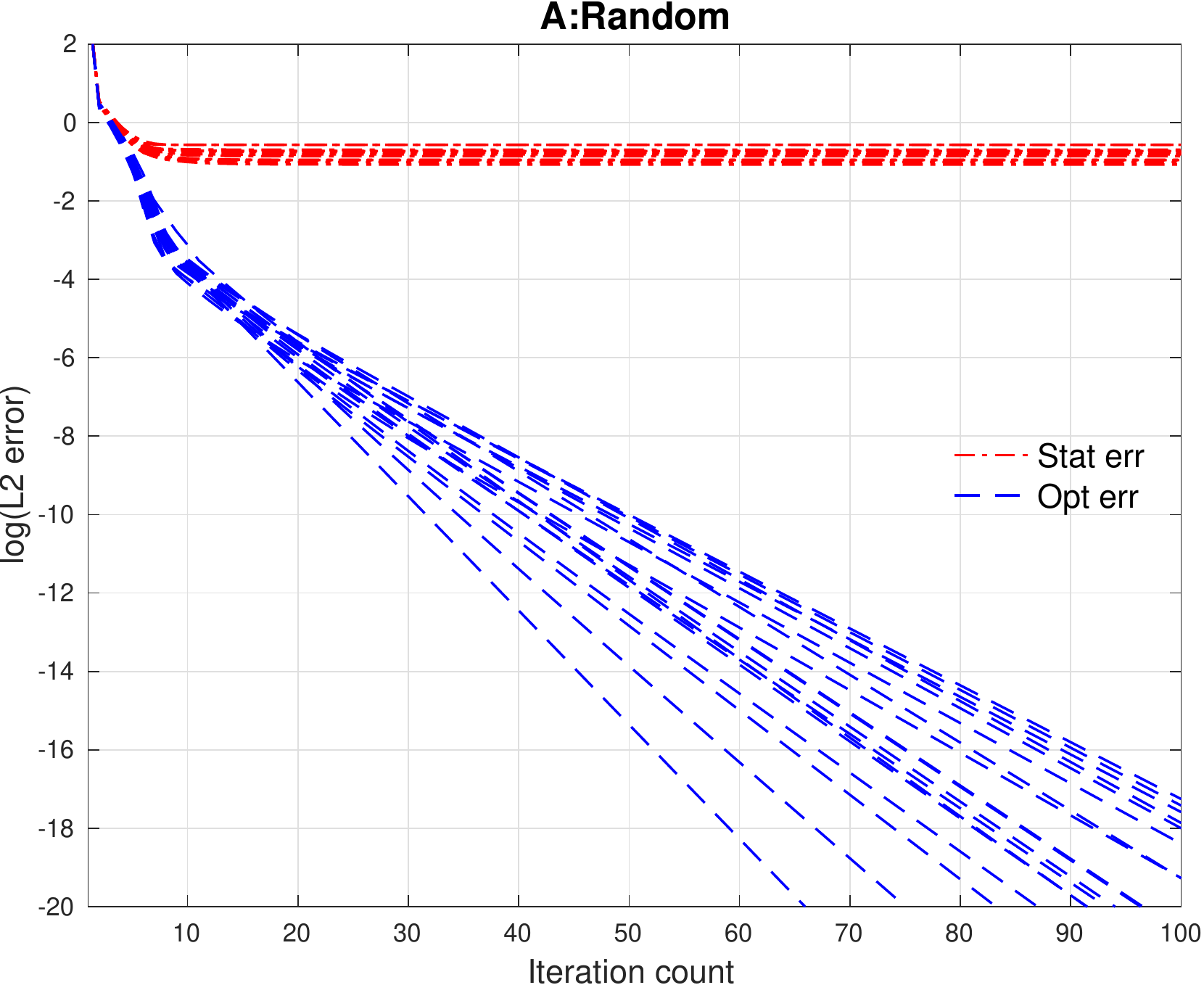} 
  \end{tabular}
  \caption{
Figures display the average of the optimization error $\log
(\|\tilde{\beta}^i_t - \hat{\beta}^i\|_2)$ 
and statistical error $\log (\|\tilde{\beta}^i_t - \beta^i\|_2)$ for $i=1,\cdots, m$ over ten
different initial inputs, where $\tilde{\beta}^i_t$ is the $t^{th}$
iterate of the composite gradient descent algorithm on the $i^{th}$ node
regression. Each line corresponds to one nodewise regression.
}
\label{fig_opt}
\end{figure}

For each chosen node, the error is averaged over ten trials with
different initialization points in the algorithm.
We observe that the sequence of iterates $\{\tilde{\beta}^{i}_t\}$ to
converge geometrically to a fixed point for each node $i$, while the
statistical error flattens out.

\subsection{Analysis on statistical error}

In the rest of the simulations, we repeat experiments $100$ times and
record the average of the relative error for estimating $\Theta$ using
$\hat{\Theta}$ in the operator and the Frobenius norm:
$\|\hat{\Theta}-\Theta\|_2/ \|\Theta\|_2$ and $\|\hat{\Theta}-\Theta\|_F/ \|\Theta\|_F$.
Figure \ref{figee143BB_changingasd} displays the error against
a rescaled sample size  when $m \in \{128,256,512\}$ for each case.

For the top two plots of Figure \ref{figee143BB_changingasd},  
we consider model \eqref{eq:main} when the number of replicates are $N \in \{1,2,4,8\}$. 
Note that when $N=1$, the model reduces to the Kronecker sum model.
The covariance matrices $A$ and $B$ are generated using the $AR(1)$
and the random model respectively, with $\rho_A=0.3$ and $\tau_B=0.5$.
We observe the error $\norm{\hat{\Theta} - \Theta}$ in the operator
and the Frobenius norm decreases as $N$ increases 
because $N$ reduces the trace effect of $B$ from $\tau_B$ to
$\tau_B/N$ as shown in \eqref{eq::Beffect}.

For the lower two figures of Figure \ref{figee143BB_changingasd}, we consider the Kronecker sum model, where
$A$ and $B$ are generated using the Star-Block and AR(1) model, respectively, with $\rho_B \in \{0.3, 0.7\}$ and 
$\tau_B \in \{0.3, 0.7\}$.
We observe that the two sets of curves corresponding to $\tau_B=0.3$ are lower than those of $\tau_B=0.7$ while
the curves with the lower $\rho_B$ have smaller errors given the same $\tau_B$ value.
The hidden factor imposed on the upper bound of the relative error
$\|\hat{\Theta}-\Theta\|_2/ \|\Theta\|_2$ in the operator norm as stated in
Theorem~\ref{coro::Theta} remains invariant for three values of $m \in \{128, 256, 512\}$
when $\rho_B$ and $\tau_B$ are fixed. See Remark
\ref{rem::factor} for discussions.   


We observe that the two sets of curves corresponding to $\tau_B=0.3$
are lower than those of $\tau_B=0.7$, which can be explained by the
fact that the factor has a smaller value when $\tau_B=0.3$. 
These results are consistent with the theoretical bounds in Theorem \ref{coro::Theta}.
Overall, we see the curves now align for different values of $m$ in
the rescaled plots, which confirms the error bound of $O(d\sqrt{\log m/n})$.

\begin{figure}[H]
\centering
  \begin{tabular}{@{}cc@{}}
         \includegraphics[width=6cm, height=5cm] {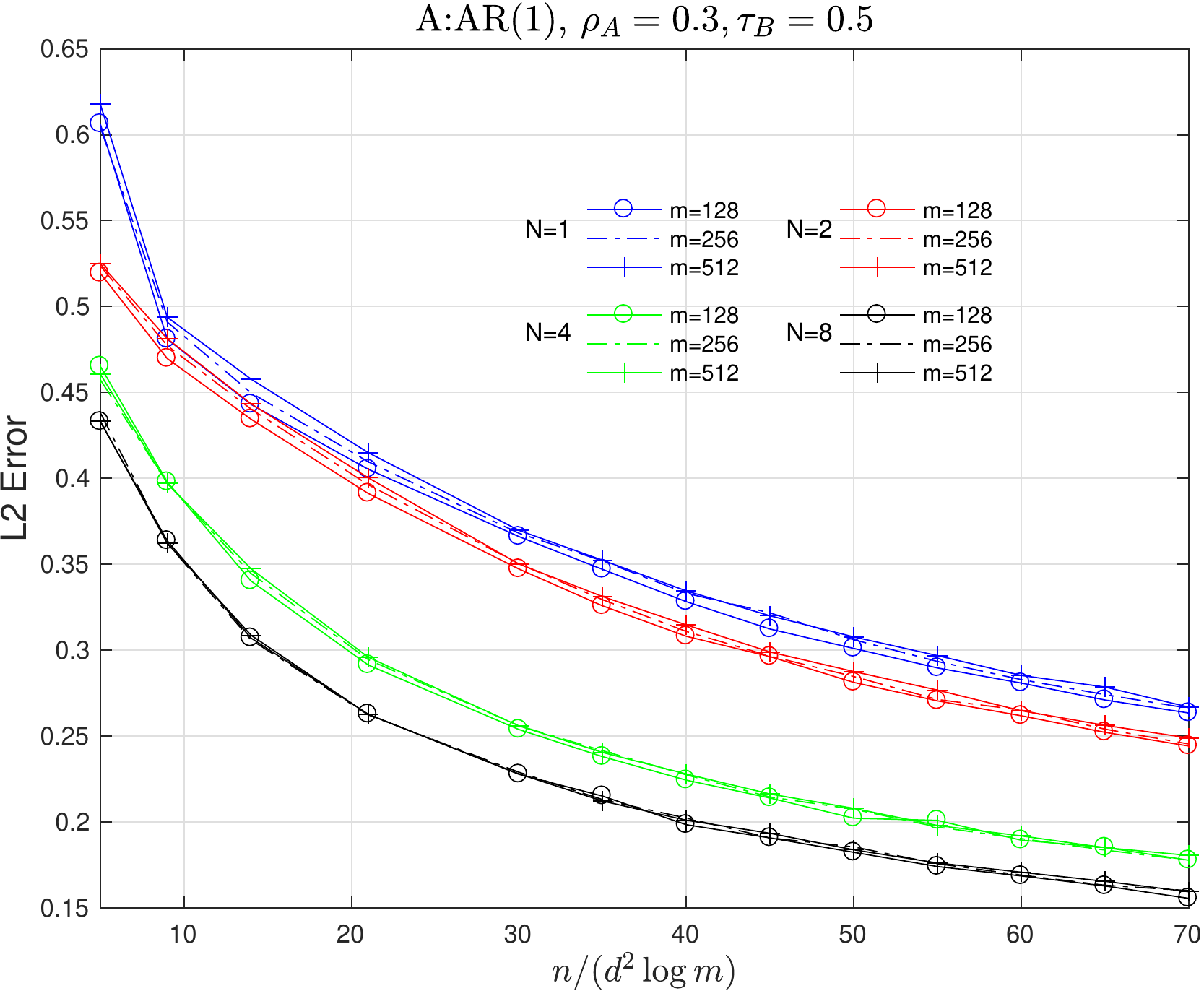}  &   \includegraphics[width=6cm, height=5cm] {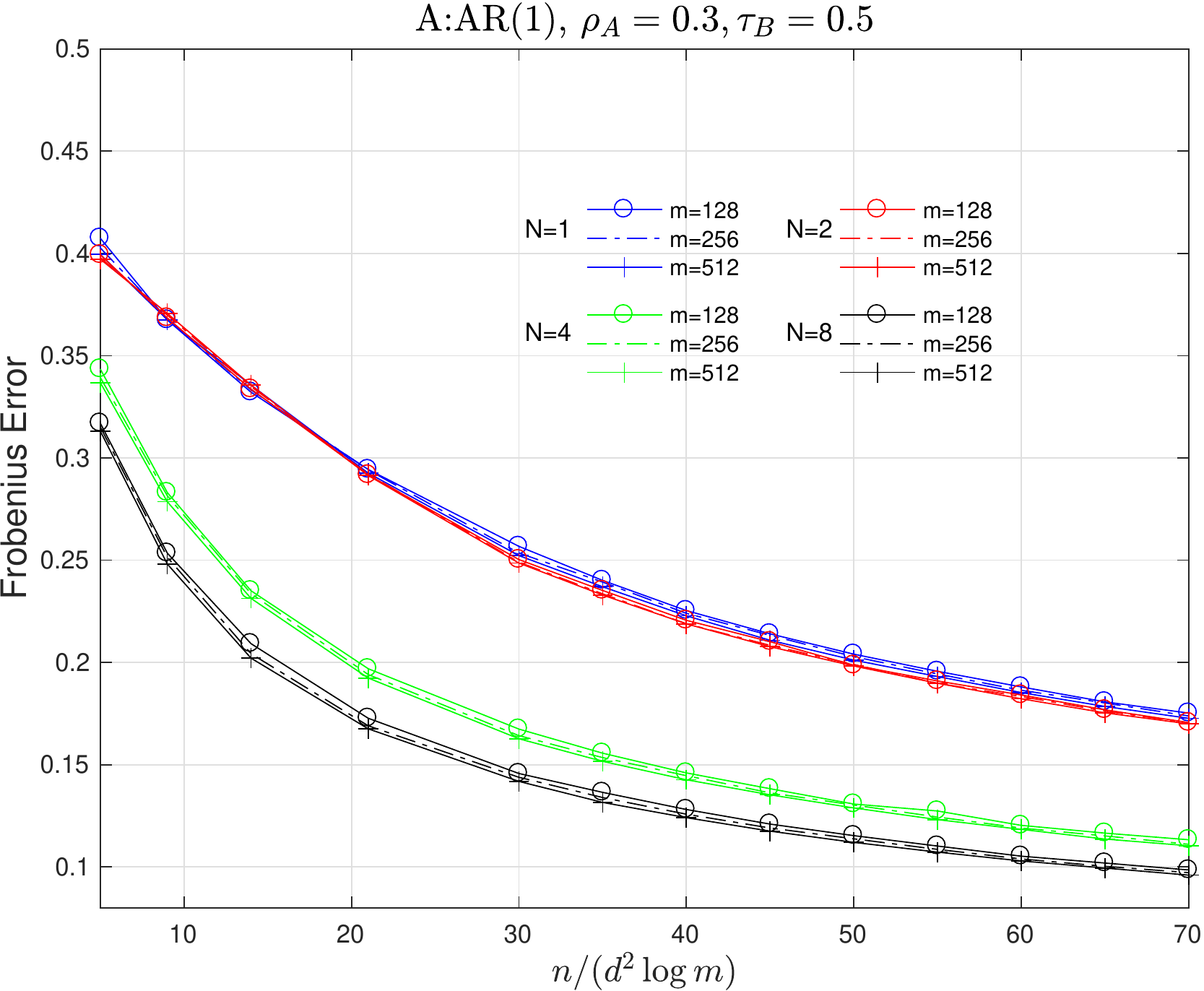} \\
                                    \includegraphics[width=6cm, height=5cm] {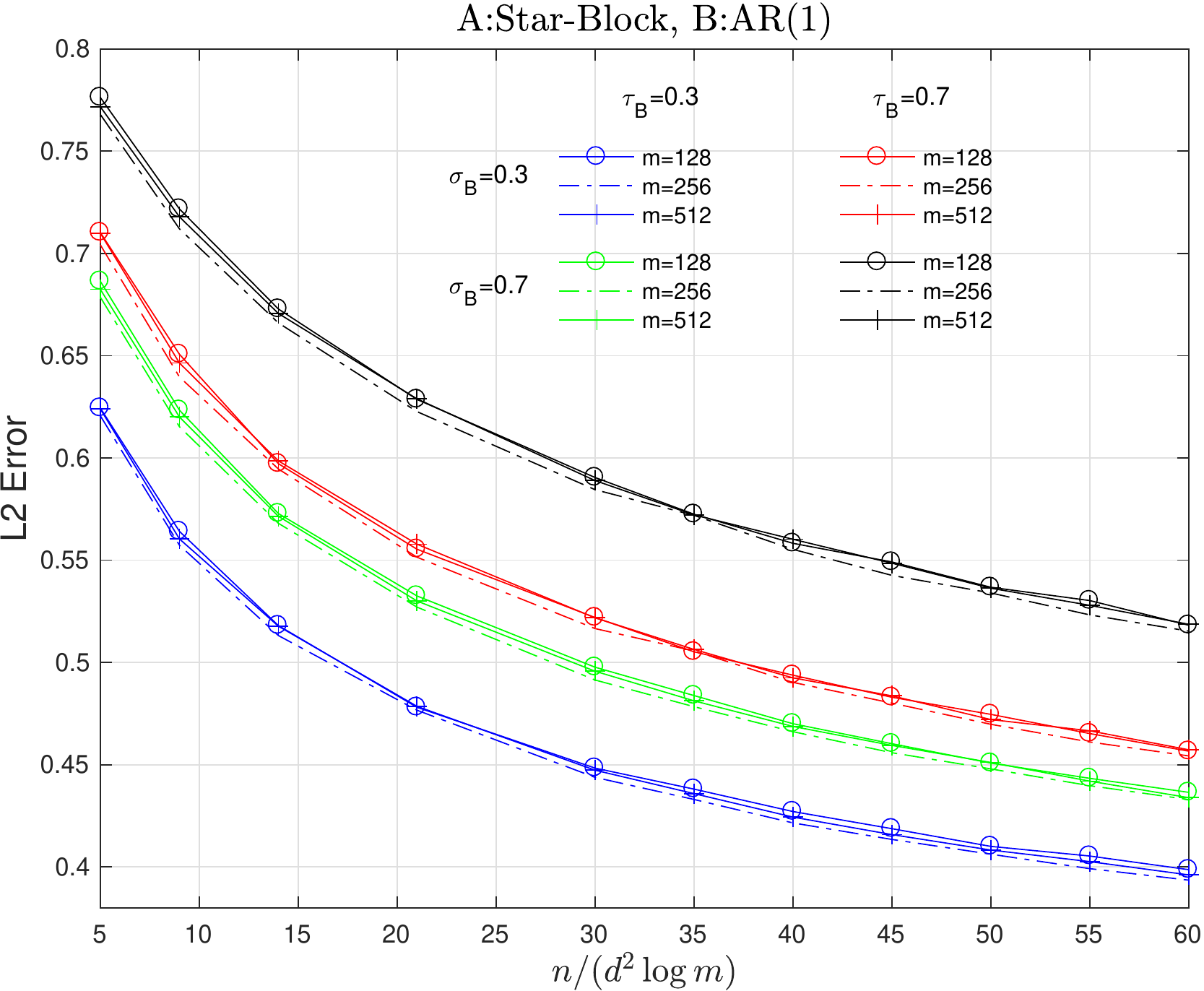}  &   \includegraphics[width=6cm, height=5cm] {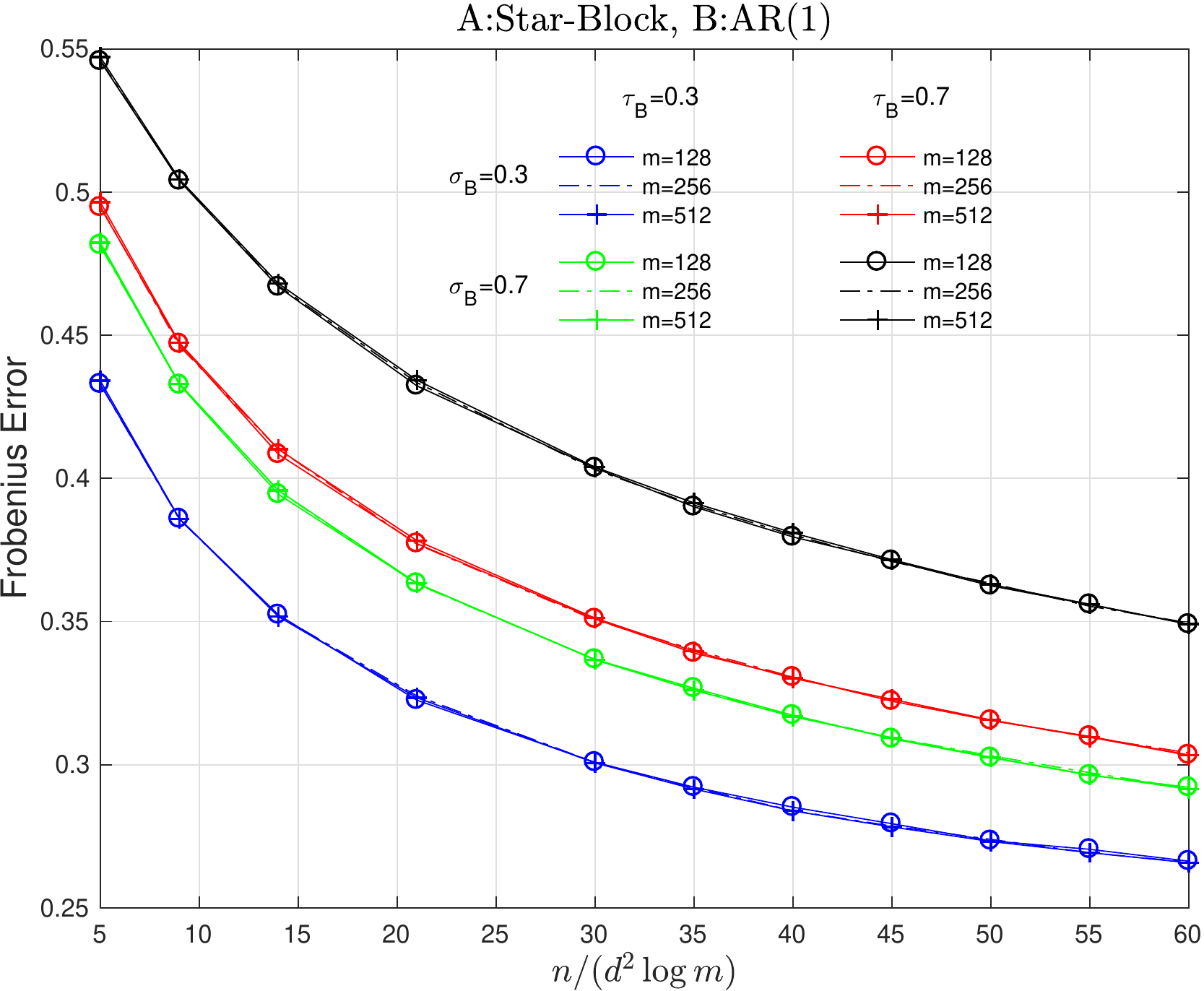} 
                                         \end{tabular}
  \caption{The top two plots display the relative L2 (left) and
    Frobenius (right) errors for the estimates when $A$ follows AR(1) and $B$ follows random model.
The bottom two plots display the relative L2 (left) and Frobenius
(right) errors when $A$ follows Star-Block and $B$ follows AR(1).
}
  \label{figee143BB_changingasd}
\end{figure}

Now we investigate the robustness of the graphical structure in estimated $\hat{\Theta}$ to the specification of $\hat{\tau}_A$.
Let $E$ and $\hat{E}$ denote support sets in $\Theta$ and $\hat{\Theta}$, respectively, i.e.,
\[
E=\{(i,j):\ i \neq j, \ \Theta_{ij} \neq 0\}, \quad  \hat{E}=\{(i,j):\ i \neq j, \  \hat{\Theta}_{ij} \neq 0\}.
\]
We consider the precision and recall of the estimated $\hat{E}$:
\[
\mathrm{precision}=\frac{|\hat{E} \cap E|}{|\hat{E}|},\quad \mathrm{recall}=\frac{|\hat{E} \cap E|}{|E|}.
\]

\begin{figure}[H]
\centering
  \begin{tabular}{@{}cc@{}}
         \includegraphics[width=6cm, height=6cm] {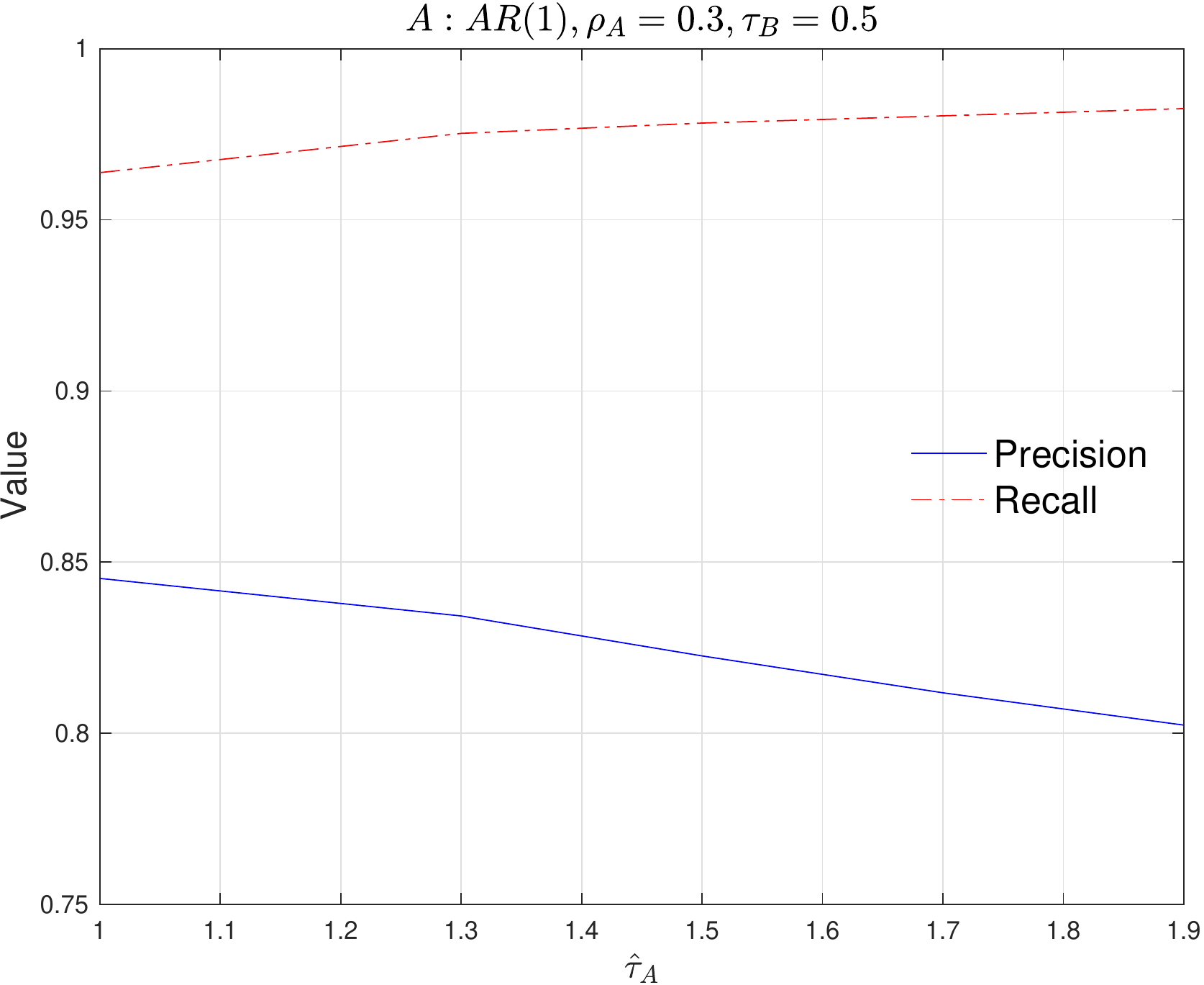}  &  \includegraphics[width=6cm, height=6cm] {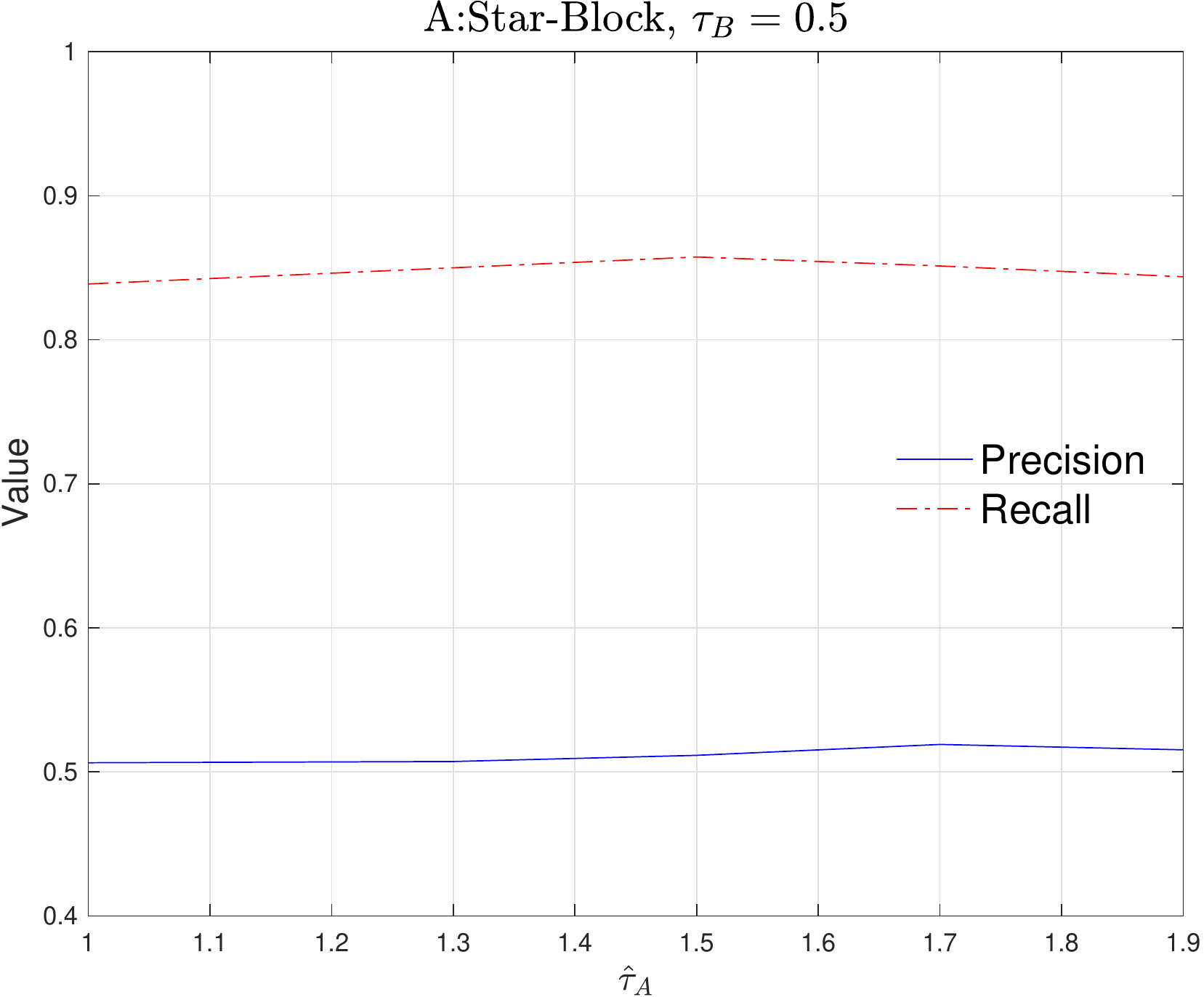}
      \end{tabular}
  \caption{The two plots display the precision and recall for the estimated $\hat{\Theta}$ with different specifications of $\hat{\tau}_A$ 
  when $m=512$, $n=256$, $\tau_A=1.5$, and $\tau_B=0.5$.
  The left plot is when $A$ follows AR(1) and $B$ follows random model. The right plot is when $A$ follows Star-Block and $B$ follows AR(1).
}
  \label{fig:robust}
\end{figure}

Figure \ref{fig:robust} displays the precision and recall against different specifications of $\hat{\tau}_A \in (1, 1.9)$ when $\tau_A=1.5$ and $\tau_B=0.5$.
We repeat experiments 100 times and record the average values of the precision and recall for each $\hat{\tau}_A$ case
using $\lambda^{(i)} = 2 \psi_0 \sqrt{\log m/n}$, $b_1=\sqrt{d} a_{\max}$, and $\eta = 1.5 \lambda_{\max}(A)$ in Algorithm 1.
For the left plot, $A$ and $B$ are generated using AR(1) and random model, and for the right plot, $A$ and $B$ are generated using Star-Block and AR(1) model, respectively.
In the left plot, we observe that $\hat{\tau}_A$ values in $(1,1.9)$ yield the precisions between $0.8$ and $0.85$ and recalls between $0.95$ and $0.98$. 
In the right plot, the precisions are between $0.50$ and $0.52$ and recalls are between $0.83$ and $0.86$. 
These examples show that the estimated graphical structures have similar performances in terms of precision and recall even when misspecified $\hat{\tau}_A$ values are used.

\subsection{Analysis on R-squared}
\label{subsec:analy:R-sq}
In this subsection, we compare the regular (with no measurement error)
and  the EIV regression estimators 
based on R-squared ($R^2$) values when the true model contains errors
in the covariates. Recall the following regular and EIV regression functions:
\begin{eqnarray}
  y &=& X \beta_1 + \epsilon, \quad
  \text{where} \; 
 X\ \rm{and} \ y \ \rm{are \ observable,} \label{eq:reg:reg}\\
y &=& X_0 \beta_2 + \epsilon, \quad  X=X_0+W,   \quad
\text{where} \;  X\ \rm{and} \ y \ \rm{are \ observable}. \label{eq:EIV:reg}
\end{eqnarray}
For the regular regression, we consider two regularization methods,
ridge regression and Lasso, to estimate $\beta_1$. 
Lasso provides a sparse solution under the setting that there are no
measurement errors.
Ridge regression adjusts the gram matrix $X^TX$ by adding some
positive constant to all diagonal components, which plays an opposite
role compared to the corrected form of gram matrix for the proposed EIV regression.
More specifically, Lasso $\hat{\beta}_L$ solves
\[
\min_{\beta} \frac{1}{2n} \|y-X\beta\|_2^2 +  \lambda_L \|\beta\|_1
\]
for some regularization parameter $\lambda_L >0$.
Ridge regression $\hat{\beta}_R$ solves for some regularization parameter $\lambda_R >0$,
\[
\min_{\beta} \frac{1}{n} \|y-X\beta\|_2^2 + \frac{\lambda_R}{n} \|\beta\|_2^2
\]
and has a closed form $\hat{\beta}_R= \left(X^TX + \lambda_R I_p\right)^{-1} X^Ty$.
To estimate $\beta_2$ for the EIV regression, we apply the composite gradient descent algorithm as in \eqref{grad_comp_2}.

We compare correlations between fitted and observed values as a
R-squared metric to 
decide on a correct model between the EIV and the regular regression model:
calculate the explanatory power of $X$ for $y$ using  
$R^2_X = \mathrm{corr}\left(X\hat{\beta}_1, y \right)^2$, where
$\hat{\beta}_1$ is either the Lasso or ridge regression estimator.
Similarly, define the explanatory power of $X_0$ for $y$ using the EIV 
estimator $\hat{\beta}_2$: 
\begin{eqnarray}
R^2_{X_0}  &=& \mathrm{corr}\left(X_0 \hat{\beta}_2, y \right)^2 = \frac{\cov(y, X_0 \hat{\beta}_2)^2}{\var(y)  \var(X_0 \hat{\beta}_2)}. \label{eq:R2}
\end{eqnarray}
We show in Lemma \ref{lem:r-sqaured} of the Supplementary materials that when $m$ is
fixed, the proposed $R^2$ metric asymptotically chooses a correct
model between the EIV and the regular regression model.

In practice, $R^2_{X_0}$ defined in \eqref{eq:R2} should be estimated since $X_0$ is not observed.
Under the setting in Theorem \ref{coro::Theta},
we obtain a lower bound of $R^2_{X_0}$ by the function of $\hat{\beta}$, $\hat{A}_+$, and $\hat{B}_+$, where $\hat{A}_+:=\hat{\Theta}_+^{-1}$ and $\hat{B}_+ :=\hat{\Omega}_+^{-1}$.
Here $\hat{\Theta}_+$ and $\hat{\Omega}_+$ are the positive definite
estimators obtained from Algorithm 3, as summarized in Section \ref{sec:add:proc}. Then, we have with high probability
\begin{equation}
\label{eq:R*}
R^2_{X_0}  \ge R^* :=  \frac{\left\{|\hat{\beta}_2^T X^T y| - \sqrt{2n}\sqrt{1+\frac{m\log n}{4n}}\|\hat{\beta}_2\|_2 \|y\|_2 \|\hat{B}_+\|_2^{1/2}\right\} \vee \ 0}{2\|y\|_2^2 \left(n+\frac{m\log n}{4}\right) \|\hat{A}_+\|_2 \|\hat{\beta}_2\|_2^2 }.
\end{equation}
This lower bound $R^*$ will be used instead of the unknown $R^2_{X_0}$.
See Section \ref{sec:R} of the Supplementary materials  for detailed mathematical derivations of \eqref{eq:R*}. 
We compare $R^*$ and $R^2_X$ to make a correct choice between the EIV and the regular regression model.
We consider the following two examples.

\begin{example}
\label{ex1}
Let $A^*$ and $B^*$ follow AR(1) model with $\rho_A=0.5$ and $\rho_B=0.5$.
Consider the EIV regression
\begin{eqnarray*}
y  &=& X_0 \beta_2 + \epsilon, \quad  X=X_0+W,  \; \; \text{ where}  \quad  X\ \rm{and} \ y \ \rm{are \ observable} \\
\beta_2 &=&0.5 \cdot [1.3,1.3,1.1,-1.1,1,1, 1, 0.9,0.9,1, 0, \cdots, 0]^T, \quad \epsilon_i \sim N(0,1)
\end{eqnarray*}
with $n=500$ and $m=1000$, and $A=\tau_A A^*$ and $B=\tau_B B^*$,
where $(\tau_A, \tau_B) =(1.5,0.5)$.  
\end{example}

\begin{example}
\label{ex2}
Consider the regular regression
\begin{eqnarray*}
y &=& X \beta_1 + \epsilon, \; \text{ where } \; \ X\ \rm{and} \ y \ \rm{are \ observable,}\\
\beta_1 &=& 0.5\cdot [1.3,1.3,1.1,-1.1,1,1, 1, 0.9,0.9,1, 0, \cdots, 0]^T, \quad \epsilon_i \sim N(0,1)
\end{eqnarray*}
with  $n=500$ and $m=1000$, and the rows of $X$ are i.i.d. $N(0,A)$, where the covariance matrix is $A=2A^*$ and $A^*$ follows AR(1) model with $\rho_A=0.5$.
\end{example}

\begin{figure}[H]
\centering
  \begin{tabular}{@{}cc@{}}
 \includegraphics[width=6cm, height=6cm] {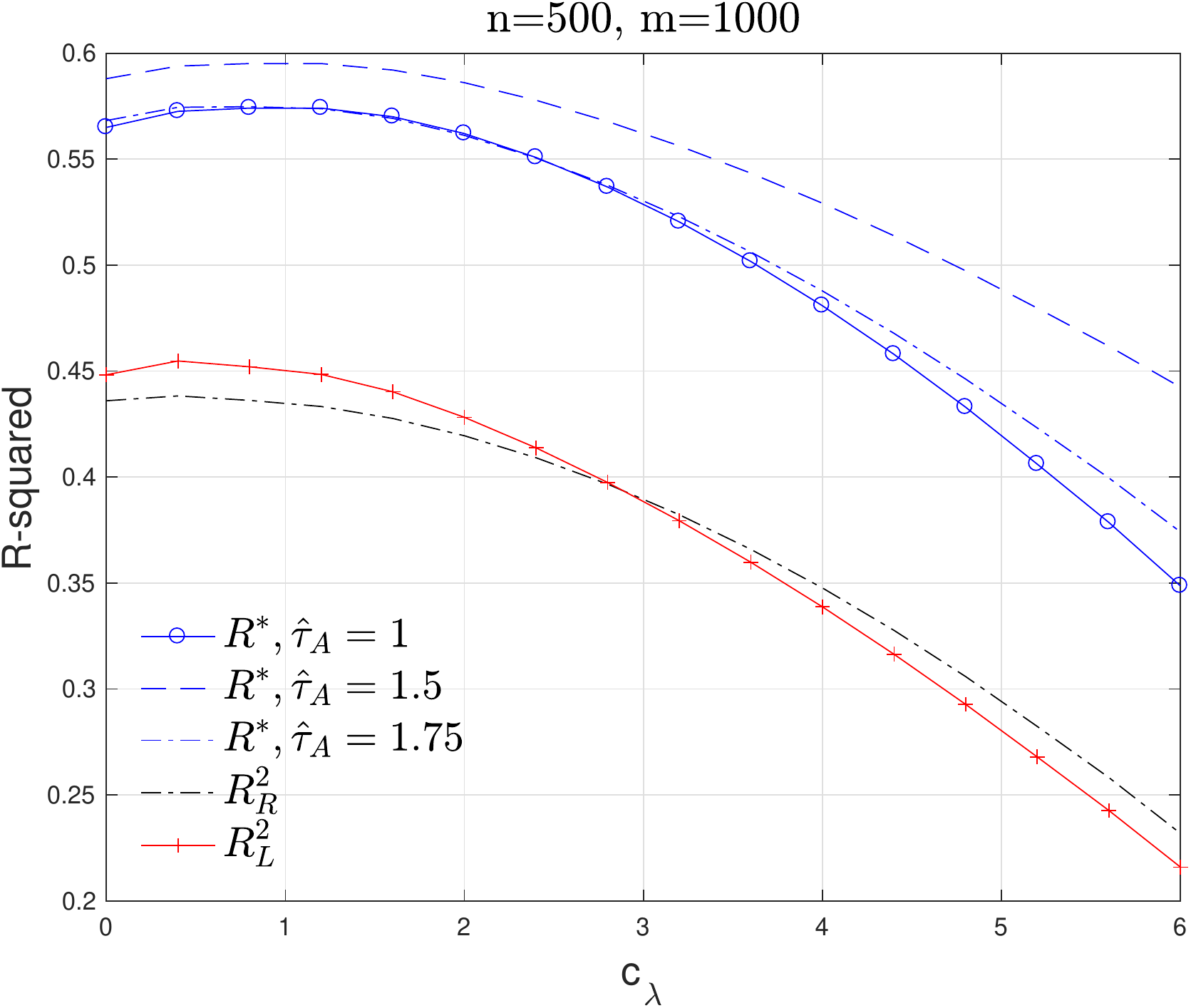}  
&   \includegraphics[width=6cm, height=6cm] {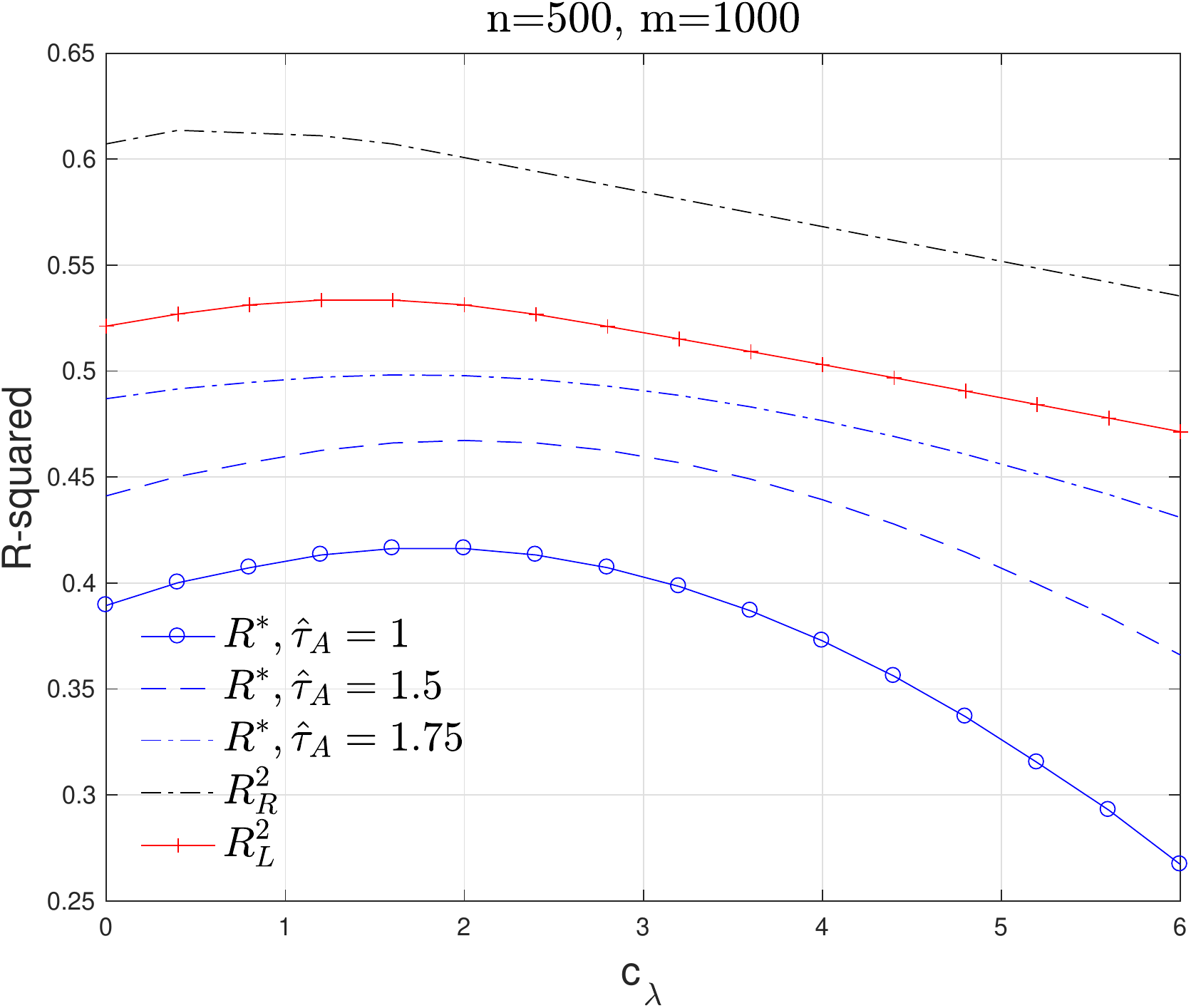}
     \end{tabular}
  \caption{
  The first plot shows the lower bound of the R-squared values of EIV
  regression (i.e., $R^*$), and the calculated R-squared values for
  the Lasso ($R_L^2$) and ridge regression ($R_R^2$) 
when the actual model follows the EIV model with $(\tau_A, \tau_B) = (1.5, 0.5)$ (Example \ref{ex1}).  
 The second plot considers the case in which the actual model follows regular linear model (Example \ref{ex2}).
}
  \label{fig-Rsqure-sim}
\end{figure}

We apply the EIV and regular regressions, that is, Lasso and ridge regression
without measurement errors, 
to models in Examples \ref{ex1}-\ref{ex2}.  
For the Lasso and EIV, the regularization parameters $\{c_{\lambda} \sqrt{{\log m}/{n}} \mid c_{\lambda} \in (0,6)\}$ are used. 
For the ridge regression,  the regularization parameters $\{nc_{\lambda} \mid  c_{\lambda} \in (0,6)\}$ are used.
The first plot of Figure \ref{fig-Rsqure-sim} shows $R^2_X$ and $R^*$ values averaged over $100$ trials when the true model follows EIV (Example \ref{ex1}).
We observe that $R^*$ is generally larger than R-squared values of Lasso ($R_L^2$) and ridge regression ($R_R^2$) even when misspecified $\hat{\tau}_A$ values are used. 
The second plot of Figure \ref{fig-Rsqure-sim} shows $R^2$ values when the true model does not involve errors-in-variables (Example \ref{ex2}). 
We see that $R^2_L$ and $R^2_R$ are larger than $R^*$ obtained by the EIV regression.

These findings show that comparison based on $R^2$ values of the EIV
regression versus the Lasso or Ridge regression can help distinguish
one regression model from the other    
and the EIV regression model is robust to the specification of $\tau_A$.


\section{Analysis of hawkmoth neural encoding data}
\label{sec:realdata}

In the experiments conducted by \cite{Sponberg:2012} and
~\cite{Sponberg:2015}, the neural firing time was measured using electrical probes, and the torque
was measured via an optical torque meter. 
For each of the 7 hawkmoths indexed $i=1, \ldots, 7$, data for $n_i$ wing
strokes were obtained.  For a single wing stroke, the measured torque $X^{(i)}
\in \R^{n_i \times 500}$ consists of 500 sampled torque values  
spanning the duration of the wing stroke.  Let $\Delta^{(i)} \in
\R^{n_i}$ be the $\Delta$ values for moth $i$, obtained at the 
 beginning of each wing stroke.  Two pre-processing approaches give rise              
to separate ``spike-triggered'' and ``phase-triggered'' data sets.   
We omit the index $i$ when no confusion arises.

\subsection{Generating model}
The data $X^{(i)}$ for one moth can be viewed as a matrix-variate
observation, with one axis (rows) indexing wing strokes, and
the other axis (columns) corresponding to time points within a wing stroke.
The Kronecker sum model decomposes this observation into two latent terms. 
One of these terms has
independent rows (wing strokes) and dependent columns (time points),
and the other term has dependent rows and independent columns. 
Our primary question here is whether the information in $X^{(i)}$ that is
related to $\Delta^{(i)}$ is concentrated in only one of these two
latent terms.  This would suggest that the other term may contain
motion features that are not strongly neurally encoded, for example
measurement noise.

For ease of visualization, in the graphical analysis the rows of the
torque data $X^{(i)}$ are sorted in descending order by the corresponding $\Delta$.
Figure \ref{fig3_spat_corr} in the Supplementary materials shows heat maps of the sample
correlation matrix of $500$ temporal points and the sample correlation
matrix of wing strokes, showing dependencies in both time points and among the wing strokes.
We will use the Kronecker sum model to explain these two-way dependencies in the torque data,
and EIV regression to explore the joint relationship between neural firing and torque.  .  

We decompose the observed torque data $X$ as
\begin{equation}
\label{eq:data:gen0}
X_{ij} = R_i + C_j + \phi_{ij} \quad \mathrm{for}\quad i=1,\cdots, n,\ j=1,\cdots, 500,
\end{equation}
where $R_i$ and $C_j$ represent the $i$th row mean and $j$th column
mean of $X$, respectively (when $\sum_{ij} X_{ij} = 0$).
The matrix form of \eqref{eq:data:gen0} is
\begin{equation}
\label{eq:data:gen}
X=  R 1^T + 1 C^T + \Phi,
\end{equation}
where  $1$ is the vector of ones, $\Phi_{n \times 500} = (\phi_{ij})$,
$R =(R_1, \cdots, R_n)^T$ and $C =(C_1, \cdots, C_{500})^T$.
In Table \ref{table:prop},     
we show the composition of three components in \eqref{eq:data:gen} through
$\|R 1^T\|_F^2/\|X\|_F^2$, $\|1 C^T\|_F^2/\|X\|_F^2$ and $\|\Phi\|_F^2/ \|X\|_F^2$.
We see that column effects are negligible and at least $70 \%$ of the
variation in the data can be explained by $\Phi$ for all moths.
As noted above, the row means are already known to be a major 
determinant of the spike times $\Delta$.  
Therefore, all analysis is based on the row-centered data
$\tilde{X}={X} - R1^T$, to focus on the possible existence of additional neurally-encoded features in the wing strokes.

\begin{table}[H]
\caption{Proportion of each part in data.}
\label{table:prop}
\centering
\begin{small}
\fbox{
\begin{tabular}{ c c  c  c  c  c c c  } 
Proportion & J &K & L & M & N & P & Q  \\ [0.5ex]  \hline 
$\|R 1^T\|_F^2/\|X\|_F^2$ & 0.08 &    0.09&    0.27&  0.11    & 0.22   &   0.06&    0.18
 \\ [0.5ex] 
$\|1 C^T\|_F^2/\|X\|_F^2$&     0.02 &    0.01 &     0.01 &      0.02 &    0.01  &  0.02   & 0.01
 \\ [0.5ex] 
$\|\theta\|_F^2/ \|X\|_F^2$& 0.71 &    0.90 &   0.71  &     0.87&      0.76    &   0.91  &  0.80
 \\ [0.5ex] 
\end{tabular}
}
\end{small}
\end{table}

\subsection{Interpretation of the graphical structure}
\label{sec:moth:inter}
In this subsection, we discuss the graphical structures of the fitted
covariance matrices $A$ and $B$ 
from the Kronecker sum model.
Without loss of generality we scale $\tilde{X}$ so that its overall
variance is 2.

\begin{figure}[H]
\centering
  \begin{tabular}{@{}cc@{}}
      \includegraphics[width=.5\textwidth]{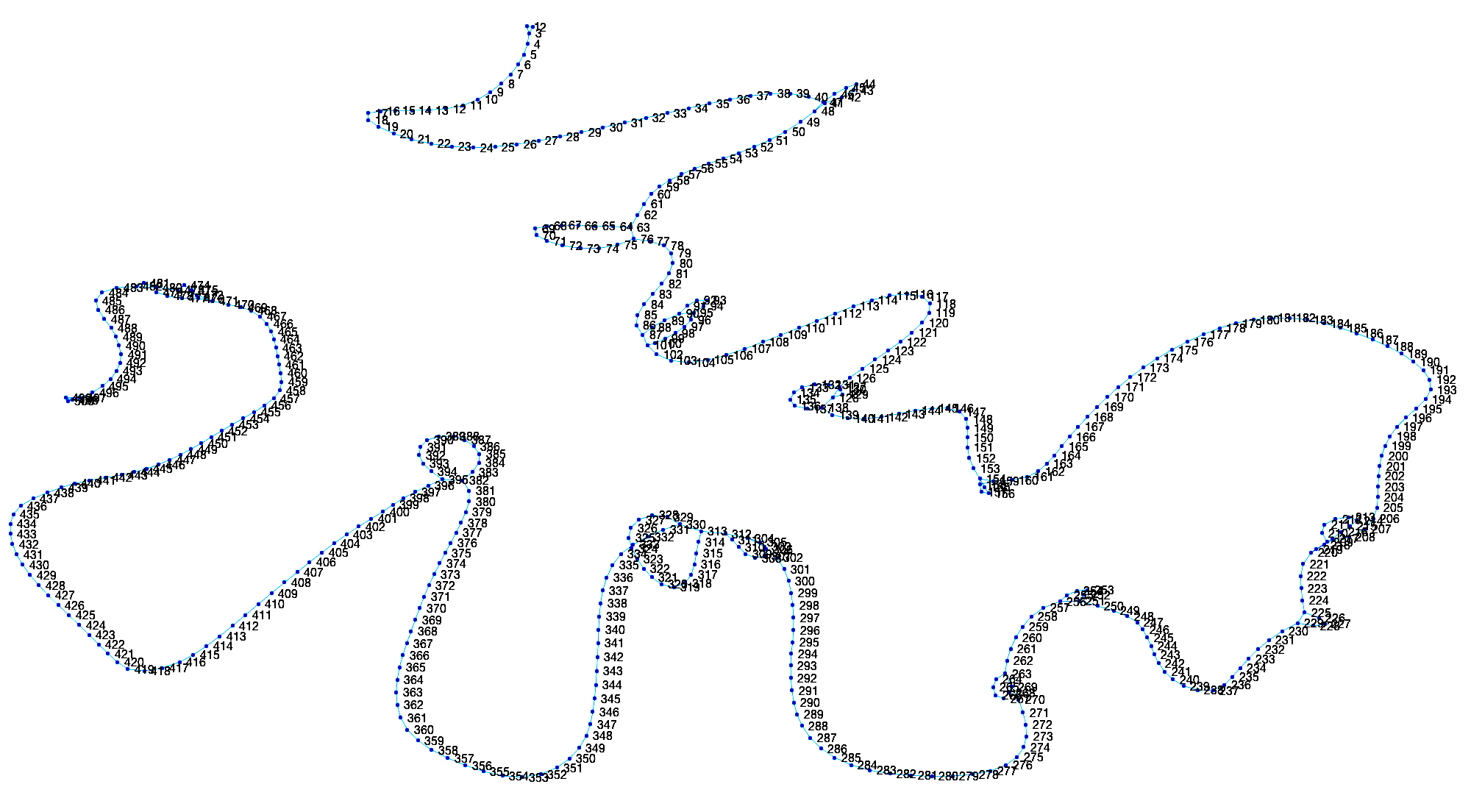}  &   \includegraphics[width=.52\textwidth]{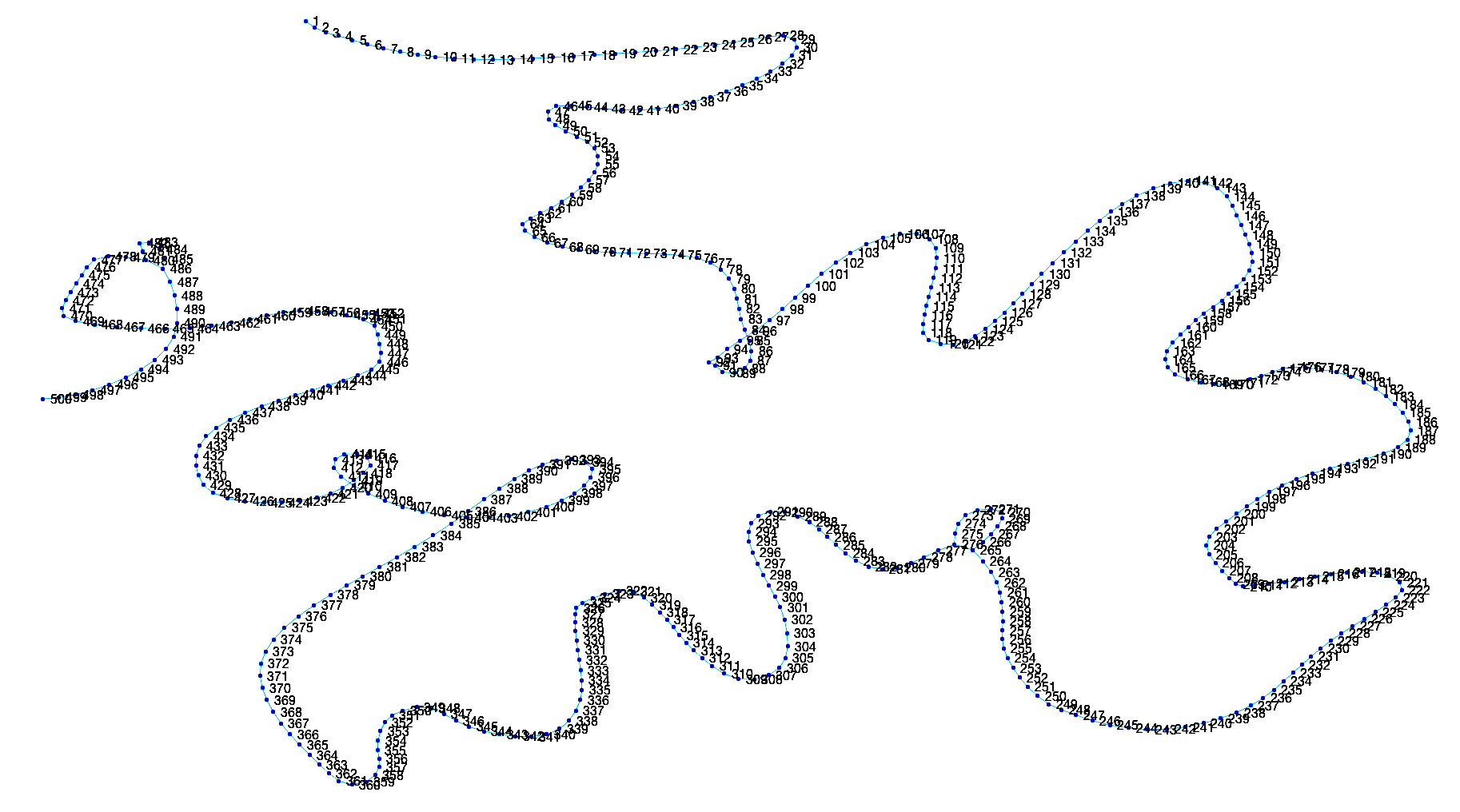}
      \end{tabular}
  \caption{The graphical structure of $\hat{\Theta}$ when $\hat{\tau}_A=1.5$ and $\lambda^{(i)}_A=0.5\sqrt{\frac{\log 500}{n}} \asymp 0.05$ for moth J (left) and L (right). 
  Edges $(i,j)$ with $|\hat{\Theta}_{ij}| > 0.05$ are connected in the graph.
  } 
\label{fig3cov2}
\end{figure}

Figure \ref{fig3cov2} displays the graphical structures of the
estimator $\hat{\Theta}$ for the Kronecker sum model when
$\hat{\tau}_A=1.5$. It is apparent that both graphs have a chain
structure, with some loops, which is expected because $\hat{\Theta}$ 
encodes the conditional dependency relationships between and across
time points, and the observed torque data (across time points) are therefore strongly autocorrelated.

\begin{figure}[H]
\centering
  \begin{tabular}{@{}c@{}}
       \includegraphics[width=12cm, height=3.5cm]{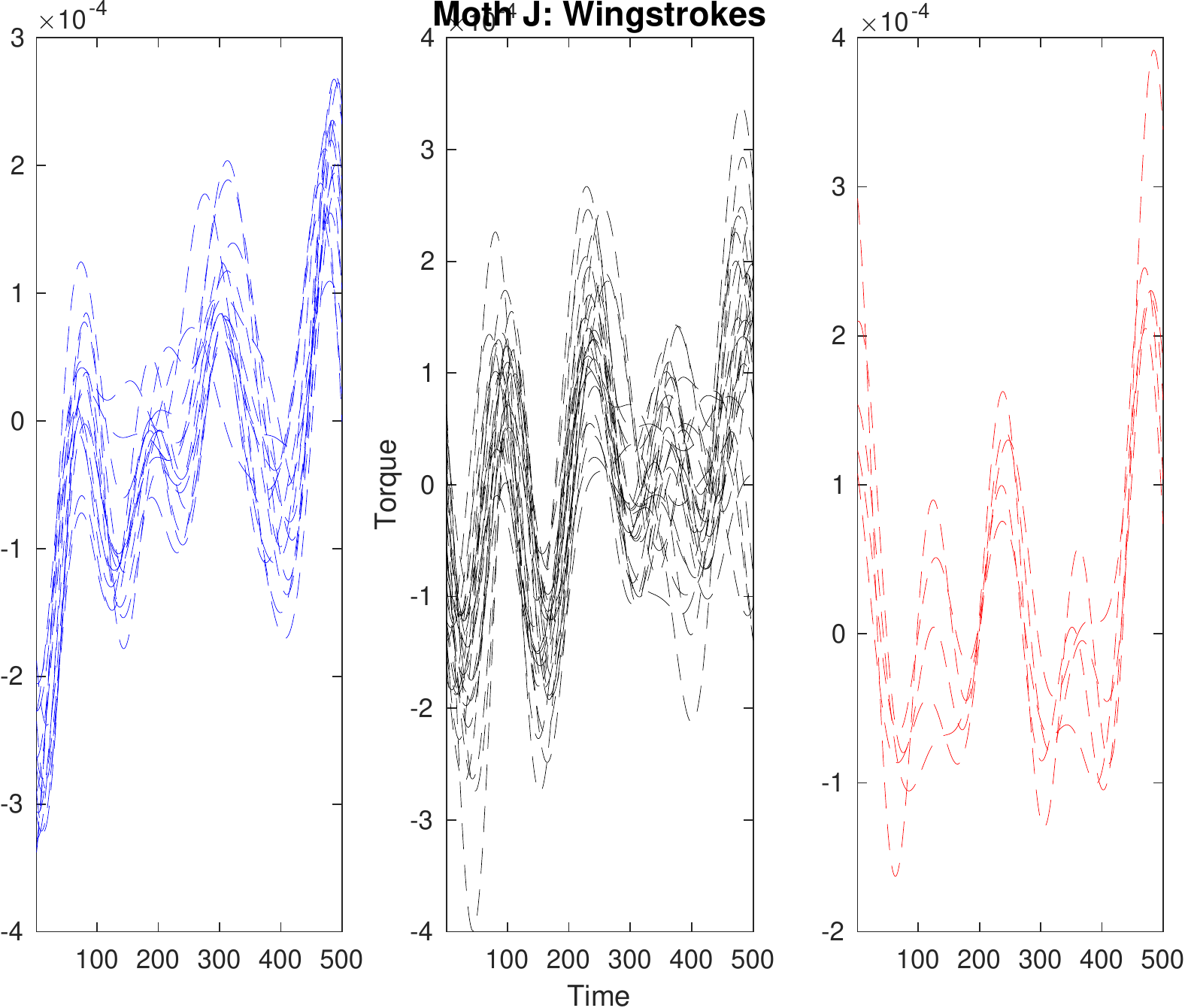} \\   
              \includegraphics[width=12cm, height=3.5cm]{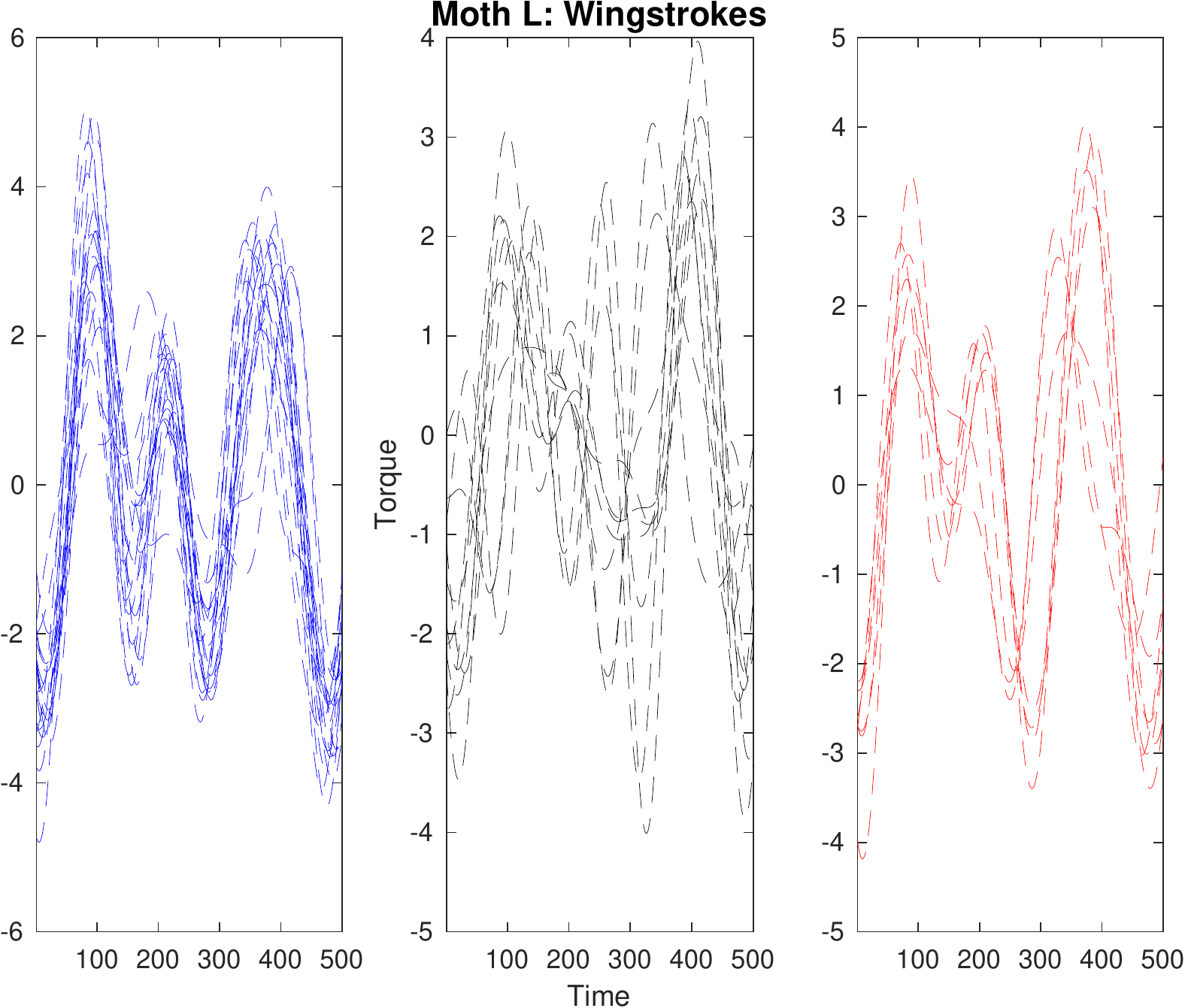}   
                      \end{tabular}
            \caption{Torque ensembles for selected three clusters for moth J (top) and L (bottom).}   
\label{fig3cov_comp_BB}
\end{figure}

We select the three largest clusters of wing strokes from the
estimated structure of $\Omega$ when $\hat{\tau}_A = 1.5$. 
See Figures \ref{fig3covb1}-\ref{fig3covb2} of the Supplementary materials for 
graphical structure of $\hat{\Omega}$.
Figure~\ref{fig3cov_comp_BB} displays torque ensembles of the three clusters for moth J and L, respectively. 
It is seen that the wing strokes from the same cluster
are highly correlated or highly anti-correlated, which implies that
the structural information encoded in 
$\hat{\Omega}$ is not arbitrary and indeed captures strong dependencies among wing strokes.

\subsection{Regression analysis of torque and DLM}
\label{sec:moth:reg}
In this subsection,  we analyze the
relationship between the neural firing time and torque ensemble data 
utilizing various regression methods, including EIV regression. 
In the EIV regression, we model $\cov(\mvec{\tilde{X}})$ using the
Kronecker sum structure, as elaborated in the previous section.

We first perform simple linear regression of $\Delta = t_L-t_R$ on the
mean torque ($\mu$) to regress out the mean torque effect from $\Delta$.
Denote by $\tilde{\Delta}$ the residual vector:  
\begin{equation}
\label{simp_reg}
\tilde{\Delta} := \Delta- \alpha- \beta \mu, \quad \alpha=\bar{\Delta}-\beta\bar{\mu}, \quad \beta=\frac{\cov(\mu,\Delta)}{\var(\mu)},
\end{equation}
where $\bar{\Delta}$ and $\bar{\mu}$ are the means of the components in  $\Delta$ and $\mu$, respectively. 
We will treat $\tilde{\Delta}$ as a response vector and fit a separate
linear model to each hawkmoth, to explore the possibility that the average torque $\mu$ does not capture all the information 
in $\Delta$.

We then fit both regular regression  \eqref{eq:reg11}  and EIV
regression model \eqref{eq:nonsm}, as previously studied in 
\cite{RZ15}, to relate the residual neural spike time differences \eqref{simp_reg} and torque values:
\begin{eqnarray}
 \label{eq:reg11} 
&&
\tilde{\Delta} =\tilde{X} \beta_1+ \epsilon, 
\quad \text{where} \; \tilde{X}\ \rm{and} \ \tilde{\Delta}  \ \rm{are
  \ observable;} \\
\label{eq:nonsm}
&&
\tilde{\Delta} = X_0 \beta_2 + \epsilon, \quad  \tilde{X}=X_0+W,
\quad \text{ where } \; \tilde{X}\ \rm{and} \; \tilde{\Delta} \; \rm{are \ observable}. 
\end{eqnarray}
By using the R-squared analysis from the above EIV regression, we will
argue that variation in $\tilde{\Delta}$ is
more  strongly related to the latent $X_0$ in model \eqref{eq:nonsm} compared to the observed $\tilde{X}$.

Based on the fact that $\beta_1, \beta_2 \in \R^{500}$ are
coefficient vectors whose domain is a temporal space, we aim to obtain
smoothed estimators for $\beta_1$ and $\beta_2$ to explain the
variation in $\tilde{\Delta}$.   
Toward this goal, we use B-spline basis functions to approximate these coefficients.  
Let $\pi_1(t),\cdots, \pi_{k_n+l}(t)$ be the normalized B-spline basis
functions of order $l$ with $k_n$ quasi-uniform knots. We set $l=3$
and $k_n=17$, 
so we use $K_n := k_n+l=20$ basis functions.
Let $\tilde{\Pi}_{t,k} = \pi_k(t)$ for $t=1,2\cdots, 500$ and $k=1,2,\cdots, K_n$.
We use an orthonormal basis for the column space of $\tilde{\Pi}$ 
obtained through QR decomposition, denoted by $\Pi \in
\R^{500 \times K_n }$, that is, $\Pi^T \Pi = I_{K_n }$.
Upon estimating $\zeta_k$, we obtain estimators of the $\beta_k$
via the relationship $\beta_k = \Pi \zeta_k$, where $\zeta_k \in \R^{K_n  \times 1}$.
Note that this relationship preserves at least $95\%$ of the variance
of the data $\tilde{X}$ in the sense that $\|\tilde{X}\Pi\|_F^2 \ge 0.95 \|\tilde{X}\|_F^2$ for all moths.

Models \eqref{eq:reg11} and \eqref{eq:nonsm} can now be rewritten as
\begin{eqnarray}
\label{eq:reg22}
&&\tilde{\Delta} = \tilde{X} \Pi \zeta_1+ \epsilon, \quad 
\text{  where} \; \; 
\tilde{X}\Pi\ \rm{and} \ \tilde{\Delta}  \ \rm{are \ observable;} \\
\label{eq:nonsm2}
&&\tilde{\Delta} = X_0 \Pi \zeta_2 + \epsilon, \quad  \tilde{X} \Pi
=X_0 \Pi +W \Pi ,  
\end{eqnarray}
where $\tilde{X}\Pi\ \rm{and} \ \tilde{\Delta}$ are 
  observable.

Note that the model \eqref{eq:nonsm2} still follows the EIV model as in~\cite{RZ15},
since the covariance matrix of $\mvec{\tilde{X}\Pi}$ is 
\bens
\cov(\mvec{\tilde{X}\Pi}) =\Gamma_A \oplus B, \quad \text{  where} \; \;  \Gamma_A := \Pi^T A \Pi.
\eens
See Subsection \ref{subsect:reg:a} of the Supplementary materials for details.

We obtain $\hat{\zeta}_1$, an estimator of $\zeta_1$ using ridge
regression or the Lasso by treating $\tilde{X}\Pi$ as a design matrix. 
We obtain the estimator $\hat{\zeta}_2$ of $\zeta_2$ by solving EIV
regression \eqref{eq::origin} with  
$  \hat{\tau}_B=   \left(\frac{1}{500n}\|\tilde{X}\|_F^2 -  \tau_A\right)_+ $
and 
\bens
\hat\Gamma  =  
\inv{n} (\tilde{X}\Pi)^T \tilde{X}\Pi -  \hat{\tau}_B I \;\; 
 \text{ and } \; \;  \hat\gamma \; = \; \inv{n} (\tilde{X}\Pi)^T \tilde{\Delta},
\eens
where $\hat\Gamma$ is an unbiased estimator of $\Gamma_A$.
For the objective function~\eqref{eq::origin},  
we choose $b_1 = 2\|\hat{\zeta}_1\|_1$ and step size parameter $\eta =
0.5 \|\hat{\Gamma}\|_2$ for the composite gradient descent algorithm, where $\hat{\zeta}_1$ is the Lasso estimator.
Then we obtain estimators $\hat{\beta}_1 = \Pi \hat{\zeta}_1$ and $\hat{\beta}_2 = \Pi \hat{\zeta}_2$.  

We calculate the explanatory power of $\tilde{X}$ for $\tilde{\Delta}$ using the proportion of explained variance as follows:
$R^2_{\tilde{X}} = \corr\left(\tilde{X}\hat{\beta}_1, \tilde{\Delta} \right)^2$. 
We record the lower bound of the explanatory power of $X_0$ for $\tilde{\Delta}$ using the estimator $\hat{\beta}_2$: 
\begin{eqnarray}
R^2_{X_0}  &=& \corr\left(X_0 \hat{\beta}_2,\tilde{\Delta} \right)^2  \ge R^*,
\label{eq:R2-2}
\end{eqnarray}
where $R^*$ is defined in \eqref{eq:R*}.

Table \ref{table:mothpower} shows the values of $R_{\tilde{X}}^2$ obtained for ridge regression and the Lasso, denoted respectively as
$R^2_{\tilde{X}}(Ridge)$ and $R^2_{\tilde{X}}(Lasso)$.
We record the maximum value of $R^2_{\tilde{X}}(Lasso)$ and $R^*$ among the values obtained from the solution paths with the regularization parameters $c_{\lambda} \sqrt{\log K_n/n}$, 
where $c_{\lambda} \in (0, 7)$.  To obtain $R^*$, we consider the EIV with $\hat{\tau}_A \in \{0.1,0.2, \cdots, 1.9\}$.
Similarly, we record the maximum value of $R^2_{\tilde{X}}(Ridge)$ among the values obtained from the ridge regression solution path with the regularization parameters $n c_{\lambda}$, where $c_{\lambda} \in (0, 7)$.
While the Kronecker sum model is not sufficient to identify $\tau_A$ from the random matrix $X$ alone, when considered in the context of a regression analysis we can identify $\tau_A$ by maximizing the $R^2$ values of the regression between the responses and the latent component $X_0$.
We record the optimal $\hat{\tau}_A$ that provides the highest $R^*$ value.

\begin{table}[H]
\caption{\label{table:mothpower}Explanatory power ($R^2$)} 
\centering
\begin{small}
\fbox{
\begin{tabular}{ c l c l c l  c | c | c } 
\multirow{2}{*}{Spike} & \multirow{2}{*}{$R^2_{\tilde{X}}(Ridge)$} &  \multirow{2}{*}{$R^2_{\tilde{X}}(Lasso)$} &  \multicolumn{2}{c}{EIV}  \\ [0.5ex] 
   &          &        &  $\hat{\tau}_A$ & $R^2_{X_0}(R^*)$  \\ [0.5ex]  \hline
J  & 0.12 & 0.12 & 1.4 & 0.20  \\ [0.5ex] 
K & 0.12  & 0.11 & 1.4 & 0.22 & \\ [0.5ex] 
L & 0.15 & 0.18 & 1.4 & 0.26  \\ [0.5ex] 
M & 0.02  & 0.02 & 1.5 & 0.17  \\ [0.5ex] 
N & 0.15  & 0.21 & 1.3 & 0.28  \\ [0.5ex] 
P & 0.19 & 0.19 & 1.5 & 0.29\\ [0.5ex] 
Q & 0.19 & 0.21 & 1.4 & 0.35 \\ [0.5ex] 
\end{tabular}
}
\end{small}
\end{table}

In all cases,  $R^*$, the lower bound of $R^2_{X_0}$ obtained from
errors-in-variables regression, is greater 
than $R^2_{\tilde{X}}$ obtained from the Lasso and Ridge regression.  
This suggests that the torque signal component $X_0$ expressed in the Kronecker sum model may be viewed as a denoised torque signal, 
and that this denoised torque signal is more strongly correlated to the motor signal  
than the observed torque signal.

To ensure that this increase in $R^2$ values as summarized in Table
\ref{table:mothpower}
 is not due to chance or overfitting, we randomly permute the
 components of $\Delta$ 
and refit the regression models, while keeping $\tilde{X}$ intact. 
We summarize the $R^*$ values in Table \ref{table:mothpower3}, and 
observe that all the obtained $R^*$ values 
are between zero and $0.05$.
Combined with the simulation results presented in Figure \ref{fig-Rsqure-sim} in Subsection \ref{subsec:analy:R-sq}, 
this implies that the increase in $R^2$ values reflects 
a true relationship between $\tilde{\Delta}$ and the latent component
of $\tilde{X}$, and is not due to chance or overfitting.

\begin{table}[H]
\caption{\label{table:mothpower3}Explanatory power (R-squared)} 
\centering
\begin{small}
\fbox{
\begin{tabular}{ c l c l c l  c | c } 
{Spike} &{$R^2_{\tilde{X}}(Ridge)$} &  {$R^2_{\tilde{X}}(Lasso)$} & {$R^2_{X_0} (R^*)$}  \\  [0.5ex] \hline 
J  & 0.01 & 0.01 & 0.01  \\ [0.5ex] 
K & 0.02  & 0.03 & 0.02  \\ [0.5ex] 
L & 0.02 & 0.01  & 0.02  \\ [0.5ex] 
M & 0.01  & 0.02 & 0.01  \\ [0.5ex] 
N & 0.03  & 0.03  & 0.04  \\ [0.5ex] 
P & 0.01 & 0.01 & 0.00 \\ [0.5ex] 
Q & 0.00 & 0.01  & 0.00 \\ [0.5ex] 
\end{tabular}
}
\end{small}
\end{table}

\section{Conclusion}
\label{sec::conclude}
Data with complex dependencies arise in many settings, for example 
when a large number of replicated experiments are run on
subjects in a research study, a practice that is common in psychology,
linguistics, neuroscience, and other areas.  
The Kronecker sum provides a non-separable alternative to widely-used
separable covariance models such as the Kronecker product, and may fit
data better in some circumstances.

We illustrate these new methods and theory using data from a neural
encoding study of hawkmoth flight behavior.  
We provide two novel insights about these data.  
First, we illustrate that although the mean torque (per wing stroke)
captures a substantial fraction of the neurally encoded flight turning
behavior, additional components of the wing stroke trajectory also
appear to be neurally encoded.  
Second, we show that the observed flight torque trajectories can be
decomposed into two latent components, with one component capturing the majority of the neurally encoded
behavior.  The latter observation provides a promising basis for characterizing neural encoding using latent structures.

\begin{table}[H]
\caption{\label{table_def}The Symbols}%
\centering
\fbox{
\begin{tabular}{l  c } 
Parameters   & Definitions \\  [0.5ex] 
\hline
 $X_{i}$ &  The $i$th column of a matrix $X$ \\
 $X_{\minus i}$ &  The sub-matrix of $X$ without the $i$th column \\
 $X_{\minus i, \minus j}$ &  The sub-matrix of $X$ without the $i$th row and $j$th column \\
$\tau_{\Sigma}$ &   $\rm{tr}(\Sigma)/p$ for a square matrix $\Sigma \in \R^{p \times p}$\\
$\Theta$ & $A^{-1}$\\
$\Omega$ & $B^{-1}$\\
$a_{\max}$ & $\max_i A_{ii}$ \\
$b_{\max}$ & $\max_i B_{ii}$ \\
$\alpha$ & $\frac{5}{8} \lambda_{\min}(A)$\\
$\eta \ge \frac{11}{8}\lambda_{\max}(A)$ & step size parameter \\
$\vp(s_0)$  &   $\rho_{\max}(s_0, A)+\tau_B$ \\
$\tau_0$ &  $\frac{400C^2 \vp(s_0+1)^2}{\lambda_{\min}(A)}$\\
$D'_0$ &  ${\twonorm{B}}^{1/2} + a_{\max}^{1/2}$\\
$\tilde{D}'_0$  & ${\twonorm{A}}^{1/2} + b_{\max}^{1/2}$\\
$D_{\ora}$  &   $2(\twonorm{A}^{1/2} + \twonorm{B}^{1/2})$\\
$\tau_B^{+/2}$  & $\sqrt{\tau_B} + \frac{D_{\ora}}{\sqrt{m}}$\\
$K$ &   $\sup_{p \ge 1} p^{-1/2}(\mathbb{E}\abs{X}^p)^{1/p}$ for $X \sim N(0,1)$  \\
$C_\psi$ & $ K^2 C_0 D_0'  \left(\tau_B^{+/2} \kappa(A)+ \sqrt{a_{\max}} \right)$ \\
$\psi_0$ &   $0.1 D_0' \left(\tau_B^{+/2} a_{\max}+ \sqrt{a_{\max}}\right)$\\
$\rho_{\max}(d, A)$ &  $\max_{t \not= 0; d-\text{sparse}} \; \;\shtwonorm{A
  t}^2/\twonorm{t}^2$\\
$\vp(s_0)$ & $\rho_{\max}(s_0, A)+\tau_B$\\
$s_0$  &    The largest integer satisfying $\sqrt{s_0} \vp(s_0) \le \frac{\lambda _{\min}(A)}{32 C}\sqrt{\frac{n }{\log m}}$\\
$\psi_1^2$ &   $\frac{16c}{\lambda_{\min}(A)}\left(\frac{1}{1-\kappa}+\frac{\lambda_{\min}(A)}{2s_0}  \right)$\\
\end{tabular}
}
\end{table}

\section*{\textbf{\large Supplementary materials}}

This supplementary material is organized as follows.
Section \ref{sec:add:proc} introduces an additional estimation procedure of $\Theta$ as described in Section \ref{sec:graphical} of the main paper. Section \ref{Sec:gau} includes details of theoretical properties of the estimator, and main lemmas are presented in Section \ref{proof:main}.
Proofs of Theorem 1 is included in Section \ref{sec:thm1}.
Section \ref{sec:cons:omega} shows the consistency of the estimator $\hat{\Omega}$.
Section \ref{sec:pre} contains the proofs of lemmas.
Section \ref{sec:R} shows details of R-squared analysis.
Section \ref{sec:add:real} includes additional real data analysis.
Section \ref{sec:add:fig} includes additional tables and figures.

Theorem \ref{thm::nodewise} is directly from \citet{RZ15} by adopting the nodewise regression case with Gaussian random ensembles. Theorems \ref{coro::Theta} and \ref{coro::Omega} show estimation error bound for $\hat{\Theta}$ and $\hat{\Omega}$.
We emphasize that Lemma \ref{lemma::AD} holds for a subgaussian model while Lemma \ref{lemma::low-noise} assumes Gaussian model.
We note that Lemma \ref{lemma::trBest} is included in \citet{RZ15} and Lemma~\ref{lemma::Tclaim1} is directly from the existing Lemma in \citet{RZ15} by applying to the nodewise regression setting (i.e., Gaussian model).
Lemma \ref{lem:r-sqaured} shows the asymptotic property of R-squared metric based on Gaussian model.
Lemma \ref{lem:eig} is used to obtain the lower bound of R-squared values for errors-in-variables regression 
estimate in high-dimensional settings based on subgaussian model.

\appendix

\section{Additional estimation procedure}
\label{sec:add:proc}
The estimated $\hat{\Theta}$ and $\hat{\Phi}$ obtained from Algorithm 1 in the main paper are not necessarily positive-semidefinite.    
To obtain positive-semidefinite estimated precision matrices, 
one can consider the following additional estimation procedure.\\
\noindent
\textbf{Algorithm 3: Obtain $\hat{\Theta}_+$ with an input $\hat{\Theta}$}\\
Consider the case in which $\hat\Theta$ is not positive-semidefinite.
Since $\hat\Theta$ is symmetric, 
there exists an orthogonal matrix $U$ and a diagonal matrix $D = \mathrm{diag}(\lambda_1, \cdots, \lambda_m)$ such that
$\hat\Theta = U D U^T$, where $\lambda_1 \le \cdots \le \lambda_m$. 
Since $\hat\Theta$ is not positive-semidefinite, $\lambda_1 < 0$.
Let
$
\hat\Theta_+ = U D_+ U^T,
$
where $D_+ :=  \mathrm{diag}(\lambda_1 \vee \epsilon, \cdots, \lambda_m \vee \epsilon)$ for some positive 
constant $0< \epsilon \le -\lambda_1$.
Then $\hat\Theta_+$ is positive-semidefinite and satisfies with high probability    
\begin{eqnarray*}
\|\hat\Theta_+ - \Theta\|_2 
\le  \|\hat\Theta - \hat\Theta_+\|_2 + \|\hat\Theta-\Theta\|_2 
\le  -2\lambda_1 + \|\hat\Theta-\Theta\|_2
\le 3\|\hat\Theta-\Theta\|_2.
\end{eqnarray*}
The estimator $\hat\Theta_+$ is positive-semidefinite and has error bound at the same order with $\hat\Theta$ as in Theorem \ref{coro::Theta}.  
In practice, we set $\epsilon=-\lambda_1$.

\section{Theoretical property}
\label{Sec:gau}

\subsection{Assumption} 
In this subsection, we define some notations and assumptions which facilitate the theoretical properties of the proposed estimator.
We will define some parameters related to the restricted and
sparse eigenvalue conditions. 
We first state Definitions~\ref{def:memory}-~\ref{def::lowRE}. 
For more details of these, see \citet{RZ15}.
\begin{definition}
\label{def:memory}
\textnormal{\bf (Restricted eigenvalue condition $\RE(s_0, k_0, A)$)}
Let $1 \leq s_0 \leq m$ and  $k_0$ be a positive number.
The $m \times m$ matrix $A$ satisfies $\RE(s_0, k_0, A)$
 condition with parameter $K(s_0, k_0, A)$ if for any $\upsilon
 \not=0$,
\beq
\inv{K(s_0, k_0, A)} := 
\min_{\stackrel{J \subseteq \{1, \ldots, p\},}{|J| \leq s_0}}
\min_{\norm{\upsilon_{J^c}}_1 \leq k_0 \norm{\upsilon_{J}}_1}
\; \;  \frac{\norm{A \upsilon}_2}{\norm{\upsilon_{J}}_2} > 0,
\eeq
where $\upsilon_{J}$ represents the subvector of $\upsilon \in \R^m$
confined to a subset $J$ of $\{1, \ldots, m\}$.
\end{definition}

We also consider the following variation of the baseline $\RE$ condition.
\begin{definition}{\textnormal{(Lower-$\RE$ condition)~\citep{LW12,RZ15}}}
\label{def::lowRE}
The $m \times m$ matrix $\Gamma$ satisfies a Lower-$\RE$ condition with curvature
$\alpha >0$ and tolerance $\tau > 0$ if 
\bens
\theta^T \Gamma \theta \ge 
\alpha \twonorm{\theta}^2 - \tau \|\theta\|_1^2, \; \;  \forall \theta \in \R^m.
\eens
\end{definition}

\begin{definition}{\textnormal{(Upper-$\RE$ condition)}}
\label{def::upRE}
The $m \times m$ matrix  $\Gamma$ satisfies an upper-$\RE$ condition with curvature
$\bar\alpha >0$ and tolerance $\tau > 0$ if 
\bens
\theta^T \Gamma \theta \le \bar\alpha \twonorm{\theta}^2 + \tau
\|\theta\|_1^2, \; \;  \forall \theta \in \R^m.
\eens
\end{definition}

\begin{definition}
\label{def::sparse-eigen}
Define the largest and smallest
$d$-sparse eigenvalue of a $m \times m$ matrix $A$: for $d<m$,
\ben
\label{eq::eigen-Sigma}
\rho_{\max}(d, A) & := &
\max_{t \not= 0; \|t\|_0 \le d} \; \;\frac{\shtwonorm{A
  t}^2}{\twonorm{t}^2}, \\
\text {and } \; 
\rho_{\min}(d, A) & := & 
\min_{t \not= 0; \|t\|_0 \le d} \; \;\frac{\shtwonorm{A t}^2}{\twonorm{t}^2}.
\een
\end{definition}
Recall that we consider the inverse covariance matrices $\Theta=A^{-1}$ and $\Omega=B^{-1}$ in
the additive model of $X = X_0 + W$ such that
\ben
\label{eq::dataplus}
\mvec{X} \sim \N(0, \Sigma) \; \; \text{ where } \; \; 
\Sigma = A \oplus B := A \otimes I_n + I_m \otimes B.
\een
We have also stated the subgaussian analog of~\eqref{eq::dataplus} in~\eqref{eq::addmodel}.
Throughout this supplement, 
we use the Gaussian ensemble design for the main results, while noting that 
the more general subgaussian random design model as in~\eqref{eq::addmodel} is allowed in order for 
some of key concentration of measure bounds to go through.

Let $s_0 \ge 1$ be the largest integer chosen such that the following inequality holds:
\ben
 \label{eq::s0cond}
\sqrt{s_0} \vp(s_0) \le \frac{\lambda _{\min}(A)}{32 C}\sqrt{\frac{n }{\log m}}, \quad
\; \; \vp(s_0) := \rho_{\max}(s_0, A)+\tau_B,
\een
where  $\tau_B = \tr(B)/n$ and $C$ is to be defined.
{
Denote by
\ben
\label{eq::defineM}
M_A = \frac{64 C \vp(s_0)}{\lambda_{\min}(A)} \ge 64 C.
\een}
We use the expression 
$\tau := (\lambda_{\min}(A)-\alpha)/{s_0}, \; \; \text{where} \; \; \alpha =
5\lambda_{\min}(A)/8.$

\subsection{EIV regression}
We first review theoretical properties of the EIV regression estimator in the following
theorem on the subgaussian model~\eqref{eq::addmodel}, which is directly from \cite{RZ15}.

\begin{theorem}{\textnormal{(\bf{Estimation for the Lasso-type estimator})}}[Theorem 3 of \citet{RZ15}]
\label{thm::lasso}
Suppose $n=\Omega(\log m)$ and $n \le (\V/e) m \log m$, 
where $\V$ is a constant which depends on $\lambda_{\min}(A)$,
$\rho_{\max}(s_0, A)$ and $\tr(B)/n$. 
Suppose $m$ is sufficiently large. Suppose (A1)-(A4) hold.
Consider the EIV regression model \eqref{eq::oby} and \ref{eq::addmodel} as defined in the main paper
with independent random matrices $X_0, W$ as in~\eqref{eq::addmodel} and 
$\norm{\e_{j}}_{\psi_2} \leq M_{\e}$.
Let $C_0, c' > 0$ be  some absolute constants. Let $D_2 := 2(\twonorm{A} + \twonorm{B})$.
Suppose that $c' K^4 \le 1$ and
\ben
\label{eq::trBLasso}
\quad
r(B) := \frac{\tr(B)}{\twonorm{B}} & \ge & 16c' K^4 \frac{n}{\log m}
\log \frac{\V m \log m }{n}.
\een
Let $b_0, \phi$ be numbers which satisfy
\ben
\label{eq::snrcond}
\frac{M^2_{\e}}{K^2 b_0^2}   \le \phi  \le 1.
\een
Assume that the sparsity of $\beta^*$ satisfies for some $0 < \phi \le
1$
\ben
\label{eq::dlasso}
&& d:= \abs{\supp(\beta^*)} \le 
\frac{c' \phi K^4}{40 M_+^2} \frac{n}{\log  m}< n/2,\\
&&\mathrm{where}\  M_+ = \frac{32C\varpi(s_0+1)}{\lambda_{\min}(A)}
\een
for $\varpi(s_0+1) = \rho_{\max}(s_0+1,A) + \tau_B$.
Let $\hat\beta$ be an optimal solution to the EIV regression  with 
\ben
\label{eq::psijune}
&& \lambda \ge 4 \psi \sqrt{\frac{\log m}{n}} \; \; \text{ where } \;\;
\psi  := C_0 D_2 K \left(K \twonorm{\beta^*}+ M_{\e}\right).
\een
Then for any $d$-sparse vectors $\beta^* \in \R^m$, such that
$\phi b_0^2 \le \twonorm{\beta^*}^2 \le b_0^2$, we have with probability  at least $1- 16/m^3$,
\bens
\twonorm{\hat{\beta} -\beta^*} \leq \frac{20}{\alpha}  \lambda \sqrt{d} \; \;
\text{ and } \; \norm{\hat{\beta} -\beta^*}_1 \leq \frac{80}{\alpha}
\lambda d.
\eens
\end{theorem}

\subsection{Multiple EIV regressions for the Gaussian Random Ensembles}

First, we define some constants: 
\ben
\label{eq::defineD0}
& & D_0 = \sqrt{\tau_B} + \sqrt{a_{\max}}, \; \; \;
 D_0' =
\twonorm{B}^{1/2} + \sqrt{a_{\max}}, \; \; \tau_B^+ :=
(\tau_B^{+/2})^2, \\
\label{eq::defineDtau}
&& 
 \tau_B^{+/2}  
:=   \sqrt{\tau_B} +
\frac{D_{\ora}}{\sqrt{m}}, \; \;
D_{\ora} \;  = \; 2\left(\sqrt{\twonorm{A}} + \sqrt{\twonorm{B}}\right).
\een
The following theorem \ref{thm::nodewise} shows oracle inequalities of the nodewise regressions, which is analogous to Theorem 6 of \citet{RZ15}, which is adapted to the Gaussian random ensembles when generating data \eqref{eq::dataplus}.
Theorem \ref{thm::nodewise} will be used for the proof of Theorem 1 as in Section \ref{sec:thm1}.

\begin{theorem}
\label{thm::nodewise}
Consider the Kronecker sum model as in \eqref{eq::dataplus}.
Suppose all conditions in Theorem~\ref{thm::lasso} hold, except that
we drop \eqref{eq::snrcond} and replace~\eqref{eq::psijune} with
\begin{equation}
\label{eq::psijune15}
 \lambda^{(i)} \ge 4 \psi_i\sqrt{\frac{\log m}{n}} \; \; \text{ where } \;\;
\psi_i := C_0 D_0' K^2  \left(\tau_B^{+/2} \twonorm{\beta^{i}} + \sigma_{V_i}  \right),
\end{equation}
where $\sigma_{V_i}^2 := A_{ii}-A_{i,-i} A_{-i,-i}^{-1} A_{-i,i}$. 
Suppose that for $0< \phi \le 1$ and $C_A := \inv{160 M_{+}^2}$,
\ben
\label{eq::doracle}
&& d:= \max_{1 \le i \le m}\abs{\supp(\beta^{i})} \le C_A \frac{n}{\log m} \left\{c' c'' D_{\phi}
  \wedge 8 \right\},  \;\;\text{ where }  \\
\nonumber
&&    
c'' := \frac{\twonorm{B} + a_{\max}}{\varpi(s_0+1)^2} , \; 
 D_{\phi}
 = \frac{K^2 M^2_{\e}}{b_0^2}  + K^4 \tau_B^{+}  \phi,   
\een
and $c', \phi, b_0$ as defined in
Theorem~\ref{thm::lasso}. 
We obtain $m$ vectors of  $\hat{\beta}^i, i=1, \ldots, m$ 
by solving \eqref{eq::hatTheta} in the main paper with $\lambda = \lambda^{(i)}$,
\bens
 \; \hat\Gamma^{(i)} & = &\onen X_{\minus i}^T  X_{\minus i} - \hat\tau_B I_{m-1} 
\; \text{ and } \; \hat\gamma^{(i)}
 \; = \; \onen  X_{\minus i}^T X_i, 
\eens
for each $i$. 
Let $b_1 := b_0 \sqrt{d}$.
Then for all $d$-sparse vectors  $\beta^{i} \in \R^{m-1}, i=1, \ldots, m$, such that
$\phi b_0^2 \le \twonorm{\beta^{i}}^2 \le b_0^2$, 
we have with probability  at least $1- 16/m^2$,
\bens
\twonorm{\hat{\beta}^{i} -\beta^{i}} \leq \frac{20}{\alpha}
\lambda^{(i)}  \sqrt{d} \; \;
\text{ and } \; \norm{\hat{\beta}^{i} -\beta^{i}}_1 \leq \frac{80}{\alpha}
\lambda^{(i)} d.
\eens
\end{theorem}

\section{Main Lemmas}
\label{proof:main}
The following lemmas are essential to prove the main theorems.
Lemmas \ref{lemma::AD} and \ref{lemma::low-noise} are analogous to Theorem 26 and Corollary 13 of \citet{RZ15}, respectively.
We emphasize that Lemma \ref{lemma::AD} holds for a subgaussian model while Lemma \ref{lemma::low-noise} assumes Gaussian model.

\begin{lemma}
\label{lemma::AD}
Suppose (A1) holds. 
Let $A_{m \times m}$ and $B_{n \times n}$ be positive definite 
covariance matrices.  Let $Z$ and $X$ be $n \times m$ random matrices as defined in Theorem~\ref{thm::lasso}.
Let
\ben
\label{eq::defineAD}
\Delta :=
 \hat\Gamma_{A} -A := \onen X^TX - \hat\tau_B I_{m} -A.
\een
Let $D_1 :=   \frac{\fnorm{A}}{\sqrt{m}} +
\frac{\fnorm{B}}{\sqrt{n}}$,  $\vp = \twonorm{B}+  a_{\max}$, and $\psi_0 = \vp C_0 K^2/\sqrt{c'}$.   
Then with probability at least $1- 6 /m^2 - 6/m^3$,
\bens
\maxnorm{\Delta} 
 & \le &   12 C K^2 \vp \sqrt{\frac{\log  m}{n}}  = O\left(\psi_0\sqrt{\frac{\log  m}{n}} \right),
\eens
where     
$C_0$ is appropriately chosen and $c'$ sufficiently small.
\end{lemma}

\begin{lemma}
\label{lemma::low-noise}
Suppose (A1) holds. 
Suppose that $m \ge 16$ and
$\frac{\fnorm{B}^2}{\twonorm{B}^2} \ge \log m.$   
Let $\hat\Gamma^{(i)}$ and $\hat\gamma^{(i)}$ be as defined in Algorithm 1 in the main paper.
Let $D_0, D_0', D_{\ora}$, and 
$\tau_B^{+/2}$  be as defined in~\eqref{eq::defineD0} and~\eqref{eq::defineDtau}.
On event $\B_0$, we have for   
some absolute constant $C_0$, for all $i$,
\ben
\norm{\hat\gamma^{(i)}
 - \hat\Gamma^{(i)} \beta^{i}}_{\infty}
& \le & 
\label{eq::psioracle} 
\psi_i \sqrt{\frac{\log m}{n}},
\een
where 
\bens
\label{eq::psi1}
\psi_i := 
C_0 K^2 D_0'\left(\tau_B^{+/2} \twonorm{\beta^{i}} + \sigma_{V_i}  \right) \le C_0 K^2 D_0' (\tau_B^{+/2} \kappa(A) + \sigma_{V} ),
\eens
where recall $D_0'=\twonorm{B}^{1/2} + a_{\max}^{1/2}$ and $\sigma_V := \max_{j} \sigma_{V_j}$. 
Then, $\prob{\B_0} \ge 1- 16/m^2$.
\end{lemma}

\section{Proof of Theorem 1}
\label{sec:thm1}
Lemmas~\ref{lemma::AD} and~\ref{lemma::low-noise} show that the following conditions hold with probability at least  $1- 20 /m^2 - 6/m^3$:
\bens
\norm{\hat\gamma^{(i)} 
- \hat\Gamma^{(i)} \beta^{i}}_{\infty} \le C_\psi \sqrt{\frac{\log m}{n}},
\eens
\bens
\maxnorm{ \hat\Gamma_{A} -A } = \maxnorm{\onen X^TX - \hat\tau_B I_{m}
  -A} = O\left(C_\psi \sqrt{\frac{\log m}{n}}\right),
\eens
where $C_\psi := K^2 C_0 D_0'  \left(\tau_B^{+/2} \kappa(A)+ \sqrt{a_{\max}} \right)$.
Next we see that the lower-$\RE$ condition holds with  curvature $\alpha = 5\lambda_{\min}(A)/8$ and
tolerance $\tau =\frac{\lambda_{\min}(A)}{2s_0}$ uniformly over the
matrices $\hat\Gamma^{(i)}$ with scaling as required
 by Corollary 5 in~\cite{LW12}, namely,
\ben
\label{eq::taumain}
\sqrt{d} \tau \le \min \left\{\frac{\alpha}{32 \sqrt{d}}, \frac{\lambda^{(i)}}{4b_0} \right\}.
\een
The theorem follows from Corollary 5 in~\cite{LW12}, so long as 
we can show that condition~\eqref{eq::taumain} holds for $\lambda^{(i)} \ge 4 \psi_i\sqrt{\frac{\log m}{n}}$, 
where  the parameter $\psi_i$
is as defined in~\eqref{eq::psijune15}.
Condition \eqref{eq::taumain} can be easily checked using the definition of $s_0$ as in \eqref{eq::s0cond} and the lower bound of $\lambda^{(i)}$.
Hence, by Theorem \ref{thm::nodewise} and Corollary 5 in~\cite{LW12},  we have     
\begin{eqnarray*}
\twonorm{\hat\Theta - \Theta}   
&=& O_p\left( \frac{K^2(A)}{\lambda_{\min}^2(A)}  d \max_i \lambda^{(i)} \right).
\end{eqnarray*}

This completes the proof. \qed

\section{Consistency of $\hat{\Omega}$}
\label{sec:cons:omega}
Theorem \ref{coro::Omega} shows the consistency of $\hat\Omega$ in the operator norm. 
Using the similar argument presented in the proof of Theorem 1 with appropriate modifications, we can show consistency results of
$\Omega = B^{-1}$.  

\begin{theorem}
\label{coro::Omega}
Suppose the columns of $\Omega$ is $d_{\Omega}$-sparse, and suppose the
condition number $\kappa(\Omega)$ is finite.
Suppose all conditions of Theorem \ref{coro::Theta} in the main paper hold by changing $m$ and $n$ each other, and replacing $d$ with $d_{\Omega}$.
Then with probability at least $1- 26 /m^2$, 
\bens
\twonorm{\hat\Omega - \Omega}  =  O_p\left( \frac{K^2(B)}{\lambda_{\min}^2(B)}  d_{\Omega} \max_{i \le n} \lambda^{(i)} \right).
\eens
\end{theorem}

\section{Proof of lemmas}
\label{sec:pre}

The large deviation bounds in Lemmas \ref{lemma::trBest} and \ref{lemma::Tclaim1} are the key results in proving Lemmas \ref{lemma::AD} and \ref{lemma::low-noise}. 
Let $C_0$ be an absolute constant appropriately chosen.
We first prove Lemma~\ref{lemma::low-noise} followed by Lemma~\ref{lemma::AD}. 

\subsection{Preliminary Results}
\label{sec::lownoiseproof}
Let $Z$ be a Gaussian ensemble such that $Z_{j k} \sim \N(0, 1)$ for all $j, k$.
First we note that for a mean zero normal random variable $V_{0, j}$ with variance 
$\sigma_{V_j}^2$, $V_{0, j}/\sigma_{V_j} \sim Z_{ij}$, for which it holds that $\norm{Z_{ij}}_{\psi_2} = K$.
Thus we have 
\bens
\norm{V_{0, j}/\sigma_{V_j}}_{\psi_2} = K \; \; \text{and thus}\;\;
\norm{V_{0, j}}_{\psi_2} = \sigma_{V_j} K  =: M_{V_j}
\eens
which reflects the relative strength of the noise in $V_0$ relative to $Z_{ij}$  
and $M_V := \max_j  M_{V_j} = K \max_j \sigma_{V_j}$ is the upper bound on $\psi_2$ norm for
$V_{0,j},  j = 1, \ldots, m$.
 
Throughout this section, we denote by:
\bens
r_{m,n} =  C_0 K^2 \sqrt{\frac{\log m}{n}} \; \; \text{ and } \; \
r_{m,m} =  2 C_0 K^2 \sqrt{\frac{\log m}{mn}}.
\eens
We first define some events $\B_4, \B_5, \B_6, \B_{10}$ which are
adapted from Lemmas 5 and 11 of~\cite{RZ15}.

Denote by $\B_0 := B_4 \cap \B_5 \cap \B_6 \cap B_{10}$, which we use throughout
this paper.  Denote by $\bar\beta^j \in \R^m$ the zero-extended $\beta^{j}$ in $\R^m$ such that 
$\bar\beta^j_i = \beta^{j}_i \; \; \forall i \ne j$ and $\bar\beta^j_j = 0$.

\begin{lemma} [Lemma 11 of \citet{RZ15}]
\label{lemma::Tclaim1} 
Assume that the stable rank of $B$, $\fnorm{B}^2/\twonorm{B}^2 \ge \log m$.
Let $Z, X_0$ and $W$ as defined in Theorem~\ref{thm::lasso}.
Let $Z_0, Z_1$ and  $Z_2$ be independent copies of $Z$.
Let $V_j^T \sim Y_j  \sigma_{V_j}$ where $Y_j := e_j^T Z_0^T$.
Denote by $\B_4$ the event such that 
\bens
\text{ for every } \; \; j, \quad
\onen \norm{A^{\half} Z_1^T V_j}_{\infty}
& \le & r_{m,n} \sigma_{V_j} a_{\max}^{1/2} \\
 \; \text{ and } \; 
\onen \norm{ Z_2^T B^{\half} V_j}_{\infty}
& \le &   r_{m,n} \sigma_{V_j}\sqrt{\tau_B}.
 \eens
Then $\prob{\B_4} \ge  1 - 4/m^2$.
Moreover, denote by $\B_5$ the  event such that 
\bens
\text{ for every } \; \; j, \quad
\onen \norm{(Z^T B Z- \tr(B) I_{m}) \bar\beta^j}_{\infty} & \le &  
r_{m,n}  \twonorm{\beta^{j}}  \frac{\fnorm{B}}{\sqrt{n}} \\
\text{and} \;\;\onen \norm{X_0^T W \bar\beta^j}_{\infty}
& \le &   r_{m,n} \twonorm{\beta^{j}}  \sqrt{\tau_B} a^{1/2}_{\max}.
\eens
Then $\prob{\B_5} \ge 1 - 4/m^2$.

Finally, denote by $\B_{10}$ the event such that
\bens
\onen \norm{(Z^T B Z- \tr(B) I_{m})}_{\max} & \le &  r_{m,n} \frac{\fnorm{B}}{\sqrt{n}} \\
\text{and} \;\; \onen \norm{X_0^T W}_{\max}
& \le &  r_{m,n} \sqrt{\tau_B} a^{1/2}_{\max}.
\eens
Then $\prob{\B_{10}} \ge 1 -  4/m^2$.
\end{lemma}

Next we state the following result from~\cite{RZ15}. 
\begin{lemma}[Lemma 5 of \citet{RZ15}]
\label{lemma::trBest}
Let $m \ge 2$. Let $X$ be defined as in~\eqref{eq::addmodel}. 
Suppose that $n \vee (r(A)  r(B)) > \log m$.
Denote by $\B_6$ the event such that 
\bens
\abs{\hat\tau_B - \tau_B} & \le &  
2 C_0 K^2 \sqrt{\frac{\log m}{m n}}
\left(\frac{\fnorm{A}  }{\sqrt{m}} +\frac{\fnorm{B}}{\sqrt{n}}
\right) =:  D_1 r_{m,m} \\ 
\text{ where} \; \; D_1 & := &  \frac{\fnorm{A}}{\sqrt{m}} +
\frac{\fnorm{B}}{\sqrt{n}} \; 
\text{ and } \; \;   r_{m,m} = 2 K^2 C_0 \sqrt{\frac{\log m}{m n}}.
\eens 
Then $\prob{\B_6} \ge 1-\frac{3}{m^3}$.
If we replace $\sqrt{\log m}$ with $\log m$ in the definition of event
$\B_6$, then we can drop the condition on $n$ or $r(A)r(B) =
\frac{\tr(A)}{\twonorm{A}} \frac{\tr(B)}{\twonorm{B}}$ to achieve 
the same bound on event $\B_6$.
\end{lemma}

\subsection{Proof of Lemma~\ref{lemma::low-noise}}
Clearly the condition on the stable rank of $B$ guarantees that 
$$n \ge r(B) = \frac{\tr(B)}{\twonorm{B}} =
\frac{\tr(B)\twonorm{B}}{\twonorm{B}^2}\ge \fnorm{B}^2 /\twonorm{B}^2 
 \ge \log m.$$
Thus the conditions in
Lemmas \ref{lemma::trBest} and \ref{lemma::Tclaim1}  hold.   
A careful examination of the proof for Theorem~1 in~\cite{RZ15} shows
that the key new component is in analyzing the following term
$\hat\gamma^{(i)}$ for all $i$.
First notice that for all $i$,
\[
n \hat\gamma^{(i)} 
 =   {X_{\minus i}^T X_i } \\
 =  
(X_{0, \minus i}^T +   W_{\minus i}^T)(X_{0, \minus i} \beta^{i} + \ve_i), 
\]
where   
$\ve_i :=  V_{0, i} + W_i$, and
\bens
\hat\Gamma^{(i)} \beta^{i}
& = &{\frac{1}{n}(X_{\minus i}^T X_{\minus i} - \hat\tr(B) I_{m-1}) \beta^{i}} \\
& = & 
\inv{n} (X_{0, \minus i}^T X_{0, \minus i}
 + W_{\minus i}^T X_{0,\minus i} + 
X_{0, \minus i}^T W_{\minus i} 
+ W_{\minus i}^T W_{\minus i} - \hat\tr(B) I_{m-1})\beta^{i}
\eens
Thus 
\bens
&&\norm{\hat\gamma^{(i)} - \hat\Gamma^{(i)} \beta^{i}}_{\infty} 
 \le 
\inv{n}\norm{ X_{0, \minus i}^T \ve_i +  W_{\minus i}^T \ve_i 
- (X_{0, \minus i}^T W_{\minus i} + W_{\minus i}^T W_{\minus i} - \hat\tr(B) I_{m-1})\beta^{i}}_{\infty}\\
 & \le & 
\inv{n}\norm{ X_{0, \minus i}^T \ve_i +  W_{\minus i}^T \ve_i }_{\infty} 
+ \inv{n}\norm{(W^T W-  \hat\tr(B) I_{m}) \bar\beta^i}_{\infty} 
+\norm{\inv{n} X_0^T W \bar\beta^i}_{\infty}  \\
 & \le & 
\inv{n}\norm{ X_{0, \minus i}^T \ve_i +  W_{\minus i}^T \ve_i }_{\infty} 
+ \inv{n}\norm{(Z^T B  Z- \tr(B) I_{m}) \bar\beta^i}_{\infty} + \inv{n} \norm{X_0^T W
     \bar\beta^i}_{\infty} \\
&& + \inv{n} \abs{\hat\tr(B) - \tr(B)} \norm{\beta^{i}}_{\infty} =:  U_1+  U_2 + U_3 + U_4.
\eens

By  Lemma \ref{lemma::Tclaim1},      
on event $\B_{10}$,
\[
\max_{i\not=j} \onen \ip{X_{0, j}, W_i}  \le 
 \frac{1}{n}\maxnorm{X_0^T W} \le C_0 K^2 
\sqrt{\tau_B} \sqrt{ a_{\max}} \sqrt{\frac{\log m}{n}} 
= \sqrt{\tau_B} \sqrt{ a_{\max}} r_{m,n}, 
\]
\[
 \max_{i \not= j} \onen \ip{W_i, W_j}  \le 
 C_0 K^2 \inv{\sqrt{n}}
\fnorm{B} \sqrt{\frac{\log m}{n}} = \inv{\sqrt{n}}\fnorm{B} r_{m,n}.
\]
On event $\B_4$,  we have  for every $i$,
\bens
\max_{j \not=i} \onen \ip{X_{0, j}, V_{0,i}} &  \le & 
C_0 K^2 \sigma_{V_i} 
\sqrt{ a_{\max}} \sqrt{\frac{\log m}{n}} = 
\sigma_{V_i} r_{m,n}\sqrt{ a_{\max}}  \\
\text{ and } \; \; 
\max_{j \not= i} \onen \ip{W_j, V_{0, i}} & \le & 
 C_0 K^2 \sigma_{V_i} \sqrt{\tau_B}\sqrt{\frac{\log m}{n}}
= r_{m,n}\sigma_{V_i} \sqrt{\tau_B},
\eens
where       
$C_0$ is adjusted so
that the error probability hold for $D_0 = \sqrt{\tau_B} + a_{\max}^{1/2}$.
On event $\B_5$ for  $D'_0 :=\sqrt{\twonorm{B}} + a_{\max}^{1/2}$, for all $j$,
\bens
\lefteqn{U_2 +  U_3 = 
\onen \norm{(Z^T B Z- \tr(B) I_{m}) \bar\beta^j}_{\infty} +
\onen \norm{X_0^T W \bar \beta^j}_{\infty}} \\
& \le & r_{m,n} \twonorm{\beta^{j}}
\left(\frac{\fnorm{B}}{\sqrt{n}} + \sqrt{\tau_B}
  a^{1/2}_{\max}\right) \le r_{m,n} \twonorm{\beta^{j}}  \tau_B^{1/2} D_0',
\eens
where recall $\fnorm{B} \le \sqrt{\tr(B)}\twonorm{B}^{1/2}$.
Denote by $\B_0 := \B_4 \cap \B_5 \cap \B_6 \cap \B_{10}$.
Suppose that event $\B_0$ holds.
By Lemmas \ref{lemma::trBest} and \ref{lemma::Tclaim1},
under (A1) and $D_1$ defined therein, for all $i$,
\ben
\nonumber
\norm{\hat\gamma^{(i)} - \hat\Gamma^{(i)} \beta^{i}}_{\infty} 
& \le & 
\nonumber
 r_{m,n} \sigma_{V_i} D_0 + D_0' \tau_B^{1/2} r_{m,n} \twonorm{\beta^{i}} 
+ \onen \abs{\hat\tr(B) - \tr(B)} \norm{\beta^{i}}_{\infty} \\
& \le & 
\nonumber
D_0 \sigma_{V_i} r_{m,n} + D_0' \tau_B^{1/2} \twonorm{\beta^{i}} r_{m,n}  +
D_1 \norm{\beta^{i}}_{\infty} r_{m,m} \\
& \le & 
\label{eq::oracle}
D_0 \sigma_{V_i} r_{m,n} + D_0' \tau_B^{1/2} \twonorm{\beta^{i}} r_{m,n}
+ 2 D_1 \inv{\sqrt{m}} \norm{\beta^{i}}_{\infty} r_{m,n}
\een
By~\eqref{eq::oracle} and that fact that 
\[
2 D_1 := 2\left(\frac{\fnorm{A} }{\sqrt{m}} +\frac{\fnorm{B}  }{\sqrt{n}}
\right)\le  2(\twonorm{A}^{1/2} + \twonorm{B}^{1/2}) (\sqrt{\tau_A} +
\sqrt{\tau_B}) \le 
D_{\ora} D_0',
\]
we have  on $\B_0$ and under (A1), for all $i$
\bens
\label{eq::oracleII}
\norm{\hat\gamma^{(i)} - \hat\Gamma^{(i)} \beta^{i}}_{\infty}
& \le &
D_0' \tau_B^{1/2} \twonorm{\beta^{i}} r_{m,n} 
+ 2D_1\inv{\sqrt{m}} \norm{\beta^{i}}_{\infty} r_{m,n} + D_0 \sigma_{V_i} r_{m,n} \\
& \le &
D_0' \twonorm{\beta^{i}} r_{m,n}
\left( \tau_B^{1/2}  +  \frac{D_{\ora}}{\sqrt{m}} \right) + D_0  \sigma_{V_i} r_{m,n}  \\
& \le &
D_0' \left( \tau_B^{1/2}  +  \frac{D_{\ora}}{\sqrt{m}} \right) \twonorm{\beta^{i}} r_{m,n}
+ D_0  \sigma_{V_i} r_{m,n} \\
&\le &
{D_0' \left(\tau_B^{+/2}  \twonorm{\beta^{i}} 
+   \sigma_{V_i} \right) r_{m,n}.}
\eens
Hence the lemma holds for $m \ge 16$ and 
$\psi_i = C_0 D_0' K^2 \left(\tau_B^{+/2} \twonorm{\beta^{i}} + \sigma_{V_i} \right)$.
Finally, we have by the union bound, $\prob{\B_0} \ge 1 -  16/m^2$.
This completes the proof. \qed

\subsection{Proof of Lemma~\ref{lemma::AD}}
 Lemma~\ref{lemma::AD} follows from the proof of Theorem 26
of~\cite{RZ15}.     
Hence we only provide a proof sketch.

Recall the following for $X_0 = Z_1 A^{1/2}$,
\bens
\lefteqn{
\Delta := \hat\Gamma_{A} -A := \onen X^TX - \onen \hat\tr(B) I_{m} -A} \\
& = & (\onen X_0^T X_0 -A)+  
 \onen \big(W^T X_0 + X_0^T W\big) + \onen \big(W^T W  - \hat\tr(B) I_{m}\big).
\eens
First notice that 
\bens
&&\lefteqn{
\maxnorm{\hat\Gamma_A -A}} \\
& \le &  
\maxnorm{\onen X_0^T X_0 - A} + 
\maxnorm{\onen (W^T X_0 + X_0^T W)} + 
\maxnorm{\onen W^T W - \frac{\hat\tr(B)}{n} I_{m}}\\
& \le &  
\maxnorm{A^{1/2}\onen Z_1^T Z_1 A^{1/2} - A} +
\maxnorm{\onen(W^T X_0 + X_0^T W)} \\
&& +
\maxnorm{\onen Z_2^T B Z_2 - \tau_B I_{m}} + 
\onen \maxnorm{\hat\tr(B) - \tr(B)} \\
&=:& I + II + III + IV.
\eens
Denote by $\B_{10}$ the event such that 
\bens
\inv{n} 
\maxnorm{X_0^T W} & \le &  C_0 K^2  \sqrt{\tau_B a_{\max} }
\sqrt{\frac{\log  m}{n}} \\
\maxnorm{\onen Z^T B Z - \tr(B) I_m/n} & = & \onen \maxnorm{W^T W - \tr(B)
  I_m} \\
& \le &  C_0 K^2  \frac{\fnorm{B}}{\sqrt{n}}  \sqrt{\frac{\log  m}{n}} \le  C_0 K^2  \twonorm{B} \sqrt{\frac{\log  m}{n}}.
\eens
Then, $\B_{10}$ holds with probability at least $1 - 2 /m^2$.  See Lemma \ref{lemma::Tclaim1} of~\cite{RZ15}.

Denote by event $\B_3$ the event such that 
\bens
\onen \maxnorm{X_0^T X_0 - A} \le 4 C \ve a_{\max}, 
\eens
where $\ve = K^2 \sqrt{\frac{\log  m}{n}} < 1/C$;
Then $\prob{\B_3} \ge 1 - 2 /m^2$ so long as $n \ge c' K^4 \log (3em/\ve)/\ve^2$, which in turn holds so
long as $c'$ is small enough; See Corollary 42 of~\cite{RZ15}.

Putting all together, we have under $\B_{10} \cap \B_{3} \cap \B_{6}$
\bens
\maxnorm{\hat\Gamma_A -A} & \le &  8 C K^2 (a_{\max} + \twonorm{B}) \sqrt{\frac{\log  m}{n}} + 4 C_0 D_1 K^2 
\sqrt{\frac{\log  m}{m n}}.
\eens
This completes the proof. \qed

\section{R-squared analysis}
\label{sec:R}

Consider the following regressions as considered in \eqref{eq:reg:reg} and \eqref{eq:EIV:reg} in the main paper:
\begin{eqnarray*}
y &=& X \beta_1 + \epsilon, \quad \ X\ \rm{and} \ y  \ \rm{are \ observable,} \\
y &=& X_0 \beta_2 + \epsilon, \quad  X=X_0+W,  \quad  X\ \rm{and} \ y \ \rm{are \ observable}.
\end{eqnarray*}
We calculate the explanatory power of $X$ for $y$ using  
$R^2_X = \corr\left(X\hat{\beta}_1, y \right)^2$, where $\hat{\beta}_1$ is the Lasso or the ridge regression estimator.

Similarly, define the explanatory power of $X_0$ for $y$ using the EIV estimator $\hat{\beta}_2$:
\begin{eqnarray}
R^2_{X_0}  &=& \corr\left(X_0 \hat{\beta}_2, y \right)^2 = \frac{\cov(y, X_0 \hat{\beta}_2)^2}{\var(y)  \var(X_0 \hat{\beta}_2)}. \label{eq:R2}
\end{eqnarray}
The following Lemma \ref{lem:r-sqaured} implies that when the model follows EIV with fixed $m$, $R^2_{X_0}$ becomes larger than $R^2_X$ as $n$ increases, i.e.,
the proposed R-squared metric asymptotically choose a correct model between EIV and the regular regression model.
\begin{lemma}
\label{lem:r-sqaured}
Suppose EIV model as in (1) and (3) in the main paper with $E \epsilon^2=\sigma^2$ and $\E \epsilon =0$. 
Suppose $X_0 = Z_1 A^{1/2}$ and  $W = B^{1/2} Z_2$ for independent random matrices $Z_1$ and $Z_2$ with i.i.d. $N(0,1)$ entries,   for $1 \le i \le j \le m$.
Let $\hat{\beta}$ be the corrected Lasso (i.e. EIV) and $\hat{\beta}_L$ and $\hat{\beta}_R$ be the Lasso and the ridge regression estimator with regularization parameters $\lambda_L$ and $\lambda_R$ respectively, as defined in Section \ref{subsec:analy:R-sq} of the main paper.
Then as $n \to \infty$ with fixed $m$,  $\lambda_L \to 0$, and $\lambda_R/n \to 0$,  it holds that
$R^2_{X_0}  -  R^2_X \to c$ in probability
for some $c>0$; $c$ depends only on $A$, $B$, $\beta_2$, and $\sigma^2$.
\end{lemma}

\begin{proof}   
Throughout the proof,  for a sequence of numbers $\{a_n\}_{n \ge 1}$ and a number $b$, we use $a_n \probto b$ when $a_n$ converges to $b$ in probability. 
For a sequence of matrices $\{\Sigma_n\}_{n \ge 1}$ and a matrix $\Sigma$, we use $\Sigma_n  \probto \Sigma$ 
when every element of $\Sigma_n$ converges to the corresponding element of $\Sigma$ in probability, 
which implies that $\|\Sigma_n - \Sigma\|_F \probto 0$.
Note that 
\[
\tilde{\beta}_L = \argmin_{\beta \in \R^m}   \frac{1}{2n} \beta^T X^T X \beta - \frac{1}{n} y^T X \beta + \lambda \|\beta\|_1 := \argmin_{\beta \in \R^m}  g_n(\beta).
\]
Since $X_0 = Z_1 A^{1/2}$ and  $W = B^{1/2} Z_2$ for random matrices $Z_1$ and $Z_2$ having i.i.d. $N(0,1)$ entries,   for $1 \le i \le j \le m$,   
\begin{eqnarray*}
\frac{1}{n} X^T X  \probto   A+  \tau_B I_m\quad\text{ and} \quad
\frac{1}{n} y^T X  \probto  \beta_2^T A.
\end{eqnarray*}
Combining the above with $\lambda \to 0$,
\[
g_n(\beta) \probto g(\beta) :=  \frac{1}{2}  \beta^T\left( A + \tau_B I_m \right)\beta  -  \beta_2^T A \beta.
\]
Since $g_n(\beta)$ is convex,  it follows from \citet{Knight:00} that 
\[
\argmin_{\beta\in \R^m}  g_n(\beta) \probto \argmin_{\beta \in \R^m}  g(\beta).
\]
Hence $\tilde{\beta}_L  \probto  (A+\tau_B I_m)^{-1} A \beta_2$.

Note that 
\[
\tilde{\beta}_R =\left(X^TX + \lambda_R I_p\right)^{-1} X^Ty = \left(\frac{X^TX}{n} + \frac{\lambda_R}{n} I_m\right)^{-1} \frac{X^Ty}{n}.
\]
Since $\frac{X^TX}{n}  \probto  A+\tau_B I_m$, $\frac{\lambda_R}{n} \to 0$, and $\frac{X^Ty}{n} \probto A \beta_2$, we have 
\[
\tilde{\beta}_R \probto (A+\tau_B I_m)^{-1} A \beta_2.
\]
Since it holds that
\bens
&& \var(y) \probto \sigma^2 + \beta_2^T A \beta_2 \quad \text{and} \\
&& 
\var(X\tilde{\beta}_R), \  
\cov(y, X\tilde{\beta}_R) \probto \beta_2^T A (A+\tau_B I_m)^{-1} A \beta_2,
\eens
and 
\bens
\var(X_0 \hat{\beta}), \  \cov(y, X_0 \hat{\beta}) \probto \beta_2^T A \beta_2, 
\eens
we have
\bens
R^2_{X}   = \frac{\cov(y, X\tilde{\beta}_R)^2}{\var(y)  \var(X\tilde{\beta}_R)} \probto \frac{\beta_2^T A(A+\tau_B I_m)^{-1} A \beta_2}{\beta_2^T A \beta_2+\sigma^2} :=\tilde{r}^*,
\eens
and
\bens
R^2_{X_0}  = \frac{\cov(y, X_0 \hat{\beta})^2}{\var(y)  \var(X_0 \hat{\beta})} \probto \frac{\beta_2^T A \beta_2}{\beta_2^T A \beta_2+\sigma^2} := r^*.
\eens
Therefore, 
\bens
(r^* -\tilde{r}^*)(\beta_2^T A \beta_2+\sigma^2) & = & 
\beta_2^T \left( A- A(A+\tau_B I_m)^{-1} A \right) \beta_2  \\
&\ge& \|\beta_2\|_2^2 \frac{\lambda_{\min}(A) \tau_B} {\lambda_{\min}(A) + \tau_B},
\eens
where the last inequality follows from the fact that $A- A(A+\tau_B I_m)^{-1} A$ is positive definite since it has positive eigenvalues $\frac{\lambda_i(A) \tau_B} {\lambda_i(A) + \tau_B} (i=1,\cdots, m)$.
This completes the proof.
\end{proof}

In practice, $R^2_{X_0}$ defined in \eqref{eq:R2} should be estimated since $X_0$ is not observed.
Under the setting in Theorem \ref{coro::Theta} in the main paper,
we obtain a lower bound of $R^2_{X_0}$ by the function of $\hat{\beta}$, $\hat{A}_+$, and $\hat{B}_+$, 
where $\hat{A}_+:=\hat{\Theta}_+^{-1}$ and $\hat{B}_+ :=\hat{\Omega}_+^{-1}$.

Consider the subgaussian model~\eqref{eq::addmodel}.
Then, we have with high probability
\begin{eqnarray}
\var(X_0 \hat{\beta}_2) &=& \hat{\beta}_2^T \left(\frac{X_0^T X_0}{n}-   \frac{X_0^T 1_n 1_n^T X_0}{n^2}   \right) \hat{\beta}_2 \label{eq:ineq:var}\\
&=& \hat{\beta}_2^T A^{1/2} \left(\frac{Z_1^T Z_1}{n}-   \frac{Z_1^T 1_n 1_n^T Z_1}{n^2}   \right) A^{1/2} \hat{\beta}_2 \nonumber\\
&\le& \lambda_{\max} (A) \lambda_{\max} (Z_1^T Z_1/n) \|\hat{\beta}_2\|^2  \nonumber\\
&\le& \left(1+ \frac{m \log n}{4n}\right) \lambda_{\max} (A)  \|\hat{\beta}_2\|^2 \nonumber \\
&\le&2\left(1+ \frac{m \log n}{4n}\right) \lambda_{\max} (\hat{A}_+)  \|\hat{\beta}_2\|^2 \nonumber 
\end{eqnarray}
where the first inequality uses $\|I_n - 1_n 1_n^T/n\|_2 \le 1$,  the second inequality follows from Lemma \ref{lem:eig}, and the last inequality uses the fact that $\lambda_{\max} (A) \le 2 \lambda_{\max} (\hat{A}_+)$, which holds with high probability; with regard to the last claim, see the proof of Theorem \ref{coro::Theta}.

When  $1_n^T y=0$, the numerator of \eqref{eq:R2} is 
\begin{eqnarray*}
\cov(y, X_0 \hat{\beta}_2)^2 &=& \left(\frac{1}{n} \hat{\beta}_2^T X_0^T  y\right)^2.   
\end{eqnarray*}
Thus we have with probability $1-2n^{-\frac{cm}{8K^2}}$,  
\begin{eqnarray}
&&\frac{1}{n} \hat{\beta}_2^T X_0^T  y= \frac{1}{n} \hat{\beta}_2^T X^T  y - \frac{1}{n} \hat{\beta}_2^T W^T y
=\frac{1}{n} \hat{\beta}_2^T X^T  y - \frac{1}{n} \hat{\beta}_2^T Z_2^T B^{1/2}  y  \label{eq:var1} \\
&\ge& \frac{1}{n} \hat{\beta}_2^T X^T y - \frac{1}{n} \|\hat{\beta}_2\|_2 \|Z_2^T B^{1/2}  y\|_2  \nonumber \\
&\ge& 
 \frac{1}{n} \hat{\beta}_2^T X^T  y - \frac{1}{\sqrt{n}} \sqrt{1+ \frac{m \log n}{4n}}\|\hat{\beta}_2\|_2  \|y\|_2  \|B\|_2^{1/2} \nonumber \\
 &\ge&  \frac{1}{n} \hat{\beta}_2^T X^T  y - \frac{\sqrt{2}}{\sqrt{n}}  \sqrt{1+ \frac{m \log n}{4n}}\|\hat{\beta}_2\|_2  \|y\|_2  \|\hat{B}_+\|_2^{1/2},\nonumber
\end{eqnarray}
where we use Lemma \ref{lem:eig} and Theorem \ref{coro::Theta} in the main paper.
Hence, we obtain a lower bound of $\cov(y, X_0 \hat{\beta}_2)^2$. 
Together with \eqref{eq:ineq:var}, we obtain a lower bound of $R^2_{X_0}$, as expressed in \eqref{eq:R*} in the main paper.

\begin{lemma}
\label{lem:eig}
Let $Z$ be an $n \times m$ random matrix with independent entries $Z_{ij}$ satisfying $E(Z_{ij})=0$, $E(Z_{ij}^2)=1$, and $\|Z_{ij}\|_{\psi_2} \le K$. 
Suppose for the absolute constant $c>0$ defined in \eqref{def:absol}, the following holds:
\[
\frac{4n}{m \log n} \vee 32\left(\sqrt{\frac{m}{n}} + \frac{m}{n}\right) \frac{n}{m \log n} \le \frac{1}{3} n^{\frac{c}{8K^2}}.
\]
Then, we have with probability at least $1-2n^{-\frac{cm}{8K^2}}$, $\|\frac{1}{n}Z^T Z\|_2 \le 1+ \frac{m \log n}{4n}$.
\end{lemma}

\begin{proof}
By Theorem 31 of \citet{RZ15}, we have for any $u \in S^{m-1}$ and $t>0$,
\[
P\left(\left|u^T \frac{Z^T Z}{n} u - 1  \right| > t \right) \le 2\exp \left( -c \min\left(\frac{nt^2}{K^4}, \frac{nt}{K^2}   \right)\right).
\]
Consider the $\epsilon$-net of $S^{m-1}$. Let $S_{\epsilon}$ be the set of centers of the $\epsilon$-net. 
Then,  there exists an absolute constant $c>0$ such that for $t>K^2$, 
\begin{equation}
\label{def:absol}
P\left(\exists u \in S_{\epsilon} \mid  \left|u^T \frac{Z^T Z}{n} u - 1  \right| > t \right) \le 2\exp \left(m \log (3/\epsilon) -c \frac{nt}{K^2}\right).
\end{equation}

Set $\epsilon = n^{-c_1}$ and $t=\frac{m \log n}{4n}$, where $c_1>0$ satisfies 
\[
\frac{4n}{m \log n} \vee 32\left(\sqrt{\frac{m}{n}} + \frac{m}{n}\right) \frac{n}{m \log n} \le n^{c_1} \le \frac{1}{3} n^{\frac{c}{8K^2}}.
\]

Then, we have
\[
P\left(\left\|\frac{1}{n}Z^T Z - I_m\right\|_2 \ge \frac{m \log n}{4n} \right) \ge 1-2n^{-\frac{cm}{8K^2}}.
\]
This completes the proof.
\end{proof}

\section{Additional Real Data Analysis}
\label{sec:add:real}
In this section, we include the additional real data analysis.

\subsection{R-squared analysis}
Figure \ref{fig3cov_comp_BB_sJqwef2} displays the R-squared values by varying regularization parameters. 
It is observed that EIV always has larger maximum $R^*$ values than maximum R-squared values of Lasso and EIV. 
Considering the fact that $R^*$ is a lower bound of R-squared value of EIV,  the relations of the torque data $\tilde{X}$ and the neural signal difference $\Delta$ 
is well explained by EIV regression than regular regression.

\begin{figure}[H]
\centering
  \begin{tabular}{@{}cc@{}} 
            \includegraphics[width=5cm, height=4cm]{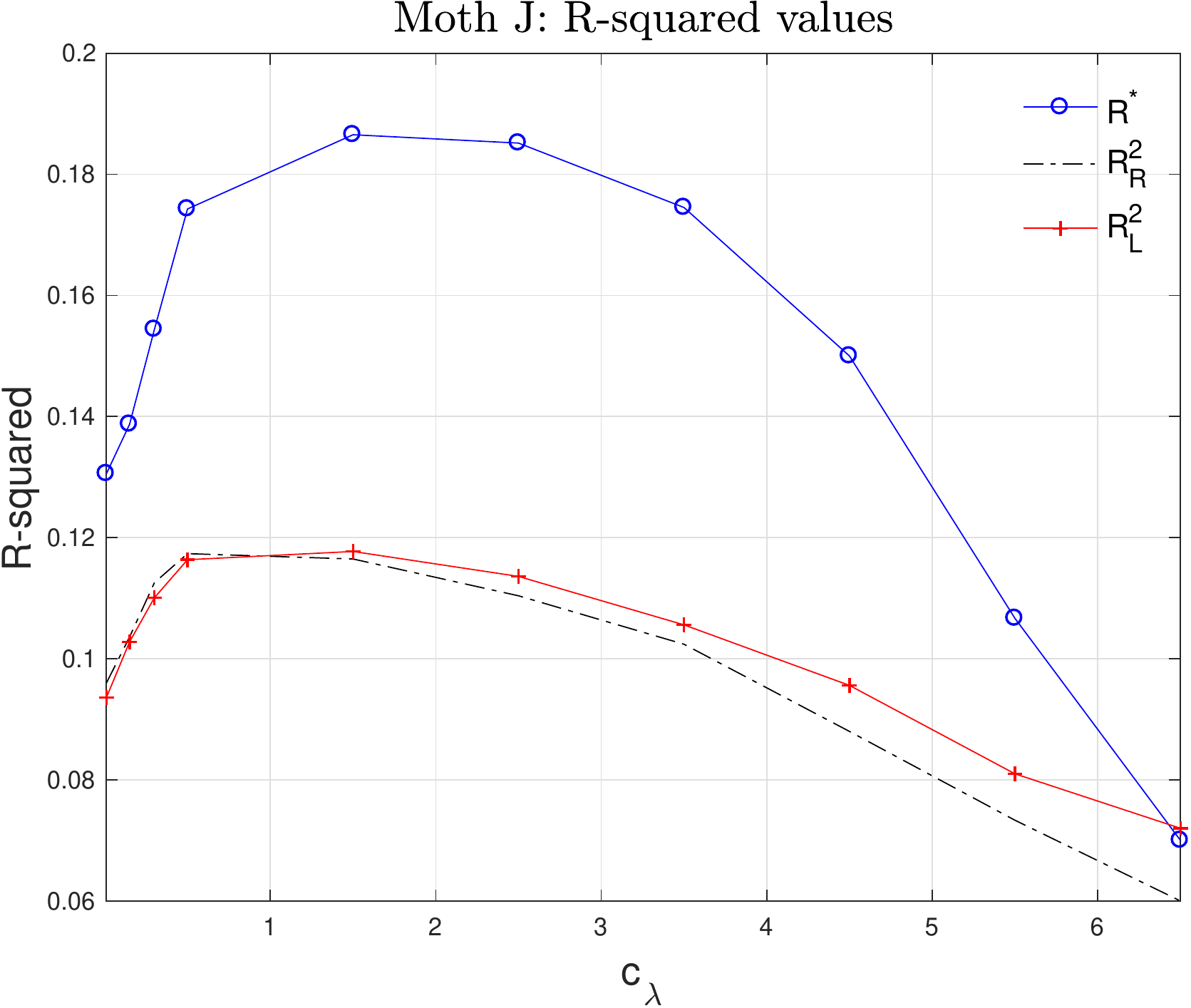}  &  \includegraphics[width=5cm, height=4cm]{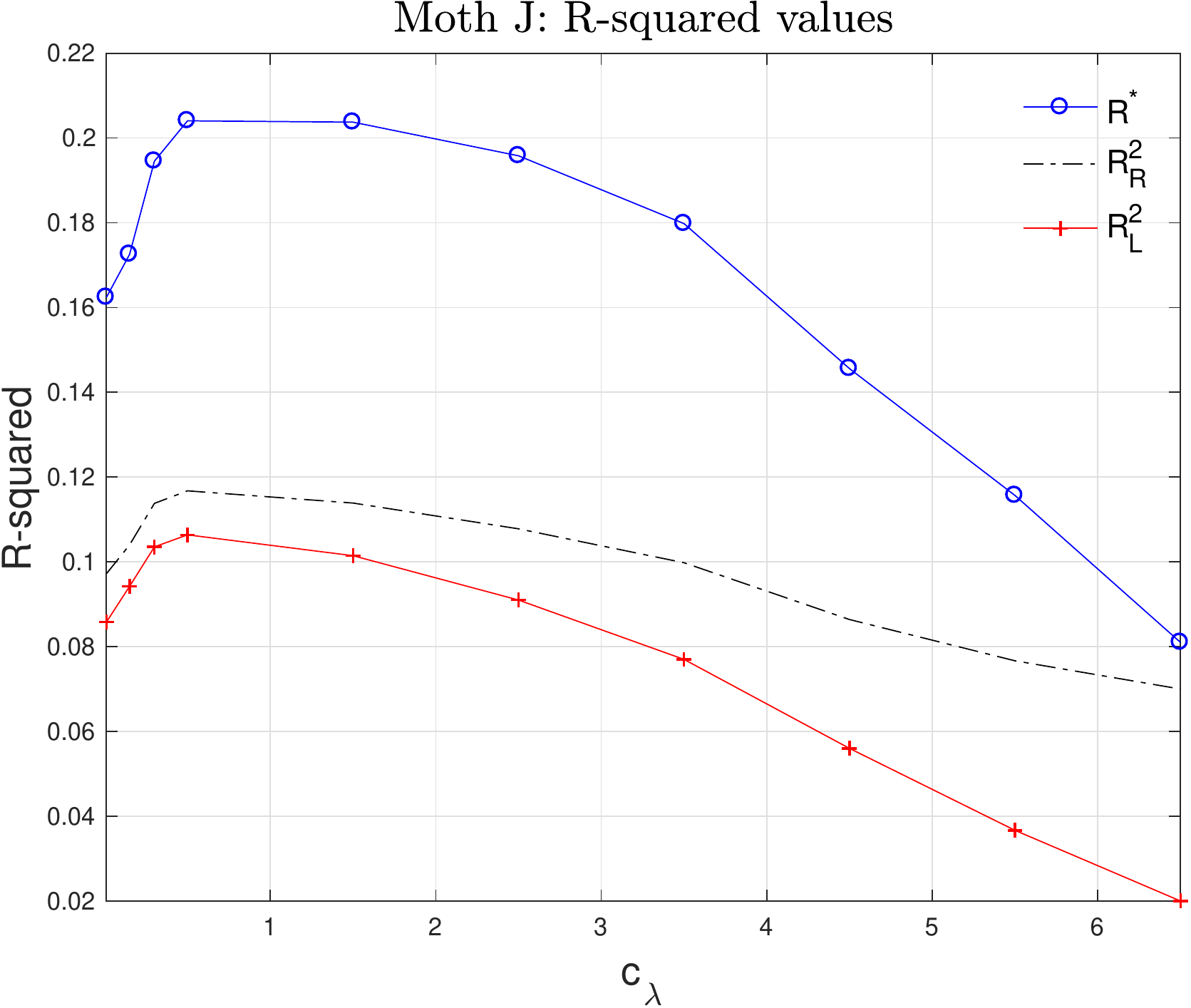}\\ 
            \includegraphics[width=5cm, height=4cm]{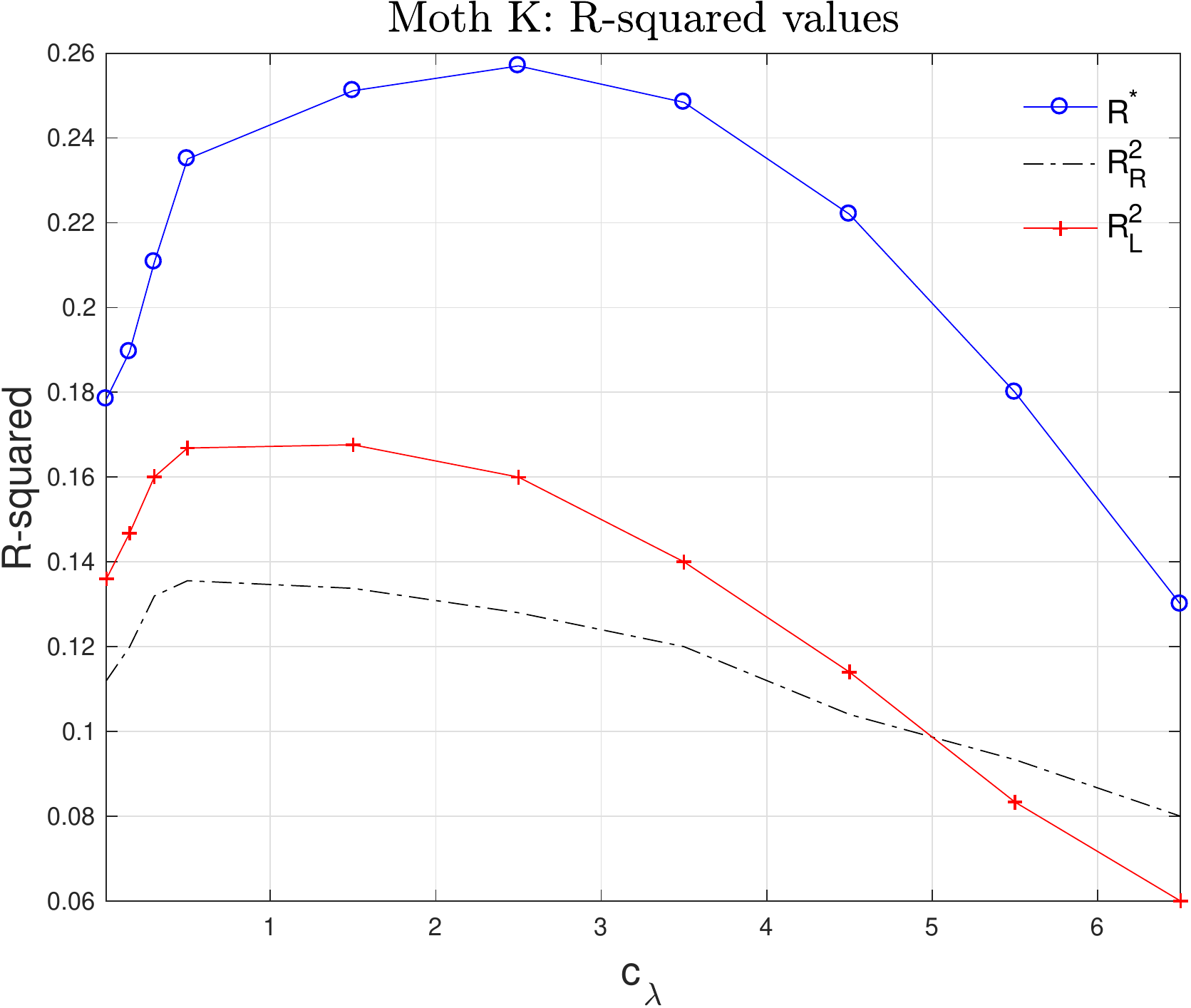}  &  \includegraphics[width=5cm, height=4cm]{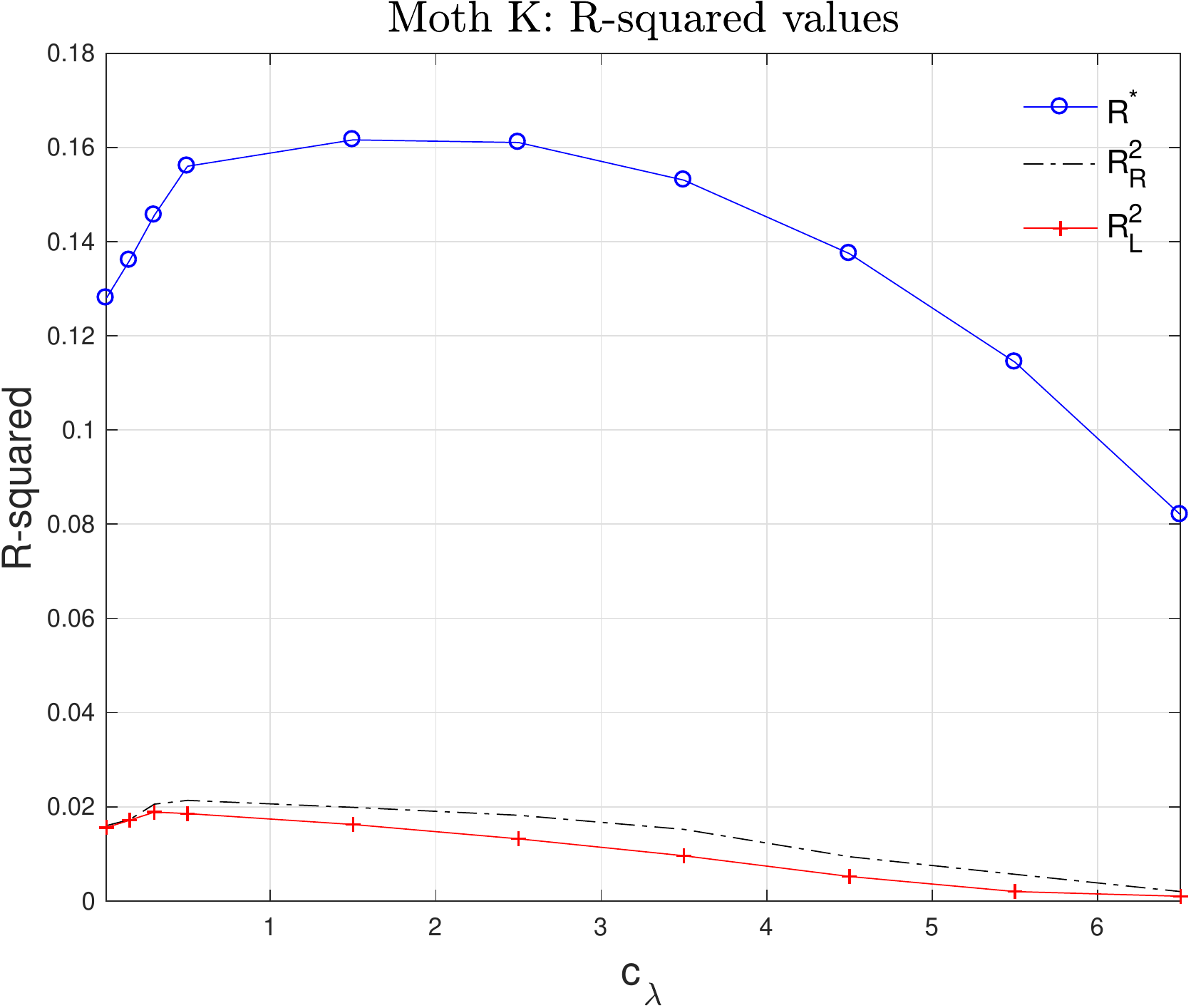}\\ 
            \includegraphics[width=5cm, height=4cm]{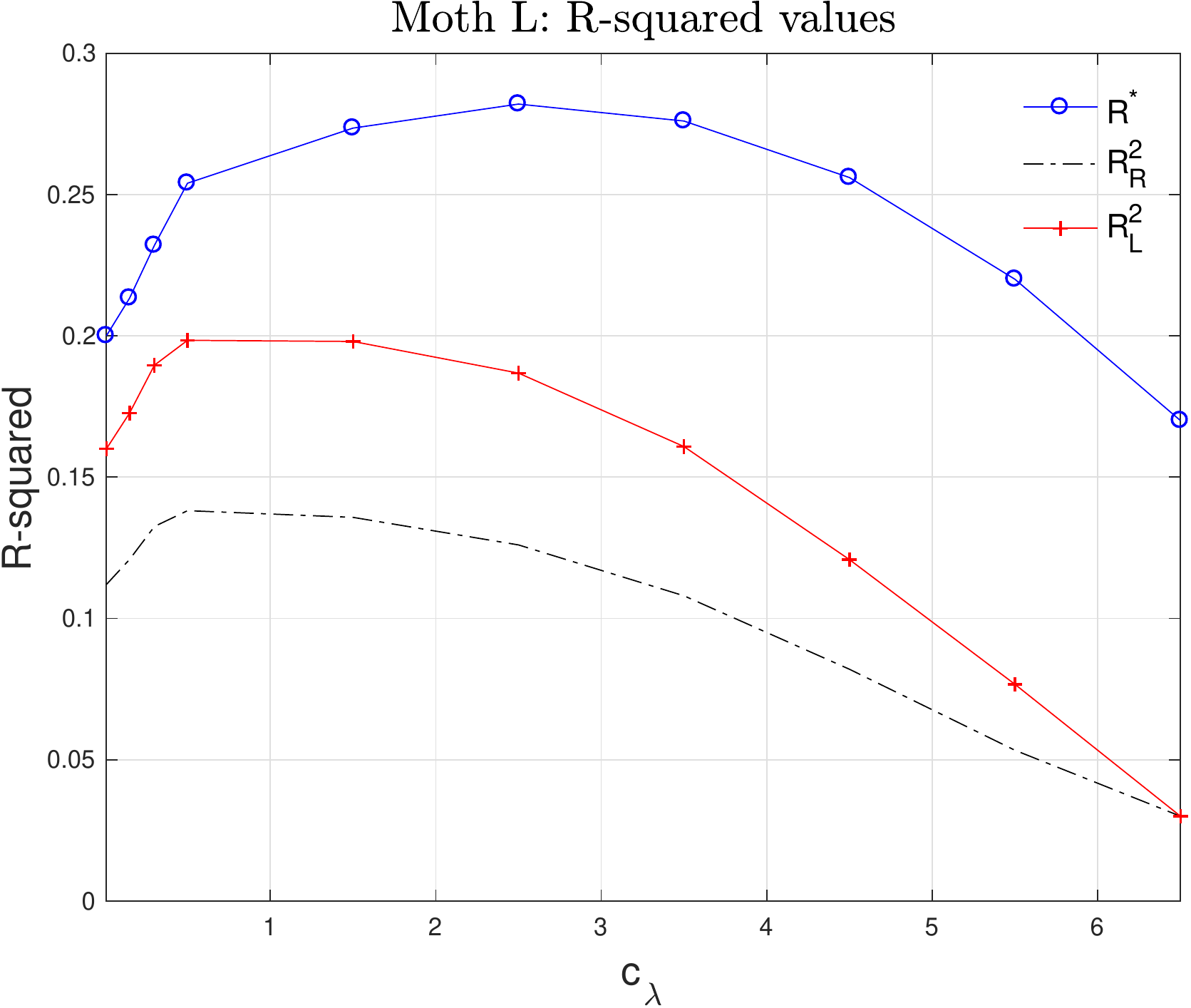}  &  \includegraphics[width=5cm, height=4cm]{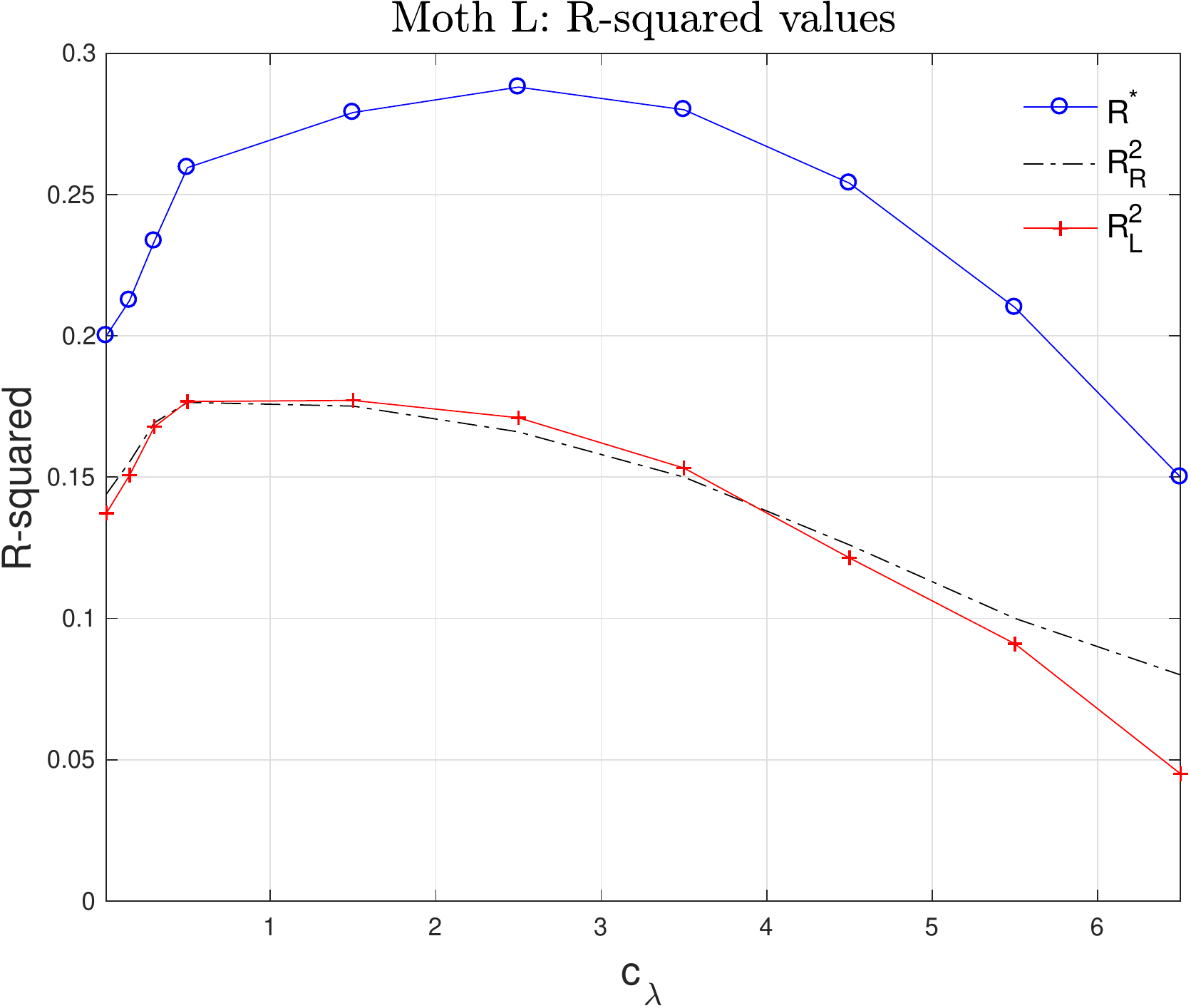}
  \end{tabular}
    \caption{R-squared values (Lasso and Ridge) and $R^*$ (EIV) values by varying regularization parameters. For the EIV,  we set optimal $\hat{\tau}_A$ as in Table \ref{table:mothpower}
    in the main paper.  
    }
 \label{fig3cov_comp_BB_sJqwef2}
\end{figure}

\subsection{Supplementary of Regression Analysis}
\label{subsect:reg:a}
Recall that in Subsection \ref{sec:moth:reg} of the main paper, we fit regression models including EIV regression relating the residualized neural spike time differences and torque values:  
\begin{eqnarray}
&&\tilde{\Delta} = \tilde{X} \Pi \zeta_1+ \epsilon, \quad \ \tilde{X}\Pi\ \rm{and} \ \tilde{\Delta}  \ \rm{are \ observable,} \label{eq:reg22}\\
&&\tilde{\Delta} = X_0 \Pi \zeta_2 + \epsilon, \quad  \tilde{X} \Pi =X_0 \Pi +W \Pi ,  \quad  \tilde{X}\Pi\ \rm{and} \ \tilde{\Delta} \ \rm{are \ observable}. \label{eq:nonsm2}
\end{eqnarray}
Suppose the model \eqref{eq:reg22}. More specifically, let
$E[\tilde{X}]=0$ and 
$\cov(\mvec{\tilde{X}}) = \Sigma \otimes I$. 
Since $\mvec{\tilde{X} \Pi} = (\Pi^T \otimes I) \mvec{\tilde{X}}$, we have
\begin{eqnarray*}
\cov(\mvec{\tilde{X}\Pi}) &=& (\Pi^T \otimes I) E[ \mathrm{vec}(\tilde{X})  \mathrm{vec}(\tilde{X})^T] (\Pi \otimes I) \\
&=& (\Pi^T \otimes I) (\Sigma \otimes I)  (\Pi \otimes I)\\
&=& \Pi^T \Sigma \Pi  \otimes I,
\end{eqnarray*}
that is, \eqref{eq:reg22} also follows the regular linear model.

Suppose the model \eqref{eq:nonsm2}. We will show that $\tilde{X} \Pi$ in \eqref{eq:nonsm2} has the Kronecker sum covariance structure, i.e., \eqref{eq:nonsm2} still follows EIV model.
Since $E[\tilde{X}]=0$ and 
$\cov(\mathrm{vec}(\tilde{X})) =A \oplus B$,   
by using the similar argument in the above with $\Pi^T \Pi = I$, we have
\begin{eqnarray*}
\cov(\mvec{\tilde{X}\Pi}) &=& (\Pi^T \otimes I) E[ \mathrm{vec}(\tilde{X})  \mathrm{vec}(\tilde{X})^T] (\Pi \otimes I) \\
&=& (\Pi^T \otimes I) (A \otimes I + I \otimes B)  (\Pi \otimes I)\\
&=& \Pi^T A \Pi  \otimes I +  \Pi^T \Pi \otimes B \\
&=&   \Pi^T A \Pi \oplus B.
\end{eqnarray*}

\section{Additional Tables and Figures}
\label{sec:add:fig}

Figures \ref{fig3covb1}-\ref{fig3covb2} show the  
graphical structure of $\hat{\Omega}$ for moth J and L for the Kronecker sum model.
Figure \ref{fig3_spat_corr} shows heat maps of the sample
correlation matrix of $500$ temporal points and the sample correlation
matrix of wing strokes, showing dependencies in both time points and among the wing strokes.
These figures are referenced in Section \ref{sec:realdata} of the main paper.

\begin{figure}[H]
\centering
  \begin{tabular}{@{}c@{}}
                           \includegraphics[width=12cm, height=15cm]{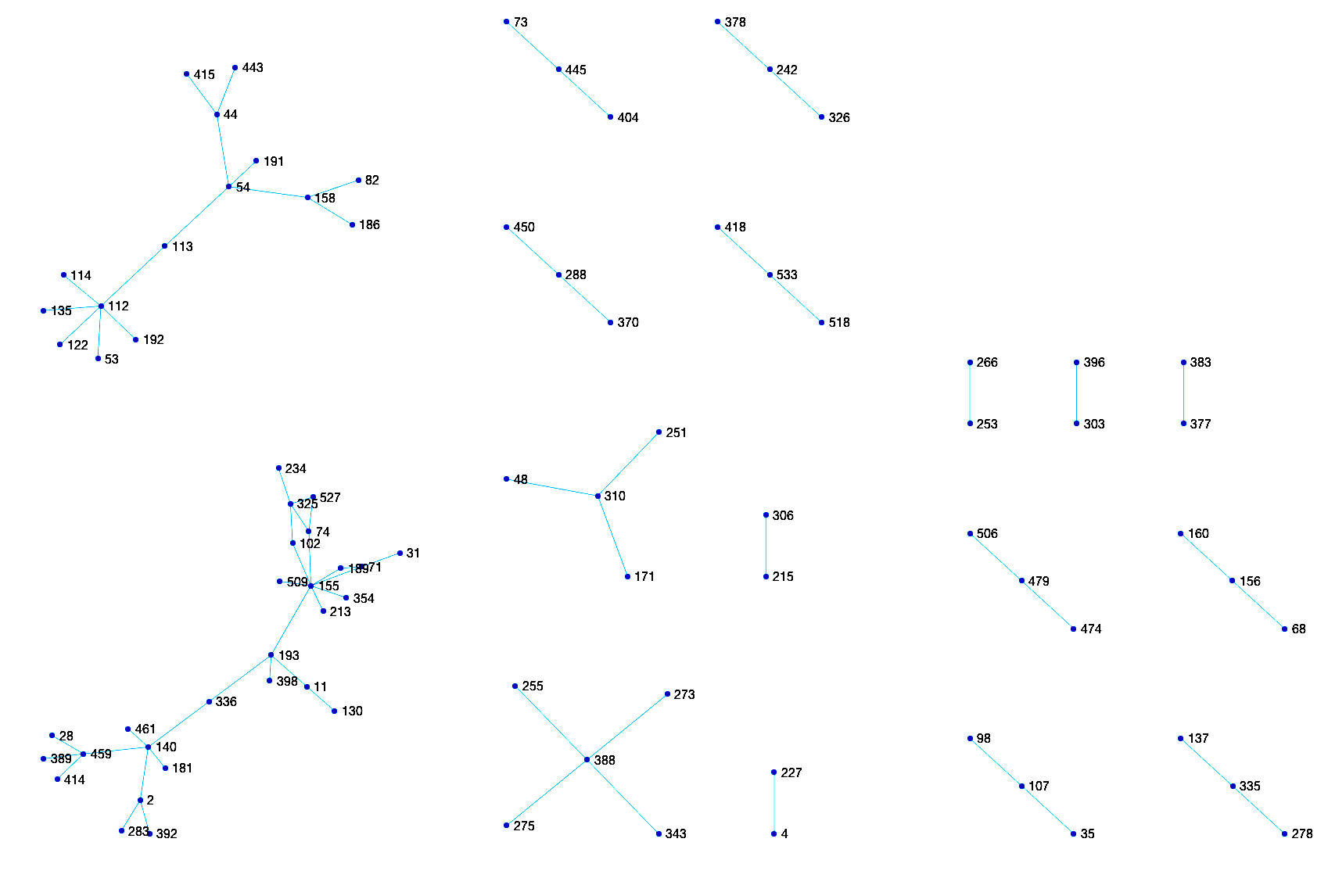}  
                                   \end{tabular}
  \caption{The estimated graphical structure of $\Omega$ for Moth J for the 
Kronecker sum model when $\hat{\tau}_A=1.5$ and 
$\lambda^{(i)}_B=0.1 \sqrt{\frac{\log n}{500}} \asymp  0.01$.
Singletons (nodes with edges) are not included in the graph.}
\label{fig3covb1}
\end{figure}

\begin{figure}[H]
\centering
  \begin{tabular}{@{}c@{}}
                  \includegraphics[width=12cm,height=15cm]{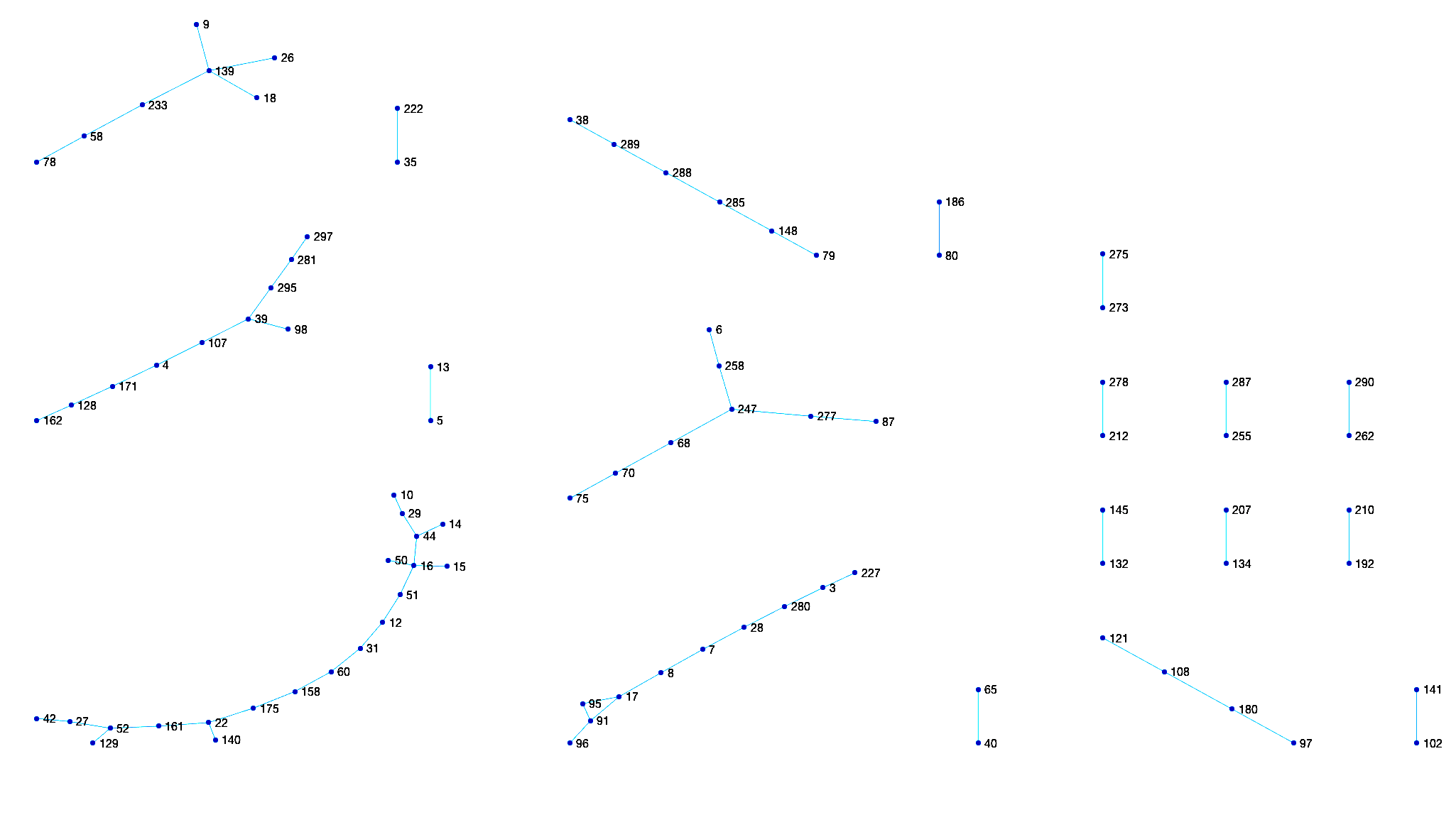}          
        \end{tabular}
  \caption{The estimated graphical structure of $\Omega$ for Moth L for the Kronecker sum model when $\hat{\tau}_A=1.5$ and $\lambda^{(i)}_B =0.1 \sqrt{\frac{\log n}{500}} \asymp 0.01$. 
Nodes that have no connected edges are not included in the graph. 
  }
\label{fig3covb2}
\end{figure}

\begin{figure}[H]
\centering
  \begin{tabular}{@{}cc@{}}
    \includegraphics[width=.5\textwidth]{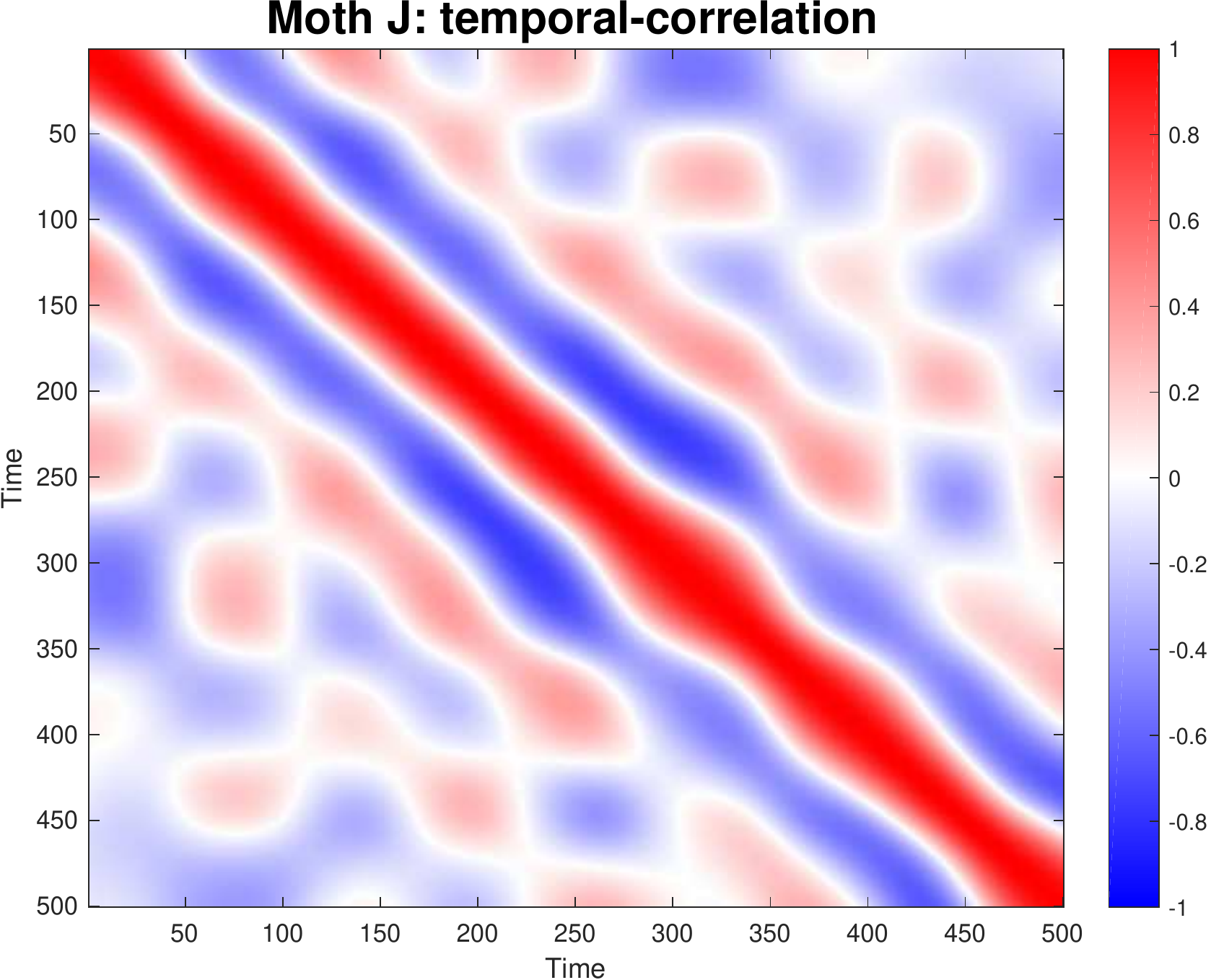}  &   \includegraphics[width=.5\textwidth]{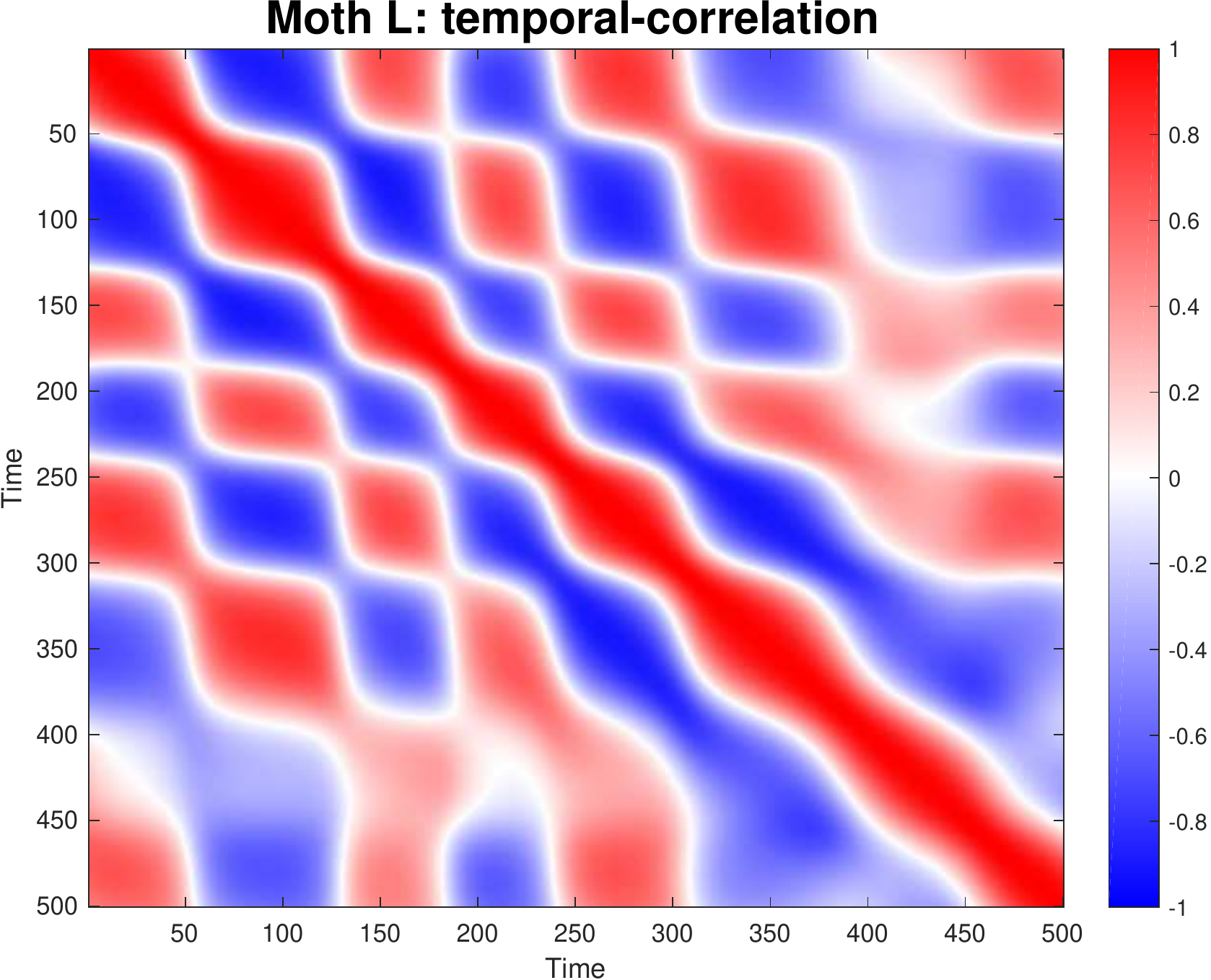}\\
  \includegraphics[width=.5\textwidth]{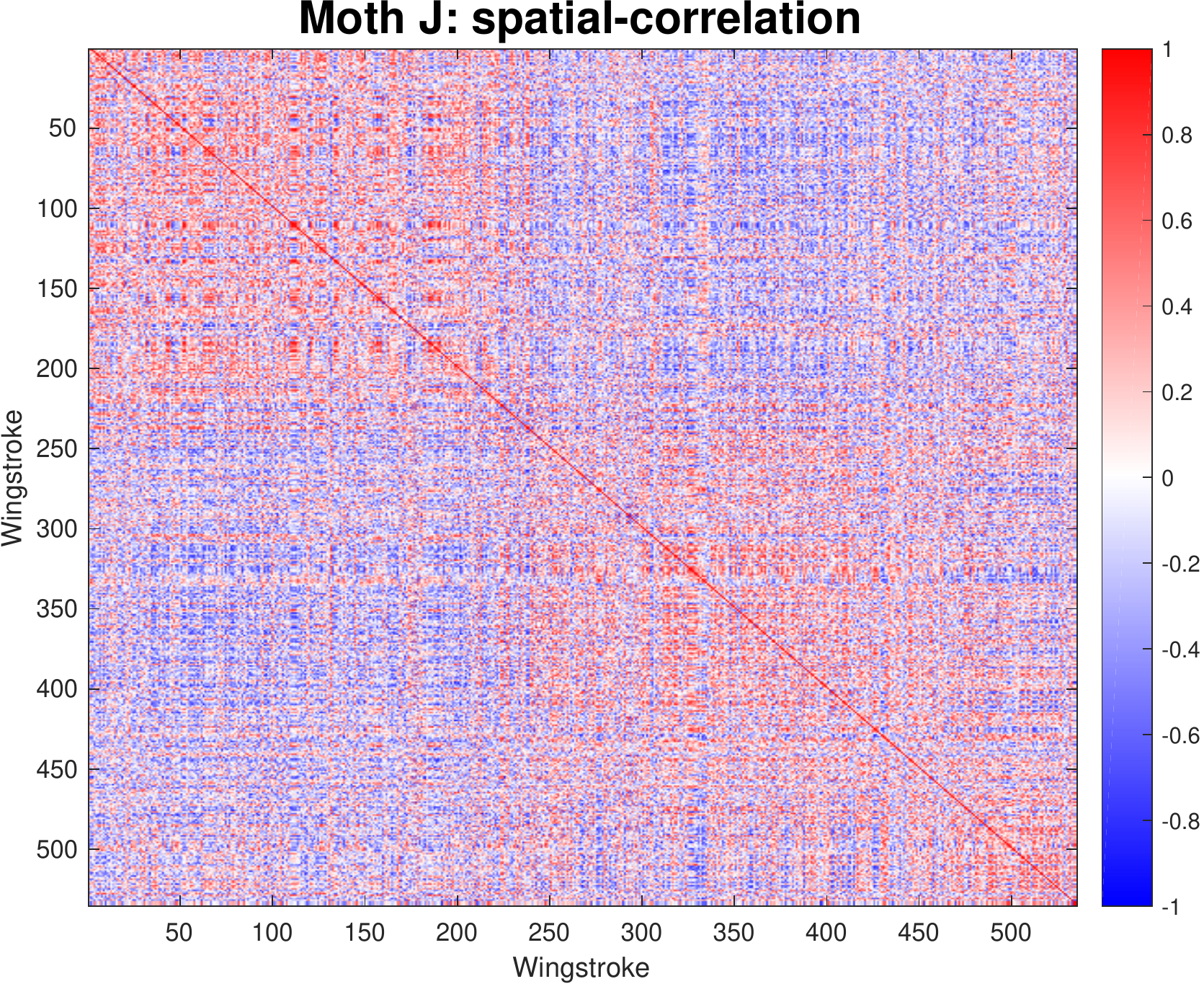}  &   \includegraphics[width=.5\textwidth]{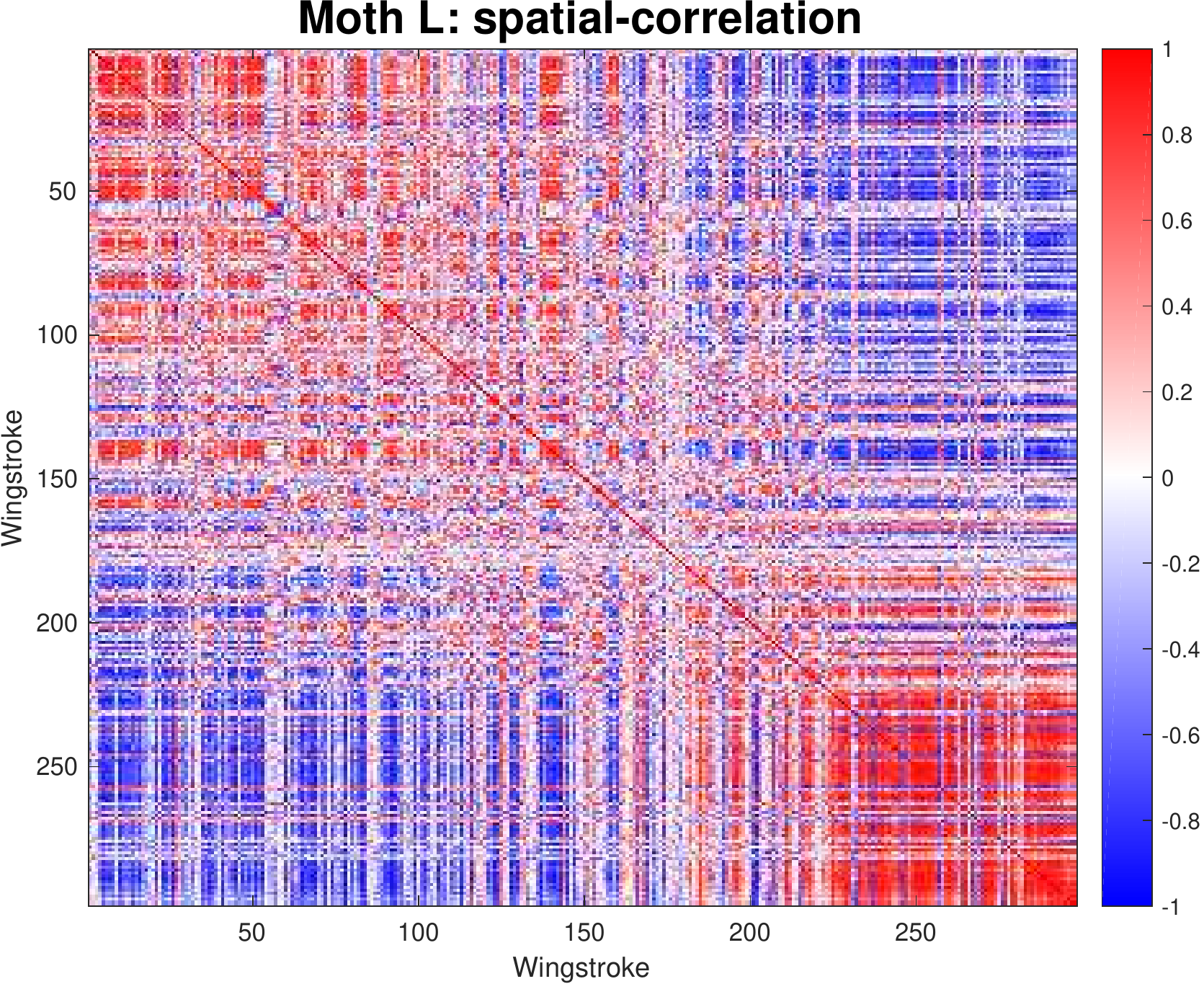}
    \end{tabular}
  \caption{\label{fig3_spat_corr}Heat maps of the sample correlation matrix of 500 temporal points (top) and the sample correlation matrix of wingstrokes (bottom) for moth J (left) and moth L (right).}

\end{figure}

\bibliography{subgaussian2}

\end{document}